\documentclass[a4paper,
               DIV=18,
               indent,
               abstract=true,
               11pt]{scrartcl}

\usepackage{graphicx} 
\usepackage{subcaption}
\usepackage{amsmath}
\usepackage{amsthm}
\usepackage{thmtools}
\usepackage{amssymb}
\usepackage{enumitem}
\usepackage{bm}
\usepackage{mathtools}
\usepackage[utf8]{inputenc}
\usepackage{graphicx}
\usepackage{longtable}
\usepackage{booktabs}
\usepackage{braket}
\usepackage[style=alphabetic]{biblatex}
\usepackage{algpseudocode}
\usepackage{authblk}
\usepackage{bbm}
\usepackage{algorithm}
\usepackage{pgfplots}
\usepackage{xcolor}
\usepackage{tikz}
\usetikzlibrary{arrows.meta,positioning,shapes.geometric,fit,backgrounds,calc}
\usepackage{hyperref}
\usepackage[capitalise,sort]{cleveref}
\pgfplotsset{compat=newest}

\newtheoremstyle{plainnonitalic}{3pt} 
{3pt} 
{\normalfont} 
{} 
{\bfseries} 
{.} 
{ } 
{} 
\theoremstyle{plainnonitalic}
\newtheorem{theorem}{Theorem}[section]
\newtheorem{corollary}[theorem]{Corollary}
\newtheorem{lemma}[theorem]{Lemma}
\newtheorem{definition}[theorem]{Definition}
\newtheorem{remark}[theorem]{Remark}

\newtheorem{example}[theorem]{Example}
\newtheorem{notation}[theorem]{Notation}
\newtheorem{problem}[theorem]{Problem}
\newtheorem{result}[theorem]{Result}
\newtheorem{fact}[theorem]{Fact}
\newtheorem{conjecture}[theorem]{Conjecture}

\crefname{fact}{Fact}{Facts}
\newcommand*{\myrefeq}[2]{#1~(\hyperref[#2]{\ref*{#2}})}

\newcommand{\tp}{\otimes}

\newcommand{\densmats}[1]{\mathcal{S}\left(#1\right)}
\newcommand{\gen}[1]{\braket{#1}}

\newcommand{\NN}{\mathbb{N}}
\newcommand{\ZZ}{\mathbb{Z}}

\newcommand{\CC}{\mathbb{C}}
\newcommand{\DD}{\mathbb{D}}
\newcommand{\FF}{\mathbb{F}}

\DeclareMathOperator*{\EE}{\mathbb{E}}

\newcommand{\indic}{\mathbbm{1}}

\DeclarePairedDelimiter{\norm}{\lVert}{\rVert}
\newcommand{\Norm}[1]{\norm*{#1}}

\numberwithin{equation}{section}

\newcommand{\ketbra}[2]{\ket{#1}\!\bra{#2}}
\newcommand{\Tr}{\operatorname{Tr}}
\DeclarePairedDelimiter{\abs}{\lvert}{\rvert}
\newcommand{\Abs}[1]{\abs*{#1}}
\newcommand{\vd}{\mathbf}

\newcommand{\Weyl}{\operatorname{Weyl}}
\newcommand{\Stab}{\operatorname{Stab}}
\newcommand{\poly}{\operatorname{poly}}
\newcommand{\polylog}{\operatorname{polylog}}
\newcommand{\core}{\operatorname{core}}
\newcommand{\ncl}{\operatorname{ncl}}
\DeclarePairedDelimiter{\ceil}{\lceil}{\rceil}
\newcommand{\Ceil}{\ceil*}

\newcommand{\QFT}{\operatorname{QFT}}
\newcommand{\Graph}{\operatorname{Graph}}
\newcommand{\lcm}{\operatorname{lcm}}
\newcommand{\ord}{\operatorname{ord}}
\newcommand{\HSS}{\operatorname{HSS}}
\newcommand{\MASD}{\operatorname{MASD}}
\newcommand{\Baer}{\operatorname{Baer}}
\newcommand{\legendre}[2]{\left(\frac{#1}{#2}\right)}
\newcommand{\ctrl}[1]{{\mathsf{Ctrl}}(#1)}

\newcommand{\prb}[1]{\textsc{#1}}
\newcommand{\cc}[1]{{\text{\sffamily #1}}}

\newcommand{\sathya}[1]{\textcolor{red}{
}}
\newcommand{\isaac}[1]{\textcolor{blue}{
}}

\setlist[description]{leftmargin=1cm, labelindent=0.5cm, nosep}

\bibliography{references.bib}

\title{Learning quantum symmetries}

\author[1]{Isaac Holt\footnote{\href{mailto:isaac.holt@cst.cam.ac.uk}{isaac.holt@cst.cam.ac.uk}}}
\affil[1]{University of Cambridge}
\author[2]{Sathyawageeswar Subramanian\footnote{\href{mailto:sathya.subramanian@cs.ox.ac.uk}{sathya.subramanian@cs.ox.ac.uk}}}
\affil[2]{University of Oxford}
\date{}

\begin{document}

\maketitle

\begin{abstract}
    Quantum algorithms are powerful tools for finding symmetries of classical objects, most famously through Shor's algorithm and the Hidden Subgroup Problem (HSP). In this work, we study quantum algorithms for learning symmetries of \emph{quantum} objects. Existing work in this area centres on the recently introduced State Hidden Subgroup Problem (StateHSP), a quantum generalisation of HSP in which the task is to learn the symmetry subgroup of a quantum state. We develop efficient quantum algorithms for non-abelian StateHSP when the hidden subgroup is normal and the ambient group belongs to a broad class of non-abelian groups, extending the previous general theory beyond the abelian setting.

    StateHSP learns \emph{Bose} symmetries, under which a state must be invariant exactly under the action of the symmetry group. In quantum mechanics, however, physically equivalent pure states are defined only up to global phase. Motivated by this, we introduce a natural notion of \emph{Anyonic} state symmetry learning, based on invariance up to global phase. We give an efficient quantum algorithm by reducing the problem to StateHSP, where the reduction rests on a new correspondence between linearisations of projective representations and linear error-correcting codes. As an application, we obtain an improved algorithm for learning stabiliser groups of mixed qudit states of arbitrary local dimension. Finally, we introduce symmetry learning problems for other quantum objects, including unitaries, Hamiltonians, and finite collections of states, and give efficient algorithms for them by reduction to state symmetry learning. Together, these results broaden the scope of state symmetry learning as a common algorithmic primitive for learning quantum symmetries.
\end{abstract}

\pagebreak
\setcounter{tocdepth}{2}
\tableofcontents

\pagebreak


\section{Introduction}

Symmetries play a fundamental role in quantum physics \cite{FR96symmetriesQuantumPhysics}. For example, symmetries of Hamiltonians correspond to conserved quantities in physical systems, swap symmetries constrain entanglement \cite{BGTW25stateHSP}, and Pauli symmetries constrain magic \cite{GNW21bellDifferenceSampling}. Identifying the symmetries of quantum objects is therefore a natural computational problem.

Quantum algorithms are already known to be powerful tools for finding symmetries of \emph{classical} objects. Shor's celebrated algorithm \cite{Sho97quantumFactoring} for factoring integers and computing discrete logarithms works by finding the symmetries of a periodic function on the integers, an observation that led to the Hidden Subgroup Problem (HSP) \cite{ME99hiddenSubgroupProblem}, a central framework in quantum algorithms \cite{Joz01HSP,Reg04latticeProblems} that asks to recover the symmetries of a periodic function over an arbitrary group.

Extending this paradigm from classical to quantum objects is substantially more challenging, since quantum data admits richer notions of symmetry. Many physically relevant symmetries of quantum objects---including pure and mixed states, unitaries, Hamiltonians, and channels---are characterised by a projective unitary representation $R$ of a group $G$, where the action of a corresponding \emph{symmetry subgroup} $S \leq G$ leaves the object invariant.

Two natural computational tasks associated with symmetry identification are \emph{property testing} and \emph{learning}, two central paradigms in theoretical computer science \cite{MdW16quantumPropertyTesting,AW17quantumLearningTheory}. In the property testing setting, we are given copies of an object and promised that the symmetry subgroup is either the full group ($S=G$), or that the object is $\epsilon$-far (with respect to an appropriate distance measure) from every $G$-invariant object, and the task is to distinguish these two cases with high probability. In the learning setting, we must instead output an efficient description of its symmetry subgroup $S$ with high probability.

While many general symmetry testing problems admit efficient quantum algorithms \cite{LRW23testingQuantumSymmetries,LW22hamiltonianSymmetries,BRRW23testingQuantumSymmetries,RLW23quantumSymmetries}, comparatively little is known about symmetry learning. Prior to this work, the only symmetry learning problem studied in generality was the State Hidden Subgroup Problem (\prb{StateHSP}) \cite{BGTW25stateHSP}, which asks to learn the \emph{Bose symmetries} of a quantum state, which are characterised by invariance under the left action of a finite group represented linearly on the state space. Beyond introducing the problem, \cite{BGTW25stateHSP} showed that identifying hidden cuts in unentangled states is an instance of \prb{StateHSP} and gave an efficient quantum algorithm for it. Subsequently, \cite{HEC25abelianStateHSP} showed that every instance over a finite abelian group admits an efficient quantum algorithm, and identified several further applications, including stabiliser group learning, finding hidden translations, and detecting symmetry-protected topological phases. These results establish \prb{StateHSP} as a natural framework for quantum symmetry learning, but leave open how far this framework extends beyond the abelian setting and beyond Bose symmetry.\\

Despite this progress, the current algorithmic understanding of symmetry learning remains substantially narrower than that of symmetry testing. Efficient testing algorithms are known for multiple notions of state symmetry, including the strictly weaker notion of \emph{conjugation symmetry}, which has a more natural physical interpretation for mixed states, in addition to Bose symmetry \cite{LRW23testingQuantumSymmetries,BRRW23testingQuantumSymmetries,RLW23quantumSymmetries}, and for a broad range of quantum objects, including Hamiltonians, channels, measurements, and Lindbladians \cite{LW22hamiltonianSymmetries,LRW23testingQuantumSymmetries,BRRW23testingQuantumSymmetries}. Moreover, these algorithms place essentially no structural restrictions on the symmetry, requiring only that the underlying group be finite. By contrast, existing symmetry learning algorithms apply only to Bose symmetries of quantum states, finite \emph{abelian} ambient groups, and linear rather than projective representations.

The \prb{StateHSP} was introduced as a quantum generalisation of the more classical \prb{HSP}. Consequently, understanding its complexity over non-abelian groups is a natural objective. Though it admits an efficient quantum algorithm over finite abelian groups, \prb{HSP} is believed to be hard over general non-abelian groups \cite{MRS08symmetricGroupHSP}, even for quantum computers. Since \prb{HSP} reduces to \prb{StateHSP}, one expects \prb{StateHSP} over non-abelian groups to exhibit similar worst-case hardness. 

This worst-case hardness, however, does not preclude positive algorithmic results. Although efficient algorithms for general non-abelian \prb{HSP} remain elusive, many important families of non-abelian groups do admit efficient quantum algorithms \cite{RB98HSP,Gav04HSP,GP11HSP,ISS12HSP,II23HSP}, and the quantum query complexity of every instance is in fact polynomial \cite{EHK04queryComplexityHSP}. In contrast, the non-abelian \prb{StateHSP} has remained almost entirely unexplored.

\section{Our contributions}
In this work, we substantially broaden the scope of quantum symmetry learning in three complementary directions. First, we extend the \prb{StateHSP} framework to broad classes of non-abelian groups. Second, we introduce new symmetry learning problems based on more physically natural notions of quantum state symmetry. Third, we extend symmetry learning beyond quantum states to other classes of quantum objects, including unitaries, Hamiltonians, and collections of states. Together, these results substantially enlarge the range of quantum symmetries known to admit efficient quantum learning algorithms.

\subsection{StateHSP over non-abelian groups}

Our first goal is to extend the \prb{StateHSP} framework beyond the abelian setting. Our starting point is a quantum algorithm that identifies the normal core of the Bose-symmetry subgroup of a mixed state with respect to an arbitrary linear representation of a finite group. This algorithm is the \prb{StateHSP} analogue of the algorithm of Hallgren, Russell, and Ta-Shma~\cite{HRT03HSP} for recovering the normal core of a classical \prb{HSP} instance, and provides the first general algorithmic result for non-abelian \prb{StateHSP}.


\begin{result}[Informal --- see \cref{thm:efficient-algorithm-for-normal-core-statehsp}]\label{res:normal-core-statehsp}
    There is an efficient quantum algorithm for finding the normal core of the Bose-symmetry subgroup of a mixed state with respect to a linear representation of a finite group.
\end{result}

This algorithm generalises both the normal-core algorithm for the classical \prb{HSP} \cite{HRT03HSP} and the existing algorithm for abelian \prb{StateHSP} \cite{HEC25abelianStateHSP}. In both settings, it matches or improves the known sample and time complexities. Since \prb{HSP} reduces to \prb{StateHSP}, our algorithm also yields an improved normal-core algorithm for the classical \prb{HSP}.

As a consequence of our normal core algorithm, we obtain an efficient quantum algorithm for Bose Symmetry Learning over ``polynomially near-Hamiltonian'' groups (defined formally in \cref{sec:bose-symmetry-learning-over-poly-near-hamiltonian-groups}). This gives the largest class of non-abelian groups currently known to admit an efficient \prb{StateHSP} algorithm.

\begin{result}[Informal --- see \cref{thm:bose-symmetry-learning-over-poly-near-hamiltonian-groups}]\label{crl:informal-bose-symmetry-learning-over-poly-near-hamiltonian-groups}
    There is an efficient quantum algorithm for finding the Bose-symmetry subgroup of a mixed state with respect to a linear representation of a finite group which is ``polynomially near-Hamiltonian''.
\end{result}

Our improved sample complexity for the normal-core algorithm also improves the best known sample complexity for the \emph{hidden translation problem}, which asks to find the hidden translation operator under which a state in $\CC^N$ is invariant.. Whereas \cite{HEC25abelianStateHSP} requires $O(\log N)$ copies, our algorithm requires $O(\log N / \log\log N)$ copies in the worst case, and $O(\log\log N)$ copies on average.

\subsection{A generalisation of Bose symmetries}

The \prb{StateHSP} framework learns \emph{Bose symmetries}, in which a state must be invariant under the action of the symmetry group exactly so that the hidden subgroup satisfies $S = \{ g \in G: \Tr(R(g) \rho) = 1 \}$. From the perspective of quantum mechanics, however, this notion of symmetry is unnecessarily restrictive, since pure states that differ only by a global phase are physically indistinguishable. Motivated by this observation, we introduce a more general symmetry learning problem based on what we call \emph{anyonic symmetries}, in which invariance is required only up to a global phase.

Formally, \emph{Anyonic Symmetry Learning} asks to find the subgroup of anyonic symmetries
\[
S=\{g\in G: \abs{\Tr(R(g)\rho)}=1\},
\]
where $G$ is a group, $R$ is a representation, and $\rho$ is a quantum state. Equivalently, this is the subgroup that leaves the state invariant up to a global phase. A primary motivation for this problem is stabiliser group learning, where the natural symmetry condition is $\abs{\Tr(Z^aX^b\rho)}=1$ for the phaseless $n$-qudit Pauli operators $Z^aX^b$, rather than $\Tr(Z^aX^b\rho)=1$. 

We will hereinafter also refer to these two notions of quantum symmetry learning as \prb{BoseSL} and \prb{AnyonicSL} respectively. Note, in particular, that we will use \prb{StateHSP} and \prb{BoseSL} interchangeably.

We first show that Anyonic Symmetry Learning admits an efficient quantum algorithm when the symmetry is described by a linear representation and the underlying group is ``polynomially near-abelian'' (defined formally in \cref{sec:non-abelian-anyonic-symmetry-learning}).

\begin{result}[Informal --- see \cref{thm:algorithm-for-anyonic-symmetry-learning}]
    There is an efficient quantum algorithm for finding the anyonic-symmetry subgroup of a mixed state with respect to an arbitrary linear representation of a finite group which is ``polynomially near-abelian''.
\end{result}

This algorithm does not yet solve our motivating stabiliser learning problem, because the phaseless Pauli operators form a \emph{projective}, rather than linear, representation $(a, b) \mapsto Z^a X^b$ of $\ZZ_d^n \times \ZZ_d^n$. More fundamentally, projective representations describe the most general algebraic symmetries in quantum mechanics, since they capture symmetries that hold only up to a global phase by allowing for a ``twisting'' (a multiplication by a scalar) in the homomorphism condition. They are therefore the natural setting for symmetry learning.   

Our main result in this direction is that projective representations can also be handled efficiently. We show that Anyonic Symmetry Learning admits an efficient quantum algorithm for arbitrary projective representations of finite abelian groups.

\begin{result}[Informal --- see \cref{thm:algorithm-for-projective-state-hsp}]
    There is an efficient quantum algorithm for finding the anyonic-symmetry subgroup of a mixed state with respect to an arbitrary projective representation of an arbitrary finite abelian group.
\end{result}

Since the phaseless Pauli operators form a projective representation of $\mathbb{Z}_d^{2n}$, the above theorem immediately yields the first explicit quantum algorithm for learning the stabiliser group of an arbitrary mixed state of any local qudit dimension.

\begin{corollary}[Informal --- see \cref{thm:stabiliser-group-learning-algorithm}]
    There is an $O(n)$-copy, $\poly(n)$-time quantum algorithm for finding the stabiliser group of an arbitrary $n$-qudit mixed state.
\end{corollary}

Our algorithm either matches or improves (depending on the local qudit dimension $d$) the sample and time complexities of the existing stabiliser group learning algorithms for pure states \cite{HEC25abelianStateHSP,ADIS25quditBellSampling}, while additionally providing the first provably correct algorithm for arbitrary mixed states.

\subsection{Learning symmetries of other quantum objects}

Finally, we extend symmetry learning beyond quantum states to other classes of quantum objects. Our approach is to show that several natural symmetry learning problems reduce to state symmetry learning, allowing the algorithms developed in this framework to be applied in substantially greater generality.

We first introduce symmetry learning problems for unitaries and Hamiltonians, where the goal is to learn the subgroup that leaves the object invariant under conjugation. In particular, given a group $G$ and a representation $R$, for a unitary $U$ the \emph{Unitary Symmetry Learning} problem is to find the subgroup of all conjugate symmetries $S = \{ g \in G: R(g) U R(g)^\dagger = U \}$. \emph{Hamiltonian Symmetry Learning} is defined analogously, replacing $U$ by a Hamiltonian $H$.

Using ideas from \cite{LW22hamiltonianSymmetries}, we show that both problems reduce naturally to Bose Symmetry Learning. This is a powerful reduction, since it means any quantum algorithm for learning Bose-symmetries of states can also be used for learning symmetries of unitaries and Hamiltonians.

\begin{result}[Informal --- see \cref{thm:unitary-symmetry-learning-reduction-to-bose-symmetry-learning}]
    Learning the symmetries of a unitary is reducible to learning the Bose symmetries of its Choi state.
\end{result}
\begin{result}[Informal --- see \cref{thm:hamiltonian-symmetry-learning-reduction-to-unitary-symmetry-learning}]
    Learning the symmetries of a Hamiltonian is reducible to learning the symmetries of its time evolution unitary, and so is reducible to Bose symmetry learning.
\end{result}

We also introduce symmetry learning for collections of pure states, where symmetry is required only at the level of the collection rather than ``pointwise'' for each individual state. This leads naturally to the problems of Subset Symmetry Learning and Subspace Symmetry Learning. Given a group $G$ and a representation $R$, we have the following definitions.

\noindent\emph{Subset Symmetry Learning}: Given a finite set $A \subseteq \mathcal{H}$ of states, find the subgroup of all symmetries $S = \{ g \in G: R(g) A = A \}$.

\noindent\emph{Subspace Symmetry Learning}: Given a subspace $V \leq \mathcal{H}$ of states, find the subgroup of all symmetries $S = \{ g \in G: R(g) V = V \}$.

Note in particular that Subspace Symmetry Learning is a natural generalisation of Anyonic Symmetry Learning, since Anyonic Symmetry Learning of a state $\ket{\psi}$ is simply Subspace Symmetry Learning where the subspace $V = \{ e^{i \theta} \ket{\psi}: \theta \in [0, 2\pi) \}$ is one-dimensional.

Our next main result is a reduction from Subset Symmetry Learning over a group $G$ to Bose Symmetry Learning over the direct product of $G$ with the symmetric group on $|A|$ elements. Consequently, whenever $|A|$ is sufficiently small and $G$ is polynomially near-abelian, the direct product is polynomially near-Hamiltonian and the problem admits an efficient quantum algorithm via \cref{crl:informal-bose-symmetry-learning-over-poly-near-hamiltonian-groups}.

\begin{result}
    If $\abs{A} = O(\log \log \abs{G} / \log \log \log \abs{G})$ and $G$ is ``polynomially near-abelian'', then there is an efficient quantum algorithm for solving Subset Symmetry Learning over $G$.
\end{result}

We also establish a connection between Subspace Symmetry Learning and the problem of \emph{Conjugate Symmetry Learning} of mixed states (where the target symmetry subgroup is $\{ g \in G: R(g) \rho R(g)^\dagger = \rho \}$).\\

Together, these reductions show that state symmetry learning serves as a general algorithmic primitive: once efficient algorithms for learning state symmetries are available, they immediately yield efficient algorithms for learning symmetries of many other quantum objects.

\section{Technical overview}
The algorithms in this paper rely on three main technical ideas. For non-abelian StateHSP, the key ingredient is a generalisation of the normal-core algorithm for the classical Hidden Subgroup Problem. For Anyonic symmetry learning, we introduce a new Fourier difference sampling subroutine, which reduces anyonic symmetry learning to Bose symmetry learning. Finally, for projective representations, we develop a new notion of linearisation and show that such linearisations are intimately connected with linear error-correcting codes. The remainder of this section gives an overview of these ideas.

\subsection{Strong and weak Fourier sampling and a normal core algorithm}

Fourier sampling is the fundamental algorithmic primitive underlying the efficient quantum algorithm for abelian Bose Symmetry Learning developed in \cite{HEC25abelianStateHSP}. Our first technical contribution is to extend this framework to non-abelian groups. Recall that Fourier sampling with respect to a representation $R$ of an abelian group $G$ is the procedure that applies the quantum Fourier transform over $G$ to the uniform superposition over $g\in G$ of $\ket{g} \tp R(g) \ket{\psi}$, and then measures the $G$-register in the computational basis. In \cref{sec:normal-core-state-hsp}, we study the non-abelian analogues of Fourier sampling—\emph{strong} (measuring the full irrep outcome) and \emph{weak} (measuring only the irrep label)—using techniques from non-abelian Fourier analysis and the representation theory of finite groups. We generalise the key Fourier-analytic properties established in \cite{HEC25abelianStateHSP}, providing the foundation for our algorithms for non-abelian \prb{StateHSP}.

Weak Fourier sampling is the specific primitive that yields our normal core algorithm (\cref{res:normal-core-statehsp}) for non-abelian \prb{StateHSP}. The key observation is that the output subgroup $A$ always satisfies $C \subseteq A$, where $C$ is the normal core of the hidden subgroup. Consequently, the algorithm succeeds precisely when $A$ contains none of minimal subgroups properly containing $C$. This yields a sample complexity depending on $\log |\mathcal{M}|$ rather than $\log |G|$, where $\mathcal{M}$ is the set of minimal subgroups properly containing $C$. Since $|\mathcal{M}| \le |G|$, and may even be exponentially smaller (i.e.\ $\approx\log|G|)$, this strictly improves the previous analyses in many cases.

\subsection{Anyonic Symmetry Learning and Fourier difference sampling}
When the representation $R$ is linear, we solve abelian Anyonic Symmetry Learning by reducing it to abelian Bose Symmetry Learning. The key observation is that if a state $\rho$ possesses anyonic symmetry under $R(S)$, then $\rho^{\otimes 2}$ possesses Bose symmetry under $(R \otimes R^{-1})(S)$. Since efficient algorithms for abelian Bose Symmetry Learning are based on Fourier sampling \cite{HEC25abelianStateHSP}, this reduction naturally leads to a new quantum algorithmic primitive, which we call \emph{Fourier difference sampling}. Fourier difference sampling simply performs Fourier sampling twice and subtracts the two outcomes, analogous to the well-known Bell difference sampling routine \cite{GNW21bellDifferenceSampling,Mon17}. This reduction allows Anyonic Symmetry Learning to inherit the efficiency of the existing Bose Symmetry Learning algorithm while using only a constant-factor increase in the required number of copies of the state.

This tensor-product reduction relies crucially on commutativity, since when $G$ is abelian, $R \otimes R^{-1}$ is again a linear representation. However, for non-abelian groups, $R \tp R^{-1}$ is generally not a linear representation, so a different reduction is required. Our solution is to enlarge the group and consider the linear representation
\[
R'(g,k)=\omega^{-k}R(g)
\]
of $G\times\mathbb{Z}_E$, where $E$ is the exponent of $G$, and $\omega$ is an $E$-th root of unity. The additional phase factor converts anyonic symmetries into Bose symmetries, since $\rho$ is anyonic-symmetry under $R(g)$ if and only if $R(g) \rho = \omega^\ell \rho$ (i.e. $R'(g, \ell) \rho = \rho$) for some $\ell \in \ZZ_E$.

\subsection{Projective Anyonic Symmetry Learning and linearisation of projective representations}

Our main technical contribution is an efficient algorithm for Projective Anyonic Symmetry Learning, wherein $R$ is a \emph{projective}, rather than linear, representation. The central obstacle is that the standard approach to linearising a projective representation is incompatible with our algorithmic framework. Given a projective representation $R:G\to\mathcal U(\mathcal H)$, one typically lifts $R$ to a linear representation of a central extension $G'$ of $G$ which also acts on $\mathcal{H}$. However, even when $G$ is abelian, the extension $G'$ is generally non-abelian, and therefore lies outside the range of groups for which efficient Bose Symmetry Learning algorithms are known. In \cref{sec:projective-state-hsp}, we overcome this obstacle by introducing a new notion of a \emph{linearisation} of a projective representation of a finite abelian group.

Our goal is therefore to construct a linear representation of an abelian group related to $G$ directly from $R$, without passing through a non-abelian central extension. The basic idea is illustrated by the following example. Say $R$ satisfies $R(g) R(h) = \omega^{B(g, h)} R(g + h)$ for a bilinear form $B(g, h)$. This means that $R(a \cdot g) R(a \cdot h) = \omega^{B(a \cdot g, a \cdot h)} R(a \cdot g + a \cdot h) = \omega^{a^2 B(g, h)} R(a \cdot (g + h))$ for any $a \in \ZZ$; in particular, $x \mapsto R(a \cdot x)$ is also a projective representation, and so is
\begin{equation*}
    R'(g) = R(g) \tp R(a \cdot g) \tp R(b \cdot g) \tp R(c \cdot g) \tp R(d \cdot g) \label{eq:linearisation-example}
\end{equation*}
for all $a, b, c, d \in \ZZ$. $R'$ has cocycle $\omega^{(1 + a^2 + b^2 + c^2 + d^2) B(g, h)}$. Let $E$ be the multiplicative order of $\omega$, We want to find $a, b, c, d \in \ZZ$ such that $a^2 + b^2 + c^2 + d^2 = E - 1$. This will mean the cocycle of $R'$ is trivial, so in fact $R'$ is a linear representation. We call $R'$ a \emph{linearisation} of $R$. Importantly, \emph{Lagrange's four-square theorem} guarantees the existence of such $a, b, c, d$, and moreover, there is an efficient algorithm for finding them \cite{PT18lagrangeFourSquares}.

The construction above is only the simplest member of a much larger family of linearisations. More generally, every linearisation corresponds naturally to a \emph{linear error-correcting code} $C$ over $\ZZ_E$, where $E$ is the exponent of $G$. Each linearisation in the family is a linear representation of the abelian group $G^s$ acting on $\mathcal{H}^{\tp t}$, where $s, t \in \NN$ vary between different representations in the family. The dimension of the code corresponds to the number $s$ of copies of $G$ which the linearisation represents, and the block length of the code corresponds to the number $t$ of copies of $\mathcal{H}$ which the linearisation acts on. In order to induce a valid linearisation, the code $C$ must satisfy two properties: it must be a free submodule of $\ZZ_E^t$, and it must be self-orthogonal (all codewords are pairwise orthogonal).

This correspondence admits a particularly natural interpretation. A linearisation can be viewed as encoding a message $\vd{g} = (g_1, \dots, g_s) \in G^s$ into a codeword $\vd{c} = (c_1, \dots, c_t) \in G^t$, after which the projective representation is applied componentwise. Here each $c_j = M_{1 j} \cdot g_1 + \cdots + M_{s j} \cdot g_s$ and $M$ is a generator matrix for $C$, and the linearisation is then given by
\begin{equation*}
    (g_1, \dots, g_s) \mapsto R(c_1) \tp \cdots \tp R(c_t).
\end{equation*}
\cref{fig:linearisation-diagram} depicts a schematic of this construction. In the case of the first example of linearisation in \cref{eq:linearisation-example}, the corresponding code has generator matrix $M = \begin{pmatrix}
1 & a & b & c & d
\end{pmatrix}$.

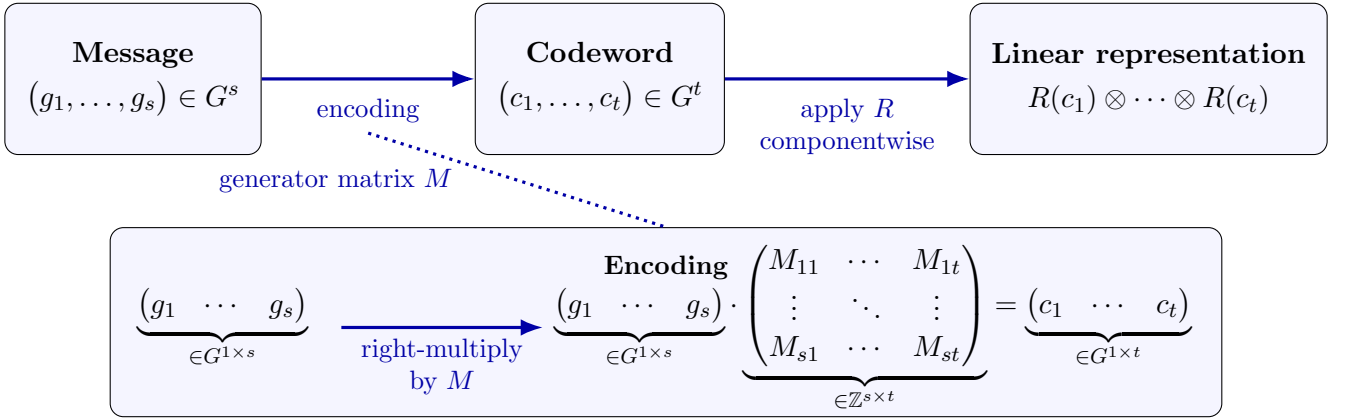
\begin{figure}[htbp]
  \centering
  \begin{tikzpicture}[
    >=Latex,
    stage/.style={
      draw=black,
      fill=blue!4,
      rounded corners=5pt,
      minimum width=3cm,
      minimum height=2.0cm,
      align=center,
      inner sep=8pt
    },
    detailbox/.style={
      draw=none,
      fill=none,
      minimum width=3.0cm,
      minimum height=1.1cm,
      align=center,
      inner sep=4pt
    },
    arrow/.style={
      ->,
      very thick,
      color=blue!65!black
    },
    arrowlabel/.style={
      midway,
      below=3pt,
      align=center,
      font=\small
    }
  ]

    \node[stage, anchor=west] (message) at (0, 0) {
      \textbf{Message}\\[3pt]
      $\bigl(g_1,\ldots,g_s\bigr)\in G^s$
    };

    \node[stage, anchor=center] (codeword)
  at (.45\linewidth,0) {
      \textbf{Codeword}\\[3pt]
      $\bigl(c_1,\ldots,c_t\bigr)\in G^t$
    };

    \node[stage, anchor=east]  (representation) at (\linewidth,0) {
      \textbf{Linear representation}\\[3pt]
      $\displaystyle R(c_1)\otimes\cdots\otimes R(c_t)$
    };

    \draw[arrow] (message) --
      node[arrowlabel] (encodinglabel) {encoding}
      (codeword);

    \draw[arrow] (codeword) --
      node[arrowlabel] {apply $R$ \\     componentwise  }
      (representation);

    \coordinate (detailbase) at ([yshift=-2.8cm]encodinglabel.south);


    \coordinate (detailcenter) at ([yshift=-2.5cm]encodinglabel.south);
\coordinate (detailleft) at ([xshift=1.4cm]message.west |- detailcenter);
\coordinate (detailright) at ([xshift=-1.4cm]representation.east |- detailcenter);

\begin{scope}[on background layer]
  \node[
    draw=black!65!black,
    fill=blue!4,
    rounded corners=5pt,
    inner xsep=0pt,
    inner ysep=12pt,
    minimum height=2.5cm,
    fit=(detailleft) (detailright)
  ] (encodingdetail) {};
\end{scope}

    \node[detailbox,yshift=-2pt] (detailinput)
    at ([xshift=-5.8cm]encodingdetail.center)
      {
      $\underbrace{\begin{pmatrix}
          g_1 & \cdots & g_s
      \end{pmatrix}}_{\in G^{1 \times s}}$
      };

    \node[detailbox,yshift=-2pt] (detailoutput)
      at ([xshift=2.8cm]encodingdetail.center)
      {
      $\displaystyle
        \underbrace{\begin{pmatrix}
            g_1 & \cdots & g_s
        \end{pmatrix}}_{\in G^{1 \times s}} \cdot
        \underbrace{\begin{pmatrix}
            M_{1 1} & \cdots & M_{1 t} \\
            \vdots & \ddots & \vdots \\
            M_{s 1} & \cdots & M_{s t}
        \end{pmatrix}}_{\in \ZZ^{s \times t}}
        = \underbrace{\begin{pmatrix}
            c_1 & \cdots & c_t
        \end{pmatrix}}_{\in G^{1 \times t}}
        $
      };

    \draw[arrow] (detailinput) --
  node[below, font=\small, align=center]
  {right-multiply \\ by $M$}
  (detailoutput);
    \draw[very thick,dotted,color=blue!65!black]
    (encodinglabel.south)
    -- node[left,font=\small]{generator matrix $M\qquad$}
    (encodingdetail.north);

    \node[
      font=\small\bfseries,
      anchor=north,
      yshift=-7pt
    ] at (encodingdetail.north) {Encoding};

  \end{tikzpicture}
  \caption{Construction of a linearisation of projective representation $R$ of finite abelian group $G$.}
  \label{fig:linearisation-diagram}
\end{figure}

The coding-theoretic correspondence immediately yields a reduction from Projective Anyonic Symmetry Learning over $G$ to Anyonic Symmetry Learning over $G^s$, where the linear representation is precisely the linearisation of $R$ corresponding to the chosen code $C$. Combining this reduction with the algorithms of the previous subsection gives an efficient algorithm for Projective Anyonic Symmetry Learning. Furthermore, if $C$ is also a \emph{zero-sum} code (meaning each of its codewords have entries summing to $0$ modulo $E$), the intermediate Anyonic Symmetry Learning instance can be bypassed entirely, yielding a direct reduction to Bose Symmetry Learning. 

\subsection{Optimising the parameters of linearisations}
The correspondence between linearisations and linear codes naturally raises the question of which codes induce the most efficient algorithms. Two parameters are particularly important: the block length, which determines the number of copies of the input state required for a single application of Fourier sampling, and the code rate, which determines the number of Fourier samples obtained per copy of the input state. Indeed, if a linearisation is induced by a code $C$ of dimension $s$ and block length $t$, then Fourier sampling with the corresponding linear representation produces $s$ Fourier samples from $t$ copies of the state, giving a sampling rate of $s/t$, precisely the rate of the code.

We therefore study the trade-off between minimising block length and maximising code rate. Small block lengths are desirable in near-term applications, where maintaining coherent access to many copies of a quantum state is difficult, while high code rates are advantageous when state preparation is expensive or only a limited number of copies is available. For general linearisations, we construct explicit families of constant-block-length codes achieving the optimal sampling rate and, for most values of the exponent $E$, explicit codes achieving every possible trade-off between block length and rate. When the code is additionally required to satisfy the zero-sum condition, we show using elementary linear algebra that achieving the optimal sampling rate necessarily requires block length at least linear in $E$.

Finally, we provide several explicit optimal constructions. In particular, we show that the generator matrices of the extended binary Hamming and ternary Golay codes are zero-sum codes that induce valid linearisations and achieve the optimal sampling rate for groups of exponent $E=2$ and $E=3$, respectively. The analysis combines elementary linear algebra with elegant number-theoretic arguments based on modular sums of squares and quadratic residues.

\section{Related work}

\paragraph{The State Hidden Subgroup Problem}

The State Hidden Subgroup Problem (\prb{StateHSP}) was introduced in \cite{BGTW25stateHSP}. Part of the motivation was to study a specific instance of abelian \prb{StateHSP}, the hidden cut problem, for which \cite{BGTW25stateHSP} gave an efficient quantum algorithm. Later, \cite{HEC25abelianStateHSP} provided an efficient quantum algorithm that works for all instances of abelian \prb{StateHSP}. \cite{LGFJ26dihedralStateHSP} is the first work to show an efficient quantum algorithm for \prb{StateHSP} over a specific family of non-abelian groups (namely, $n$ copies of the dihedral group of order $8$). \cite{GJMM26quantumStateIsomorphism} studied the problem of testing whether two quantum states are related by a group action, which they refer to as ``quantum state isomorphism (QSI) under group action''. They related several versions of QSI to various symmetry learning problems; in particular, they showed that QSI over abelian groups is reducible to \prb{StateHSP} over the generalised dihedral group, and that an approximate mixed state version of abelian \prb{StateHSP} is \cc{QSZK}-hard in the worst case, which means an efficient algorithm for that problem is unlikely. \cite{SKZDS26approximateHiddenCut} studied an approximate version of a specific instance of \prb{StateHSP} (the hidden cut problem) and provided a heuristic method for solving it.

\paragraph{Hidden subgroup problem} The well-known hidden subgroup problem \cite{Joz01HSP,NC10quantumComputationInformation} was introduced as a generalisation of the period finding problem which is solved in Shor's algorithm \cite{Sho97quantumFactoring} for factoring and discrete logarithms. HSP over finite abelian groups is well-known to be efficiently solvable by quantum computers \cite{NC10quantumComputationInformation}. \prb{HSP} over certain infinite abelian groups has been shown to also admit efficient quantum algorithms \cite{EHKS14quantumUnitGroup,dBDF20continuousHSP,Kup25infiniteHSP}. After the abelian HSP was shown to be efficiently solvable, attention turned to the non-abelian HSP, which has now been studied for over two decades. Earlier works on it include \cite{HRT03HSP} who provided an efficient quantum algorithm that finds the normal core of the hidden subgroup. \cite{GSVV01quantumNonabelianHSP} used this normal core algorithm to show an efficient quantum algorithm for non-abelian HSP over groups which are sufficiently close to Hamiltonian, \cite{Gav04HSP} improved upon this by expanding the range of ``near-Hamiltonian'' groups which admit an efficient HSP algorithm. \cite{EHK04queryComplexityHSP} showed that the quantum query complexity of any non-abelian \prb{HSP} is polynomial. Many works \cite{RB98HSP,Gav04HSP,GP11HSP,ISS12HSP,II23HSP} have shown efficient quantum algorithms for \prb{HSP} over specific families of non-abelian groups. Instances of \prb{HSP} over the symmetric group and dihedral group are of particular interest, due their connections to graph isomorphism \cite{MRS08symmetricGroupHSP} and lattice problems \cite{Reg04latticeProblems}. A subexponential time algorithm for dihedral \prb{HSP} \cite{Kup05dihedralHSP} is known, but no polynomial time algorithm is known. \cite{MRS08symmetricGroupHSP} showed that \prb{HSP} over the symmetric group is likely to be hard.

\paragraph{Testing symmetries} Collectively, \cite{LRW23testingQuantumSymmetries,LW22hamiltonianSymmetries,BRRW23testingQuantumSymmetries,RLW23quantumSymmetries} considered problems in property testing whether a given quantum object (pure state, mixed state, channel, Hamiltonian, measurement) is symmetric under the full action of a group. The combination of these four works show that symmetry testing can be done with respect to any \emph{projective} representation of any finite (possibly non-abelian) group. Furthermore, \cite{RLW23quantumSymmetries} shows that when the notion of distance is changed in the promise of the property testing problem (from $L^2$ distance to $L^1$ distance or fidelity), certain state symmetry testing problems become \cc{QZSK}-complete.

\section{Outlook and future directions}

In this work, we formally introduce a variety of symmetry learning problems which all fit into the same framework inspired by the \prb{StateHSP}. We present efficient quantum algorithms for a wide array of instances of these problems; namely, a normal core algorithm for \prb{StateHSP}, an algorithm for \prb{StateHSP} over polynomially near-Hamiltonian groups, algorithms for instances of a new symmetry learning problem (\emph{Anyonic Symmetry Learning}), and algorithms for learning symmetries of unitaries, Hamiltonians, and subsets of states. Furthermore, through our study of anyonic symmetry learning, we establish a new connection between linearisations of projective representations and linear codes.

We hope that our introduction of a more general symmetry learning framework encourages further progress in these area.\isaac{is this previous sentence needed?} Below, we outline some interesting and relevant directions for future research.

\paragraph{Applications of our general algorithms}

While we have identified a few applications of these algorithms, we expect that there are many more interesting and physically relevant problems in quantum computing that can be formulated as an instance of one of these symmetry learning problems, and so admits an efficient quantum algorithm via our general algorithms. We are particularly hopeful that our quantum algorithm for Projective Anyonic Symmetry Learning will find many useful applications, given the importance of projective representations in quantum physics; these arise, for example, when considering magnetic translations of a charged particle moving in a lattice.

\paragraph{Learning approximate symmetries}

Throughout this work, we use an \emph{exact} definition of symmetry. However, in practical scenarios, quantum computations are subject to noise, which means that input states no longer satisfy exact symmetries. Breaking of exact symmetries can also occur if we are only able to implement the representation $R$ approximately. Therefore, it is an important question to ask whether there exist efficient quantum algorithms for finding \emph{approximate} symmetries. With that motivation, we define the \textit{Approximate Anyonic Symmetry Learning} problem as follows:
\begin{definition}[Approximate Anyonic Symmetry Learning]\label{prb:approx-state-hsp}\leavevmode
    Given a group $G$, a representation $R$, a state $\rho$ and a threshold $\epsilon_1 \in [0, 1)$, find the set of all $\epsilon_1$-approximate symmetries $S = \{ g \in G: \abs{\Tr(R(g) \rho)} \geq 1 - \epsilon_1 \}$.
\end{definition}

One can define other approximate symmetry learning problems (e.g. involving unitaries or Hamiltonians) similarly.
  
It is fairly straightforward to show that abelian Approximate Anyonic Symmetry Learning can be solved with high probability when $\epsilon_1 = O(1 / \log \abs{G})$, using the same quantum algorithm as for (exact) abelian Anyonic Symmetry Learning. However, we leave it as an open problem to determine whether it can be solved with high probability for larger $\epsilon_1$ (ideally, for any $\epsilon_1$ below a fixed constant threshold).

Since the symmetries are now approximate rather than exact, the most general set of symmetries $S$ may not be a \emph{subgroup} of $G$ but rather an \emph{arbitrary subset}. This in itself leads to some interesting questions, such as whether this subset has some ``approximate'' algebraic structure (e.g. an approximate subgroup, with small doubling constant), which could lead to connections to additive combinatorics. Also, since subgroups have efficient (i.e. $\polylog \abs{G}$-sized) descriptions (as a set of their generators) whereas arbitrary subsets do not, we must change the notion of finding the hidden symmetry set from outputting a full description of the set. For \emph{subgroups}, learning a description of the subgroup and sampling from uniformly random elements of the subgroup are equivalent under randomised polynomial-time reductions. Thus, a natural generalisation of learning a symmetry \emph{subset} would be to sample from a distribution supported on the set; ideally, the uniform distribution.

\paragraph{Relaxing requirements on the group and representation}

The quantum symmetry testing algorithms in \cite{LRW23testingQuantumSymmetries,LW22hamiltonianSymmetries,BRRW23testingQuantumSymmetries,RLW23quantumSymmetries} work for symmetries characterised by any finite group $G$ and projective representation $R$. In contrast, the symmetry learning problems for which we provide quantum algorithms all require the group $G$ and representation $R$ to satisfy certain further properties. For example, for Anyonic Symmetry Learning, if $R$ is linear, then $G$ can be polynomially near-abelian (which strictly includes all finite abelian groups), but if $R$ is projective (fully general), then $G$ must be finite and abelian. For Bose Symmetry Learning, $G$ must be poly-near Hamiltonian. Therefore, a natural future direction is to determine which of these requirements can be relaxed, in order to close the gap between requirements needed for testing versus learning. For example:
\begin{itemize}
    \item Can we allow $G$ to be discrete but infinite, or even continuous?
    \item Can we allow $G$ to be any finite non-abelian group?
    \item Can we allow for $R$ to be projective for Anyonic Symmetry Learning over some non-abelian groups?
\end{itemize}
The ultimate goal would be to have a quantum algorithm for Anyonic Symmetry Learning that works for an arbitrary finite (or even infinite) group $G$ and projective representation $R$. Of course, given that \prb{HSP} is reducible to Anyonic Symmetry Learning and \prb{HSP} appears to be hard for general non-abelian groups (e.g. for the symmetric group), this appears unlikely;\isaac{TODO: is this sentence too negative?} however, it is possible that the more general framework of quantum symmetry learning offers new insights into the non-abelian case that the narrower lens of the Hidden Subgroup Problem failed to reveal.

\paragraph{Linearising projective representations and coding theory}

One of our most interesting contributions is establishing a correspondence between linearisations of projective representations of finite abelian groups and linear error correcting codes. By using this correspondence, we are able to solve Anyonic Symmetry Learning with respect to any projective representation of a finite abelian group. This connection between codes and linearisations opens up many exciting avenues for future research. We have already shown that optimising two properties of the code (the block length and rate) gives rise to more favourable properties of the corresponding linearisation. We suspect that other properties of the code, such as the minimum distance, may also relate to relevant properties of the linearisation.

A natural question is whether this linearisation construction can be generalised to projective representations of finite \emph{non-abelian} groups. Such a generalisation would require some non-trivial new ideas: the current construction for abelian groups works because the map $g \mapsto g^k$ is a group homomorphism for any abelian $G$ and $k \in \ZZ$. However, it is generally not a homomorphism for non-abelian groups.

We expect that this new notion of linearisation and its connections to linear codes may have other interesting applications outside of symmetry learning.
Indeed, in an upcoming work, we show that linearisations of the phaseless Pauli representation of $\ZZ_d^{2n}$ give rise to a general notion of \emph{Bell sampling} and \emph{Bell difference sampling}, two important quantum algorithmic primitives which have found a vast range of applications in quantum algorithms, property testing and learning theory \cite{Mon17,GNW21bellDifferenceSampling,GIKL24pseudoentanglement,GIKL24stabiliserEstimation,GIKL25learningFewNonClifford,CGYZ25stabilizerBootstrapping,AD25tolerantStabiliserTesting}. Our generalisation includes as special cases all previous Bell (difference) sampling algorithms from \cite{Mon17,GNW21bellDifferenceSampling,HEC25abelianStateHSP,ADIS25quditBellSampling}. In this upcoming work, we also explore some of the applications of this generalised Bell (difference) sampling procedure.

\section{Preliminaries}

Write $\log$ for the base-$2$ logarithm. Denote the set of natural numbers between $1$ and $n$ by $[n]$. Write $S^1 = \{z \in \CC: \abs{z} = 1\}$ for the unit circle and $\DD = \{z \in \CC: \abs{z} \leq 1\}$ for the unit disc.

For a Hilbert space $\mathcal{H}$, write $\densmats{\mathcal{H}}$ for the set of density operators on $\mathcal{H}$, i.e. the set of positive semidefinite linear operators $\rho$ on $\mathcal{H}$ with $\Tr(\rho) = 1$. Write $\mathcal{U}(\mathcal{H})$ for the set of unitary operators on $\mathcal{H}$. For operators $A$ and $B$ on $\mathcal{H}$, write $[A, B] = A B - B A$ for their commutator.

\subsection{Quantum information}

We assume all Hilbert spaces to be finite-dimensional. A \emph{quantum channel} is a completely positive, trace preserving map between density operators on Hilbert spaces. For every quantum channel $\mathcal{E}: \densmats{\mathcal{H}} \to \densmats{\mathcal{H}'}$, there is an associated \emph{Choi state} $\Phi^{\mathcal{E}} \in \densmats{\mathcal{H} \tp \mathcal{H}'}$ which completely encodes the action of $\mathcal{E}$, and is defined by $\Phi^{\mathcal{E}} = (\mathcal{E} \tp \text{id})(\ketbra{\Phi^+}{\Phi^+})$, where $\ket{\Phi^+} = \frac{1}{\sqrt{N}} \sum_{x = 1}^N \ket{x} \ket{x}$ is the maximally entangled state on $\mathcal{H} \tp \mathcal{H}$ and $N = \dim \mathcal{H}$. A quantum channel $\mathcal{E}$ is a \emph{unitary channel} (i.e. $\mathcal{E}(\rho) = U \rho U^\dagger$ for some unitary $U$) if and only if its Choi state $\Phi^{\mathcal{E}}$ is pure; in which case, we may write $\Phi^U$ for the Choi state of $\mathcal{E}$.

\subsection{Coding theory}


\begin{definition}[Linear Code]\label{def:linear-code}
    For $d \geq 2$, a \emph{(linear) code} $C$ over $\ZZ_d$ of \emph{block length $n$} is a submodule (i.e. subgroup) of $\ZZ_d^n$.
\end{definition}

\begin{definition}[Free Code]\label{def:free-code}
    $C$ is called \emph{free} if it is a free submodule, i.e. if it has a basis (a generating set which is linearly independent). The \emph{dimension} $\dim C$ of a free code $C$ is the size of any basis for $C$. The \emph{rate} of $C$ is the ratio between its dimension and its block length, $r(C) = k / n$.
\end{definition}

\begin{definition}[Generator Matrix]\label{def:generator-matrix}
    A \emph{generator matrix} for a free code $C$ is any matrix $G \in \ZZ_d^{k \times n}$ whose rows form a basis for $C$.
\end{definition}

\begin{definition}[Dual Code]\label{def:dual-code}
    The \emph{dual code} of $C$ is the code $C^\perp = \{x \in \ZZ_d^n : x . c = 0 \: \forall c \in C\}$. If $C$ is free with dimension $k$, then $C^\perp$ is also free with dimension $n - k$. $C$ is called \emph{self-orthogonal} if $C \subseteq C^\perp$, and \emph{self-dual} if $C = C^\perp$.
\end{definition}

\begin{definition}[Check Matrix]\label{def:check-matrix}
    A \emph{check matrix} for a free code $C$ is any matrix $H \in \ZZ_d^{(n - k) \times n}$ which is a generator matrix for $C^\perp$.
\end{definition}

We will be interested in three properties of linear codes: self-orthogonal, zero-sum, and isotropic.

\begin{definition}[Self-orthogonal Code]\label{def:self-orthogonal-code}
    A free code $C$ is \emph{self-orthogonal} if $C \subseteq C^\perp$, i.e. if $x . y = 0$ for all $x, y \in C$.
\end{definition}

\begin{definition}[Isotropic Code]\label{def:isotropic-code}
    A free code $C$ is \emph{isotropic} if all codewords are orthogonal to themselves, i.e. $x . x = 0$ for all $x \in C$.
\end{definition}

\begin{definition}[Zero-Sum Code]\label{def:zero-sum-code}
    A free code $C$ is \emph{zero-sum} if the sum of the entries of each codeword is zero, i.e. $\sum_{i = 1}^n x_i = 0 \pmod{d}$ for all $x \in C$.
\end{definition}

\subsection{Number theory}

Key to our construction of linearisations of projective representations is Lagrange's four-square theorem, which says that every natural number can be decomposed into a sum of four integer squares.

\begin{fact}[Lagrange's Four-Square Theorem]\label{fct:lagrange-four-square-theorem}
    For any natural number $n \in \NN$, there exist $a, b, c, d \in \ZZ$ such that $n = a^2 + b^2 + c^2 + d^2$.
\end{fact}
    
Moreover, there exist randomised (classical) algorithms (\cite{RS86integerSquareDecompositions,PT18lagrangeFourSquares}) for finding explicit decompositions of integers $n$ as a sum of four squares in expected time $\polylog(n)$.

\subsection{Group theory}

\begin{definition}[Group Exponent]\label{def:group-exponent}
    The \emph{exponent} $\exp(G)$ of a finite group $G$ is the smallest positive integer $E$ such that $g^E = e_G$ for all $g \in G$, where $e_G$ is the identity element of $G$. Equivalently, $\exp(G) = \lcm(\{\ord(g): g \in G\})$.
\end{definition}

\begin{definition}[Normal Core]\label{def:normal-core}
    The \emph{normal core} $\core_G (H)$ (or just $\core(H)$ when $G$ is clear) of a subgroup $K \leq G$ is the largest normal subgroup of $G$ contained in $K$.
\end{definition}

\begin{definition}[Normal Closure]\label{def:normal-closure}
    The \emph{normal closure} $\ncl_G (H)$ (or just $\ncl(H)$ when $G$ is clear) of a subgroup $H \leq G$ is the smallest normal subgroup of $G$ containing $H$.
\end{definition}

\begin{definition}[Normaliser]\label{def:subgroup-normaliser}
    The \emph{normaliser} $N_G (K)$ of a subgroup $K \leq G$ in $G$ is the largest subgroup $L$ of $G$ such that $K$ is normal in $L$.
\end{definition}

\begin{definition}[Subgroup Index]\label{def:subgroup-index}
    The \emph{index} of a subgroup $K \leq G$ in a finite group $G$ is $[G: K] = \abs{G} / \abs{K}$.
\end{definition}

\begin{fact}[Lagrange's Theorem]\label{fct:lagranges-theorem}
    The index of a subgroup $K \leq G$ in a finite group $G$ is integral, i.e. $\abs{K}$ divides $\abs{G}$.
\end{fact}

\subsection{Representation theory}

\begin{definition}[Linear Representation]\label{def:linear-representation}
    A \emph{unitary linear representation} of a group $G$ on a vector space $V$ is a homomorphism $R: G \to \mathcal{U}(V)$.
\end{definition}

\begin{definition}[Projective Representation]\label{def:projective-representation}
    A \emph{unitary projective representation} is a map $R: G \to \mathcal{U}(V)$ such that
    \begin{equation*}
        R(g) R(h) = \omega(g, h) R(g h) \quad \forall g, h \in G,
    \end{equation*}
    where $\omega: G \times G \to S^1$ is called the \emph{co-cycle}.
\end{definition}

\noindent Note that every linear representation is a projective representation with trivial co-cycle. We will often use the term ``projective representation'' to refer to a non-linear (i.e. has a non-trivial co-cycle) projective representation, and use ``representation'' to refer to a linear or non-linear projective representation.

\begin{notation}
    Write $\ctrl{R}$ for the unitary on $\CC^G \tp V$ defined by $\ctrl{R} = \sum_{g \in G} \ketbra{g}{g} \tp R(g)$.
\end{notation}

\begin{definition}[Unitary Equivalence]\label{def:unitary-equivalence}
    Two representations $R: G \to \mathcal{U}(V)$ and $R': G \to \mathcal{U}(V)$ are \emph{unitarily equivalent} if there exists a unitary $U \in \mathcal{U}(V)$ such that $R'(g) = U R(g) U^\dagger$ for all $g \in G$.
\end{definition}

\begin{definition}[Dimension]
    The \emph{dimension} $d_R$ of a representation $R$ is the dimension of the vector space it acts on.
\end{definition}

\begin{fact}[Pontryagin Duality]\label{fct:pontryagin-duality}
    If $G$ is a finite abelian group, then $\widehat{G}$ is also a finite abelian group, and $G \cong \widehat{\widehat{G}}$.
\end{fact}

\noindent When working with finite abelian groups, we will occasionally use Pontryagin duality to identify $G$ with $\widehat{\widehat{G}}$.

\begin{definition}[Irreducibility]\label{def:irrep}
    A linear representation $R$ is \emph{irreducible} if there are no non-trivial subspaces of $V$ which are invariant under the action of every $R(g)$. We refer to irreducible representations as \emph{irreps}. A finite group $G$ has a finite number of irreps up to unitary equivalence. We denote a set of the labels of such irreps as $\hat{G}$. We often label irreps of $G$ by $\lambda$, and write $R_\lambda$ for the irrep with the label $\lambda$. Write $d_\lambda$ for the dimension of $R_\lambda$.
\end{definition}

\begin{definition}[Character]\label{def:character}
    The \emph{character} $\chi_R$ of a representation $R$ is the function $\chi_R: G \to \CC$ defined by $\chi_R (g) = \Tr(R(g))$. Write $\chi_\lambda$ for the character of an irrep labelled by $\lambda \in \hat{G}$.
\end{definition}

\begin{notation}
    For a subgroup $K \leq G$, write $\hat{G}[K] = \{\lambda \in \hat{G}: K \subseteq \ker R_\lambda\}$ for the set of irreps which act trivially on $K$. Since $\ker R_\lambda$ is always normal in $G$, $\hat{G}[K] = \hat{G}[\ncl(K)]$.
\end{notation}

\begin{definition}[Direct and Tensor Products]\label{def:direct-and-tensor-products-of-representations}
    Given two representations $R: G \to \mathcal{U}(V)$ and $R': G \to \mathcal{U}(V')$, we can define their \emph{direct product} $R \oplus R': G \to \mathcal{U}(V \oplus V')$ by $(R \oplus R')(g) = R(g) \oplus R'(g)$ for all $g \in G$, and their \emph{tensor product} $R \tp R': G \to \mathcal{U}(V \tp V')$ by $(R \tp R')(g) = R(g) \tp R'(g)$ for all $g \in G$. These are both clearly also representations.
\end{definition}

An important linear representation is the \emph{left regular representation} $R_\text{lreg}: G \to \mathcal{U}(\CC^G)$ which is defined linearly by $R_\text{lreg}(g) \ket{h} = \ket{g h}$ for all $g, h \in G$. Similarly, the \emph{right regular representation} $R_\text{rreg}: G \to \mathcal{U}(\CC^G)$ is defined linearly by $R_\text{rreg}(g) \ket{h} = \ket{h g^{-1}}$ for all $g, h \in G$. The left and right regular representations are unitarily equivalent.

\begin{fact}[Decomposition of Regular Representation]\label{fct:regular-representation-irrep-decomposition}
    The left and right regular representations of a finite group decompose into irreps as
    \begin{equation*}
        R_\text{lreg} \cong R_\text{rreg} \cong \bigoplus_{\lambda \in \hat{G}} R_\lambda^{\oplus d_\lambda}.
    \end{equation*}
    In particular, the character of both regular representations, $r_G: G \to \CC$, is given by $r_G (g) = \sum_{\lambda \in \hat{G}} d_\lambda \chi_\lambda (g)$.
\end{fact}

\begin{fact}[{\cite[Proposition 2.5]{Ser77representationTheory}}]\label{fct:regular-representation-character-values}
    $r_G (g) = 0$ for all $g \in G \setminus \{e_G\}$, and $r_G (e_G) = \abs{G}$, where $e_G$ is the identity element of $G$.
\end{fact}

\begin{fact}[Schur Orthogonality]\label{fct:schur-orthogonality}
    Let $\lambda, \mu \in \hat{G}$. Then for all $i, j \in [d_\lambda]$ and $k, l \in [d_\mu]$,
    \[
        \frac{1}{|G|} \sum_{g\in G} R_\lambda(g)_{ij}\overline{R_\mu(g)_{kl}} =
        \begin{cases}
            0 & \text{ if } \lambda \neq \mu, \\
            \frac{1}{d_\lambda} \delta_{ik}\delta_{jl} & \text{ if } \lambda = \mu.
        \end{cases}
    \]
\end{fact}

An important fact about projective representations of finite abelian groups is that the cocycle must be of a specific form:

\begin{fact}\label{fct:abelian-projective-representation-cocycle-form}
    Let $G$ be a finite abelian group and let $R: G \to \mathcal{U}(V)$ be a projective representation of $G$ with cocycle $\omega$. Then up to a \emph{gauge transformation} $f: G \to S^1$ (i.e. replacing $R(g)$ with $f(g) R(g)$), $\omega$ can be written as
    \begin{equation*}
        \omega(g, h) = \omega^{B(g, h)} \quad \forall g, h \in G,
    \end{equation*}
    where $\omega$ is a primitive $k$-th root of unity for some $k \in \NN$, and $B: G \times G \to \ZZ_k$ is a bilinear map, i.e. $B(g + g', h) = B(g, h) + B(g', h)$ and $B(g, h + h') = B(g, h) + B(g, h')$ for all $g, g', h, h' \in G$.

    Since $B$ is bilinear, we must have $0 = B(g, 0) = B(g, E \cdot h) = E \cdot B(g, h)$ for all $g, h \in G$, where $E = \exp(G)$ is the exponent of $G$. So $\omega(g, h)^E = 1$ for all $g, h \in G$. So by multiplying $B$ by a scalar if necessary, we can assume that $k = E$.
\end{fact}

\subsection{Fourier analysis}

Let $G$ be a finite group. Recall that the set $\CC^G$ of functions $f : G \rightarrow \CC$ is a vector space over $\CC$ of dimension $\abs{G}$. Define the usual inner product $\gen{\cdot, \cdot}$ on $\CC^G$ by
\begin{equation*}
    \gen{f_1, f_2} \coloneq \frac{1}{\abs{G}} \sum_{g \in G} \overline{f_1 (g)} f_2(g).
\end{equation*}

\begin{definition}[Fourier Transform]
    Let $f: G \to \CC$. The \emph{Fourier transform} of $f$ at an irrep labelled by $\lambda \in \hat{G}$ is defined as
    \begin{equation*}
        \hat{f}(\lambda) \coloneq \frac{1}{\abs{G}} \sum_{g \in G} f(g) R_\lambda (g)^\dagger.
    \end{equation*}
    In particular, if $G$ is abelian, then $\hat{f}(\lambda) = \frac{1}{\abs{G}} \sum_{g \in G} f(g) \overline{\chi_\lambda (g)} = \gen{\chi_\lambda, f}$.
\end{definition}

\begin{definition}[Quantum Fourier Transform]
    The \emph{quantum Fourier transform} over a finite group $G$ is the unitary map $\QFT$ defined by
    \begin{equation*}
        \QFT \ket{g} = \sum_{\lambda \in \hat{G}} \sum_{i, j = 1}^{d_\lambda} \sqrt{\frac{d_\lambda}{|G|}} R_\lambda(g)_{ij} \ket{\lambda, i, j}.
    \end{equation*}
\end{definition}

\noindent Throughout this work, we assume that any group $G$ we consider admits a quantum Fourier transform which can be implemented by a $\polylog \abs{G}$-sized quantum circuit. This is true for all finite abelian groups \cite{CVD10quantumAlgorithmsAlgebraicProblems} as well as many commonly occurring non-abelian groups \cite{Bea97symmetricGroupQFT,MRS04QFT,MALHG24nonAbelianQFT}.

For the rest of this subsection, we assume $G$ to be a finite abelian group, and let $K$ be an arbitrary subgroup of $G$.

\begin{definition}\label{def:convolution}
    For $f_1, f_2 : G \rightarrow \CC$, the \textbf{convolution} of $f_1$ and $f_2$ is
    \begin{align*}
        f_1 * f_2 : G & \to \CC, \\
        g & \mapsto \EE_{h \in G} f_1 (h) f_2 (g - h).
    \end{align*}
\end{definition}

\begin{fact}[Convolution Theorem]\label{fct:convolution-theorem}
    For all $f_1, f_2 : G \rightarrow \CC$ and $\lambda \in \widehat{G}$, then
    \begin{equation*}
        \widehat{(f_1 * f_2)}(\lambda) = \widehat{f_1}(\lambda) \widehat{f_2}(\lambda).
    \end{equation*}
\end{fact}


\begin{definition}\label{def:dual-subgroup}
    The \textbf{dual subgroup} $K^{\perp}$ of $K$ in $\widehat{G}$ is defined as
    \begin{equation*}
        K^{\perp} = \{\lambda \in \widehat{G} : \chi_\lambda(k) = 1 \: \forall k \in K\} = \widehat{G}[K].
    \end{equation*}
\end{definition}

\noindent We recall a few standard facts about dual subgroups:

\begin{fact}\label{prop:facts-about-dual-subgroup}\leavevmode
    \begin{itemize}
        \item $K^{\perp}$ is a subgroup of $\widehat{G}$
        \item If we view $K^{\perp}$ as a subgroup of $G$ (via Pontryagin duality), then $(K^{\perp})^{\perp} = K$.
        \item For any subgroup $L \leq G$, $L \subseteq K \Longleftrightarrow K^{\perp} \subseteq L^{\perp}$ (inclusion-reversing).
        \item $\abs{K} \cdot \abs{K^{\perp}} = \abs{G}$.
    \end{itemize}
\end{fact}

\begin{fact}[Subgroup Orthogonality Relations]\label{lmm:subgroup-orthogonality-relations}
    For all $x \in G$, 
    \begin{equation*}
        \sum_{\lambda \in K^{\perp}} \chi_\lambda (x) = \begin{cases}
            \left|K^{\perp}\right| & \text{if } x \in K \\
            0 & \text{otherwise} .
        \end{cases}
    \end{equation*}
    Also, for any $\lambda \in \widehat{G}$,
    \begin{equation*}
        \sum_{g \in K} \chi_\lambda (g) = \begin{cases}
            \abs{K} & \text{if } \lambda \in K^{\perp} \\
            0 & \text{otherwise} .
        \end{cases}
    \end{equation*}
\end{fact}

\begin{proof}
    The first equality follows immediately if $x \in K$, since then $\chi_\lambda (x) = 1$ for all $\lambda \in K^{\perp}$. So suppose $x \notin K$, then there is some $\lambda' \in K^{\perp}$ such that $\chi_{\lambda'}(x) \neq 1$ (if not, then $K^{\perp} \subseteq \gen{x, K}^{\perp}$, so $\gen{x, K} \subseteq K$, but $x \notin K$). Now
    \begin{align*}
        \sum_{\lambda \in K^{\perp}} \chi_\lambda(x) = \sum_{\lambda \in K^{\perp}} (\chi_\lambda \chi_{\lambda'}) (x) = \chi_{\lambda'}(x) \sum_{\lambda \in K^{\perp}} \chi_\lambda(x),
    \end{align*}
    hence $\sum_{\lambda \in K^{\perp}} \chi_\lambda(x) = 0$.

    The proof is similar for the second equality: if $\lambda \in K^{\perp}$, then $\chi_\lambda (g) = 1$ for all $g \in K$, so the result follows immediately. If $\lambda \notin K^{\perp}$, then there must be some $h \in K$ such that $\chi_\lambda(h) \neq 1$. Now
    \begin{align*}
        \sum_{g \in K} \chi_\lambda(g) & = \sum_{g \in K} \chi_\lambda(h + g) = \chi_\lambda(h) \sum_{g \in K} \chi_\lambda(g),
    \end{align*}
    so $\sum_{g \in K} \chi_\lambda (g) = 0$.
\end{proof}

\begin{fact}[Poisson Summation Identity]\label{lmm:poisson-summation-identity} 
    Let $f : G \rightarrow \CC$. Then
    \begin{equation*}
        \frac{1}{\abs{K}} \sum_{g \in K} f(g) = \EE_{g \in K} f(g) = \sum_{\lambda \in K^\perp} \widehat{f}(\lambda).
    \end{equation*}
\end{fact}

\begin{proof}
    We have
    \begin{align*}
        \sum_{\lambda \in K^{\perp}} \widehat{f}(\lambda) & = \sum_{\lambda \in K^{\perp}} \frac{1}{\abs{G}} \sum_{g \in G} \chi_\lambda (g) f(g) \\
        & = \frac{1}{\abs{G}} \sum_{g \in G} f(g) \sum_{\lambda \in K^\perp} \chi_\lambda (g) \\
        & = \frac{1}{\abs{G}} \sum_{g \in G} f(g) \cdot \abs{K^\perp} \indic_K (g) & \quad \text{by \cref{lmm:subgroup-orthogonality-relations}} \\
        & = \frac{\left|K^\perp\right|}{\abs{G}} \sum_{g \in K} f(g),
    \end{align*}
    and the result follows since $\abs{K} \cdot \abs{K^\perp} = \abs{G}$.
\end{proof}

\noindent A version of the Poisson summation identity also holds when summing over the coset of a dual subgroup:

\begin{lemma}[Poisson Summation Identity on Cosets]\label{lmm:poisson-summation-identity-on-cosets} 
    Let $f : G \rightarrow \CC$. Then
    \begin{equation*}
        \sum_{\lambda \in \lambda' + K^\perp} \widehat{f}(\lambda) = \EE_{g \in K} \overline{\chi_{\lambda'} (g)} f(g).
    \end{equation*}
\end{lemma}

\begin{proof}
    We have
    \begin{align*}
        \sum_{\lambda \in K^\perp} \widehat{f}(\lambda + \lambda') & = \sum_{\lambda \in K^\perp} \frac{1}{\abs{G}} \sum_{g \in G} \overline{\chi_\lambda (g)} \overline{\chi_{\lambda'} (g)} f(g) \\
        & = \frac{1}{\abs{G}} \sum_{g \in G} \overline{\chi_{\lambda'} (g)} f(g) \sum_{\lambda \in K^\perp} \overline{\chi_\lambda (g)} \\
        & = \frac{1}{\abs{G}} \sum_{g \in G} \overline{\chi_{\lambda'} (g)} f(g) \cdot \abs{K^\perp} \indic_K (g) & \quad \text{by \cref{lmm:subgroup-orthogonality-relations}} \\
        & = \frac{\left|K^\perp\right|}{\abs{G}} \sum_{g \in K} \overline{\chi_{\lambda'} (g)} f(g),
    \end{align*}
    and the result follows since $\abs{K} \cdot \abs{K^\perp} = \abs{G}$.
\end{proof}

\subsection{Fourier sampling and Bose symmetry learning}

Recall the Fourier sampling primitive from \cite{HEC25abelianStateHSP, BGTW25stateHSP}:

\begin{algorithm}
\caption{Fourier Sample}\label{alg:state-fourier-sample}
\begin{algorithmic}[1]
    \State Starting with the state $\ket{0} \ket{\psi}$, apply the inverse QFT over $G$ to obtain $\frac{1}{\sqrt{\abs{G}}} \sum_{x \in G} \ket{x} \ket{\psi}$.
    \State Apply $\ctrl{R}$, giving $\frac{1}{\sqrt{\abs{G}}} \sum_{x \in G} \ket{x} \tp R(x) \ket{\psi}$.
    \State Apply the quantum Fourier transform over $G$ to the first register.
    \State Measure the irrep label register $\ket{\lambda}$ in the computational basis.
\end{algorithmic}
\end{algorithm}

\begin{fact}[{\cite{BGTW25stateHSP,HEC25abelianStateHSP}}]\label{lmm:state-fourier-sampling-output-distribution}
    For abelian $G$, \cref{alg:state-fourier-sample} implements the projective measurement $\{\Pi_\lambda : \lambda \in \widehat{G}\}$, where $\Pi_\lambda = \frac{1}{\abs{G}} \sum_{x \in G} \overline{\chi_{\lambda} (x)} R(x)$, i.e. the output distribution $q_{\rho, R}: \widehat{G} \rightarrow [0, 1]$ of \cref{alg:state-fourier-sample} performed on a mixed state $\rho$ is
    \begin{equation*}
        \lambda \mapsto \frac{1}{\abs{G}} \sum_{x \in G} \chi_{\lambda} (x) \Tr(R(x) \rho)
    \end{equation*}
    Note that we can perform the transformation $\lambda \mapsto -\lambda$ on the output of \cref{alg:state-fourier-sample}. This means that the output distribution becomes
    \begin{equation}
        q_{\rho, R} (\lambda) = \frac{1}{\abs{G}} \sum_{x \in G} \chi_{-\lambda} (x) \Tr(R(x) \rho) = \frac{1}{\abs{G}} \sum_{x \in G} \overline{\chi_\lambda (x)} p_{\rho, R} (x) = \hat{p}_{\rho, R} (\lambda) \label{eq:fourier-sampling-output-distribution},
    \end{equation}
    where $p_{\rho, R} : G \rightarrow \CC$ is the function defined by $p_{\rho, R}(x) = \Tr(R(x) \rho)$. Note that $p_{\rho, R} (x)$ gives a measure of how symmetric $\rho$ is under $R(x)$.
\end{fact}

When working with abelian groups, we assume that the transformation $\lambda \mapsto -\lambda$ is always applied after Fourier sampling so that the output distribution of the procedure is as in \cref{eq:fourier-sampling-output-distribution}. This has the nicer interpretation of being the Fourier transform of $p_{\rho, R}$.

\begin{proof}
    First consider when the input $\rho$ is a pure state $\ketbra{\psi}{\psi}$. After step 3 of \cref{alg:state-fourier-sample}, the state is
    \begin{equation}
        \frac{1}{\abs{G}} \sum_{g \in G} \sum_{\lambda \in \widehat{G}} \chi_\lambda (g) \ket{\lambda} \tp R(g) \ket{\psi} = \frac{1}{\abs{G}} \sum_{\lambda \in \widehat{G}} \ket{\lambda} \tp \left(\sum_{g \in G} \overline{\chi_\lambda (g)} R(g) \ket{\psi}\right).
    \end{equation}
    Now the probability of measuring $\lambda$ is
    \begin{align}
        q_{\psi, R} (\lambda) & = \frac{1}{\abs{G}^{2}} \Norm{\sum_{g \in G} \chi_\lambda (g) R(g) \ket{\psi}}^{2} \\
        & = \frac{1}{\abs{G}^{2}} \sum_{g, h \in G}\chi_{\lambda} (g) \chi_\lambda (h) \bra{\psi} R(g)^{\dagger} R(h) \ket{\psi} \\
        & = \frac{1}{\abs{G}^{2}} \sum_{g, h \in G} \chi_\lambda (-g) \chi_\lambda (h) \braket{\psi | R(-g) R(h) | \psi} \\
        & = \frac{1}{\abs{G}^{2}} \sum_{g, h \in G} \chi_\lambda (h - g) \braket{\psi | R(h - g) | \psi} \\
        & = \frac{1}{\abs{G}} \sum_{g \in G} \chi_\lambda (g) \braket{\psi | R(g) | \psi},
    \end{align}
    where the third line follows from the fact that $R$ is unitary
    representation, the fourth line follows from the fact $\chi_{\lambda}$ and $R$ are homomorphisms and $G$ is abelian, and the fifth line follows from doubling summing over the group. This shows the result for pure states; the result for general mixed states follows by linearity.
\end{proof}

\begin{fact}\label{lmm:fourier-sampling-distribution-mass-on-subgroup-dual}
    Let $G$ be a abelian. For a subgroup $K < G$, we have
    \begin{equation*}
        q_{\rho, R} (K^\perp) \coloneqq \sum_{\lambda \in K^\perp} q_{\rho, R} (\lambda) = \EE_{g \in K} p_{\rho, R} (g)
    \end{equation*}
\end{fact}
\begin{proof}
    This follows from \cref{lmm:state-fourier-sampling-output-distribution} and the subgroup orthogonality relations.
\end{proof}

Below we recall the mixed state version of the State Hidden Subgroup Problem (\prb{StateHSP}) \cite{BGTW25stateHSP,HEC25abelianStateHSP}. Given that a key focus of this paper is on learning different kinds of symmetries, we find it more appropriate to refer to this problem as \emph{Bose symmetry learning}, which captures the specific type of symmetry that we are interested in here.

\begin{problem}[Bose Symmetry Learning (\prb{BoseSL})]\label{prb:state-hsp}\leavevmode
    \begin{description}[style=standard]
        \item[Input] Copies of a state $\rho \in \densmats{\mathcal{H}}$.
        \item[Promise] There is a finite group $G$, a unitary linear representation $R: G \rightarrow \mathcal{U}(\mathcal{H})$, a subgroup $S \leq G$, and $0 \leq \epsilon < 1$ such that: 
        \begin{itemize}
            \item $\rho$ possesses \emph{Bose symmetry} under the action of $S$:
            \begin{equation*}
                \Tr(R(g) \rho) = 1 \quad \forall g \in S .
            \end{equation*}
            \item For all $g \notin S$, $\rho$ is $\epsilon$-far from possessing Bose symmetry under the action of $g$:
            \begin{equation*}
                \abs{\Tr(R(g) \rho)} \leq 1 - \epsilon \quad \forall g \in G \setminus S .
            \end{equation*}
        \end{itemize}
        \item[Task] Find (a set of generators of) the hidden symmetry subgroup $S$.
    \end{description}
    Write such an instance $\mathcal{IB}$ of \prb{BoseSL} as $\mathcal{IB} = \prb{BoseSL}(G, R, \rho)$. Write $\HSS(\mathcal{IB})$ for the associated \emph{hidden symmetry subgroup} $S$ and $\MASD(\mathcal{IB})$ for the associated \emph{minimum asymmetry distance} $\epsilon$.
\end{problem}

\noindent To show an efficient quantum algorithm for \prb{AbelianBoseSL}, \cite{HEC25abelianStateHSP} use the following two facts which are corollaries of \cref{lmm:fourier-sampling-distribution-mass-on-subgroup-dual}:
\begin{fact}[{\cite[Equation 22]{HEC25abelianStateHSP}}]\label{fct:fourier-sampling-distribution-supported-on-dual-of-symmetry-subgroup}
    The distribution $q_{\rho, R}$ is supported on the dual subgroup $S^\perp$, i.e. $q_{\rho, R} (S^\perp) = 1$.
\end{fact}
\begin{fact}[{\cite[Lemma 3]{HEC25abelianStateHSP}}]\label{fct:fourier-sampling-distribution-anti-concentration}
    For any subgroup $K > S$, the distribution $q_{\rho, R}$ is anti-concentrated on $K^\perp$, i.e. $q_{\rho, R} (K^\perp) \leq 1 - \epsilon / 2$.
\end{fact}

They use these properties to show that $O(\log \abs{G})$ samples from $q_{\rho, R}$ (which are obtained by Fourier sampling) generate the full dual subgroup $S^\perp$, with high probability. $S$ can then be recovered from $S^\perp$ by using standard Gaussian elimination techniques.

\section{A normal core algorithm for Bose symmetry learning}\label{sec:normal-core-state-hsp}


In this section, we prove that there is an efficient quantum algorithm for finding the normal core of the hidden symmetry subgroup $\HSS(\mathcal{IB})$ of any instance $\mathcal{IB} = \prb{BoseSL}(G, R, \rho)$ of (possibly non-abelian) Bose symmetry learning, where an efficient quantum Fourier transform over $G$ exists. 

\begin{theorem}\label{thm:efficient-algorithm-for-normal-core-statehsp}
    Let $\mathcal{IB} = \prb{BoseSL}(G, R, \rho)$ be an instance of Bose symmetry learning with hidden symmetry subgroup $S = \HSS(\mathcal{IB})$ and minimum asymmetry distance $\epsilon = \MASD(\mathcal{IB})$. Let $C = \core_G (S)$ be the normal core of $S$ in $G$, and let $\mathcal{M}$ be the set of all minimal non-trivial subgroups of $G / C$.

    Then there is a quantum algorithm that, using $O(\log \abs{\mathcal{M}} + \log(1 / \delta)) / \epsilon$ copies of $\rho$, outputs a set of generators for $C$ with probability at least $1 - \delta$.

    Moreover, if the quantum Fourier transform over $G$ and the controlled representation unitary $\ctrl{R} = \sum_{g \in G} \ketbra{g}{g} \tp R(g)$ can be implemented efficiently (i.e. in time $\polylog \abs{G}$), then the algorithm runs efficiently.
\end{theorem}

Note that $\abs{\mathcal{M}} \leq \abs{G}$, since every minimal non-trivial subgroup of $G / C$ is generated by a single element of $G / C$.

We recover the quantum algorithm for abelian \prb{StateHSP} from \cite{HEC25abelianStateHSP} and the normal core quantum algorithm for \prb{HSP} from \cite{HRT03HSP} as special cases of the algorithm in \cref{thm:efficient-algorithm-for-normal-core-statehsp}. Since the sample complexity of the algorithms in \cite{HEC25abelianStateHSP} and \cite{HRT03HSP} is $O(\log \abs{G} + \log(1 / \delta)) / \epsilon$, we have improved sample complexity upper bound compared to the original algorithms, due to the dependence on $\abs{\mathcal{M}}$ rather than $\abs{G}$.

\begin{corollary}
    There is a quantum algorithm that solves abelian \prb{StateHSP} (abelian \prb{BoseSL}) with probability at least $1 - \delta$ using $O(\log \abs{\mathcal{M}} + \log(1 / \delta)) / \epsilon$ copies of $\ket{\psi}$.
\end{corollary}

\begin{proof}
    This follows immediately from \cref{thm:efficient-algorithm-for-normal-core-statehsp}, since every subgroup of an abelian group is normal.
\end{proof}

To recover the algorithm from \cite{HRT03HSP}, we recall a result from \cite{BGTW25stateHSP} that \prb{HSP} is reducible to \prb{StateHSP}, which we prove for completeness's sake:

\begin{lemma}\label{fct:hsp-can-be-formulated-as-state-hsp}
    \prb{HSP} is reducible to \prb{BoseSL} (\prb{StateHSP}) with minimum asymmetry distance $\epsilon = 1$.
\end{lemma}

\begin{proof}
    Let $\mathcal{I} = \prb{HSP}(G, f)$ be an instance of \prb{HSP} with hidden subgroup $S$. Define the state $\ket{\psi}$ by
    \begin{align*}
        \ket{\psi} & \coloneq \ket{\Graph(f)} = \frac{1}{\sqrt{\abs{G}}} \sum_{g \in G} \ket{g} \ket{f(g)} = \sqrt{\frac{\abs{S}}{\abs{G}}} \cdot \sum_{c \in G / S} \ket{c S} \ket{f(c)} \in \CC^G \tp \CC^X,
    \end{align*}
    where $X$ is the range of $f$. Define the representation $R: G \rightarrow \mathcal{U}(\CC^G) \tp \mathcal{U}(\CC^X)$ by
    \begin{equation*}
        R(g) \ket{h} \ket{x} = \ket{h g^{-1}} \ket{x} \quad \forall g, h \in G,
    \end{equation*}
    i.e. $R$ is the tensor product of the right-regular representation of $G$ with the identity representation acting on $\CC^X$. Then for all $g \in S$, $R(g) \ket{\psi} = \ket{\psi}$, since $\ket{c S g} = \ket{c S}$. Also, for all $g \in G \setminus S$, since $c S$ and $c S g$ are disjoint, $\braket{\psi | R(g) | \psi} = 0$. Hence, the instance $\prb{HSP}(G, f)$ reduces to the instance $\prb{BoseSL}(G, R, \ket{\psi})$ of \prb{BoseSL} with hidden symmetry subgroup $S$ and minimum asymmetry distance $\epsilon = 1$.
\end{proof}

\begin{corollary}\label{crl:efficient-algorithm-for-normal-core-hsp} 
    There is a quantum algorithm that, given a quantum oracle for $f: G \to X$ satisfying the conditions of the hidden subgroup problem with hidden subgroup $S$, finds the normal core $\core_G (S)$ of $S$ in $G$ using $O(\log \abs{\mathcal{M}} + \log(1 / \delta))$ queries to the oracle for $f$.
\end{corollary}

\begin{proof}
This follows from \cref{fct:hsp-can-be-formulated-as-state-hsp} and the fact that the state $\ket{\psi}$ in the proof of \cref{fct:hsp-can-be-formulated-as-state-hsp} can be prepared using $1$ query to the oracle for $f$. \end{proof}

\subsection{Analysis of the algorithm}

For the rest of this section, let $\mathcal{IB} = \prb{BoseSL}(G, R, \rho)$ be an instance of Bose symmetry learning with hidden symmetry subgroup $S = \HSS(\mathcal{IB})$ and minimum asymmetry distance $\epsilon = \MASD(\mathcal{IB})$.

Consider a generalisation of the strong Fourier sampling procedure defined in \cite{GSVV01quantumNonabelianHSP}: it is the same Fourier sampling procedure as in \cref{alg:state-fourier-sample}, except that we measure row and index registers $\ket{j, k}$ as well as the irrep label register $\ket{\lambda}$.

\begin{algorithm}
\caption{Strong Fourier Sampling}\label{alg:strong-state-fourier-sample}
    \begin{algorithmic}[1]
        \State Prepare the state $\frac{1}{\sqrt{\abs{G}}} \sum_{g \in G} \ket{g} \ket{\psi}$.
        \State Apply the unitary $\ctrl{R}$, giving $\frac{1}{\sqrt{\abs{G}}} \sum_{g \in G} \ket{g} \tp R(g) \ket{\psi}$.
        \State Apply the quantum Fourier transform over $G$ on the input register.
        \State Measure the irrep label, row index and column index registers $\ket{\lambda, j, k}$.
    \end{algorithmic}
\end{algorithm}

Note that performing \cref{alg:state-fourier-sample} is equivalent to performing \cref{alg:strong-state-fourier-sample} and discarding the row and column index measurements.

The following lemma gives an expression for the output probability distribution of \cref{alg:strong-state-fourier-sample}. It generalises \cite[Theorem 7]{GSVV01quantumNonabelianHSP}, as the distribution from \cite[Theorem 7]{GSVV01quantumNonabelianHSP} coincides with the distribution below when $\rho = \ketbra{\Graph(f)}{\Graph(f)}$ (which is the state we consider when working with the StateHSP formulation of the HSP, as in \cref{fct:hsp-can-be-formulated-as-state-hsp}).

\begin{lemma}\label{lmm:strong-state-fourier-sample::probability} 
    The probability of measuring an irrep label $\lambda$, row index $j$ and column index $k$ in \cref{alg:strong-state-fourier-sample} is
    \begin{equation*}
        q_{\rho, R} (\lambda, j, k) \coloneqq \frac{1}{\abs{G}} \sum_{g \in G} R_\lambda (g)_{k k} \cdot \Tr(R(g) \rho).
    \end{equation*}
    In particular, the probability does not depend on the row index $j$. The above distribution is called the \emph{strong Fourier sampling distribution}.
\end{lemma}

\begin{proof}
    Applying $\QFT_G$ in step 3 yields the state
    \begin{equation*}
        \frac{1}{\sqrt{\abs{G}}} \sum_{\lambda \in \widehat{G}} \sqrt{\frac{d_\lambda}{\abs{G}}} \sum_{j, k = 1}^{d_\lambda} \ket{\lambda, j, k} \otimes \sum_{g \in G} R(g)_{j k} R(g) \ket{\psi}.
    \end{equation*}
    Hence,
    \begin{align*}
        q_{\rho, R} (\lambda, j, k) & = \Norm{\frac{\sqrt{d_\lambda}}{\abs{G}} \sum_{g \in G} \rho_\lambda (g)_{j k} R(g) \ket{\psi}}^2 \\
        & = \frac{d_\lambda}{\abs{G}^2} \sum_{g , h \in G} \rho_\lambda (h)_{j k} \overline{\rho_\lambda (g)_{j k}} \cdot \braket{\psi | R(g)^{\dagger} R(h) | \psi} \\
        & = \frac{d_\lambda}{\abs{G}^2} \sum_{g , h \in G} \rho_\lambda (gh^{-1})_{j k} \overline{\rho_\lambda (g)_{j k}} \cdot \braket{\psi | R(g) ^{\dagger} R(gh^{-1}) | \psi} \\
        & = \frac{d_\lambda}{\abs{G}^2} \sum_{g , h \in G} \left(\rho_\lambda (g) \rho_\lambda (h)\right)_{j k} \overline{\rho_\lambda (g)_{j k}} \braket{\psi | R(h) | \psi} \\
        & = \frac{1}{\abs{G}} \sum_{\ell = 1}^{d_\lambda} \delta_{\ell k} \sum_{g' \in G} \rho_\lambda (g')_{\ell k} \braket{\psi | R(g') | \psi} \quad & \text{by \cref{fct:schur-orthogonality}} \\
        & = \frac{1}{\abs{G}} \sum_{g \in G} \rho_\lambda (g)_{k k} \cdot \braket{\psi | R(g) | \psi}.
    \end{align*}
    This shows the result for pure states; the result for mixed states follows by linearity.
\end{proof}

\begin{corollary}[{\cite[Fact 3.2]{BGTW25stateHSP}}]\label{crl:weak-state-fourier-sample::probability}
    The probability of measuring an irrep $\lambda$ in \cref{alg:state-fourier-sample} is
    \begin{equation*}
        q_{\rho, R} (\lambda) \coloneqq \frac{d_\lambda}{\abs{G}} \cdot \sum_{x \in G} \chi_\lambda (x) \cdot \Tr(R(x) \rho).
    \end{equation*}
    The above distribution is called the \emph{weak Fourier sampling distribution}, and is the marginal distribution of the irrep label $\lambda$ in the strong Fourier sampling distribution $q_{\rho, R} (\lambda, j, k)$.
\end{corollary}

\begin{proof}
    We simply sum over all $j, k \in \left[d_\lambda\right]$, and the result follows instantly.
\end{proof}

\noindent Now, we extend some of the analysis in \cite{HEC25abelianStateHSP} (\cref{lmm:fourier-sampling-distribution-mass-on-subgroup-dual,fct:fourier-sampling-distribution-anti-concentration,fct:fourier-sampling-distribution-supported-on-dual-of-symmetry-subgroup}) to the non-abelian setting. Notably, we must now often work with normal subgroups of $G$, which the hidden subgroup $S$ may not be. Thus, we consider the \textbf{normal core} $\core_G (S)$ of $S$ in $G$, which is the largest subgroup of $S$ which is normal in $G$. We also consider the \textbf{normal closure} $\ncl_G (K)$ of a subgroup $K$ in $G$, which is the smallest normal subgroup of $G$ which contains $K$.

\begin{fact}\label{lmm:bijection-between-irreps-of-quotient-group-and-trivial-irreps}
    Let $K$ be a normal subgroup of $G$. There is a natural bijection between irreps $\lambda \in \widehat{G}[K]$ and irreps of $G / K$, which preserves the degree of the representations.
\end{fact}

\begin{proof}
    Let $R_\lambda : G \rightarrow \mathcal{U}(\CC^{d_\lambda})$ be an irrep in $\widehat{G}$ with $K \leq \ker \lambda$. Define the representation $\rho'_\lambda: G / K \rightarrow \mathcal{U}(\CC^{d_\lambda})$ by $\rho'_\lambda\left(g K\right) = R_\lambda\left(g\right)$. $\rho'_\lambda$ is well-defined since $R_\lambda\left(k\right) = I$ for all $k \in K$ and $R_\lambda$ is a homomorphism. $\rho'_\lambda$ is also an irrep, which is clear from the definition of irreducibility.

    Conversely, given an irrep $R': G / K \rightarrow \mathcal{U}(\CC^{d_\lambda})$, we can define a representation $R : G \rightarrow \mathcal{U}(\CC^{d_\lambda})$ by $R(g) = R'(g K)$. Again, this is well-defined and irreducible.
\end{proof}

\begin{lemma}\label{lmm:probability-of-triv-set-under-weak-state-fourier-sampling}
    For any subgroup $K$ of $G$, $q_{\rho, R} (\widehat{G}[K]) = \frac{1}{\abs{N}} \sum_{x \in N} \Tr(R(x) \rho) = \EE_{x \in N} p_{\rho, R} (x)$, where $N = \ncl_G (K)$ is the normal closure of $K$.
\end{lemma}

\begin{proof}
    We have
    \begin{align*}
        q_{\rho, R} (\widehat{G}[K]) & = q_{\rho, R} (\widehat{G}[N]) = \sum_{\lambda \in \widehat{G}[N]} q_{\rho, R} (\lambda) = \sum_{x \in G} \left(\sum_{\lambda \in \widehat{G}[N]} \frac{d_\lambda}{\abs{G}} \chi_\lambda (x)\right) p_{\rho, R} (x).
    \end{align*}
    Now
    \begin{align*}
        \sum_{\lambda \in \widehat{G}[N]} \frac{d_\lambda}{\abs{G}} \chi_\lambda \left(x\right) & = \frac{1}{\abs{G}} \sum_{\lambda' \in \widehat{G / N}} d_{\lambda'} \chi_{\lambda'}\left(x N\right) \quad & \text{by \cref{lmm:bijection-between-irreps-of-quotient-group-and-trivial-irreps}} \\
        & = \frac{1}{\abs{G}} r_{G / N} (x N) \quad & \text{by \cref{fct:regular-representation-irrep-decomposition}} \\
        & = \frac{1}{\abs{G}} \cdot \abs{G / N} \cdot \indic\{x N = N\} \quad & \text{by \cref{fct:regular-representation-character-values}} \\
        & = \frac{1}{\abs{N}} \cdot \indic\{x \in N\},
    \end{align*}
    where $r_{G / N}$ is the character of the regular representation of $G / N$.
\end{proof}

\begin{corollary}\label{lmm:weak-state-fourier-sampling-distribution-supported-on-trivial-irreps-of-hidden-subgroup-normal-core}
    The distribution $q_{\rho, R}$ is supported on $\widehat{G}[C]$, where $C = \core_G (S)$.
\end{corollary}

\begin{proof}
    $C$ is normal so is equal to its normal closure in $G$. The result follows immediately from \cref{lmm:probability-of-triv-set-under-weak-state-fourier-sampling}, since $R(g) \rho = \rho$ for all $g \in C$.
\end{proof}

\noindent A generalisation of the anti-concentration property of \cref{fct:fourier-sampling-distribution-anti-concentration} holds:

\begin{lemma}\label{lmm:weak-state-fourier-sampling::anti-concentration-of-distribution}
    Let $K$ be a subgroup of $G$ with $K > S$. Then
    \begin{equation*}
        q_{\rho, R} (\widehat{G}[K]) \leq 1 - \frac{\epsilon}{2} .
    \end{equation*}
\end{lemma}

\begin{proof}
Let $N = \ncl_G \left(K\right)$. We have $S < K \leq N$, so $\abs{S} \leq \abs{N} / 2$ by Lagrange's theorem. Now
\begin{align*}
    q_{\rho, R} (\widehat{G}[K]) & = \Abs{\frac{1}{\abs{N}} \sum_{x \in N} \Tr(R(x) \rho)} \quad & \text{by \cref{lmm:probability-of-triv-set-under-weak-state-fourier-sampling}} \\
    & = \Abs{\frac{1}{\abs{N}} \sum_{x \in S} \Tr(R(x) \rho) + \frac{1}{\abs{N}} \sum_{x \in N \setminus S} \Tr(R(x) \rho)} \\
    & \leq \frac{1}{\abs{N}} \sum_{x \in S} \Tr(R(x) \rho) + \frac{1}{\abs{N}} \sum_{x \in N \setminus S} \abs{\Tr(R(x) \rho)} \quad & \text{by the triangle inequality} \\
    & \leq \frac{\abs{S}}{\abs{N}} + \frac{\abs{N} - \abs{S}}{\abs{N}} (1 - \epsilon) \\
    & \leq 1 - \frac{\epsilon}{2} .
    \end{align*}
\end{proof}

\noindent We are now ready to prove \cref{thm:efficient-algorithm-for-normal-core-statehsp}. 

\begin{proof}[Proof of \cref{thm:efficient-algorithm-for-normal-core-statehsp}]
    The algorithm is: obtain $m$ (a constant to be specified later) independent samples $\lambda_1, \dots, \lambda_m \in \widehat{G}$ from the weak Fourier sampling distribution $q_{\rho, R}$ by running \cref{alg:state-fourier-sample} $m$ times. Then output $A = \bigcap_{j \in [m]} \ker \lambda_j$.

    Note that by \cref{lmm:weak-state-fourier-sampling-distribution-supported-on-trivial-irreps-of-hidden-subgroup-normal-core}, $C = \core_G (S) \leq \ker R_j$ for all $j$, so $C \leq A$ (note that $A$ itself is a normal subgroup of $G$). We want to upper bound the probability that $C \neq A$ (i.e. $C < A$). We have that $C = A$ if and only if for all subgroups $K$ of $G$ with $K > C$, it holds that $K \nleq A$; equivalently, $K \nleq A$ for all subgroups $K$ of $G$ such that $C < K$ and there is no subgroup $L$ with $C < L < K$. We write $\mathcal{M'}$ to denote the collection of such subgroups $K$. Note that there is a bijection between $\mathcal{M'}$ and the set of minimal non-trivial subgroups $\mathcal{M}$ of $G / C$, given by $K \leftrightarrow K / C$. Hence, $\abs{\mathcal{M'}} = \abs{\mathcal{M}}$.

    Hence,
    \begin{equation*}
        \Pr(A \neq C) = \Pr(K \leq A \text{ for some } K \in \mathcal{M'}) \leq \sum_{K \in \mathcal{M'}} \Pr(K \leq A) \label{eq:normal-core-state-hsp::failure-probability-upper-bound}
    \end{equation*}
    by the union bound. But for $K \in \mathcal{M'}$, $K \leq A$ iff $K \leq \ker \lambda_j$ for all $j$, i.e. $\lambda_j \in \widehat{G}[K]$ for all $j$. Since the samples $\lambda_1, \ldots, \lambda_m$ are independent, we thus have $\Pr(K \leq A) \leq \left(1 - \frac{\epsilon}{2}\right)^m$ by \cref{lmm:weak-state-fourier-sampling::anti-concentration-of-distribution}.

    Therefore, $\Pr(A \neq C) \leq \abs{\mathcal{M'}} \cdot (1 - \frac{\epsilon}{2})^m \leq \abs{\mathcal{M'}} \cdot e^{-m \epsilon / 2}$ (by the inequality $1 - x \leq e^{-x}$). In order for the failure probability to upper bounded by $\delta$, it is thus sufficient to have $\abs{\mathcal{M'}} \cdot e^{-m \epsilon / 2} \leq \delta$, or equivalently,
    \begin{equation*}
        m \geq \frac{\ln \abs{\mathcal{M'}} + \ln(1 / \delta)}{\epsilon / 2} .
    \end{equation*}
\end{proof}

\subsection{Applications of the algorithm}

\subsubsection{Improved sample and time complexity for certain instances of abelian Bose symmetry learning}

 Since an upper bound for $\abs{\mathcal{M}}$ is $\abs{G} / \abs{C}$ (since each minimal subgroup of $G / C$ is generated by a single element), the sample complexity dependence of our algorithm is no higher than $\log \abs{G}$, and may be lower for certain $G$ and $C$. As an example, we obtain an improved sample complexity for solving the hidden translation problem, which is defined in \cite[Definition 12]{HEC25abelianStateHSP} as:

\begin{problem}[Hidden Translation Problem]\label{prb:hidden-translation-problem}\leavevmode
    \begin{description}
        \item[Input] Copies of a state $\ket{\psi} \in \CC^N$.
        \item[Promise] There is a $k \mid N$ such that $T^k \ket{\psi} = \ket{\psi}$, and for all $j \notin k \ZZ$, $\abs{\braket{\psi | T^j | \psi}} \leq 1 - \epsilon$, where $T$ is the translation operator defined by $T \ket{x} = \ket{x + 1}$.
        \item[Task] Find $k$.
    \end{description}
\end{problem}

\begin{corollary}
    The hidden translation problem can be solved efficiently with probability at least $1 - \delta$ using $O(\log N / \log \log N + \log(1 / \delta)) / \epsilon$ copies of $\ket{\psi}$ in the worst case, and $O(\log \log N + \log(1 / \delta)) / \epsilon$ copies of $\ket{\psi}$ for average $N$.
\end{corollary}

This is an improvement over the sample complexity upper bound of $O(\log N + \log(1 / \delta)) / \epsilon$ in \cite[Theorem 9]{HEC25abelianStateHSP}.

\begin{proof}
    It is clear that the hidden translation problem is an instance $\mathcal{I B} = \prb{BoseSL}(\ZZ_N, x \mapsto T^x, \psi)$ of \prb{BoseSL}, with hidden symmetry subgroup $S = \HSS(\mathcal{I B}) = k \ZZ_N$ and minimum asymmetry distance $\epsilon = \MASD(\mathcal{I B})$.

    The hidden subgroup $S$ is isomorphic to $\ZZ_{N / k}$, and since $\ZZ_N$ is abelian, $\core(S) = S$. So $\ZZ_N / \core(S) \cong \ZZ_k$. The number of minimal subgroups $\abs{\mathcal{M}}$ of $\ZZ_k$ is the number of distinct prime factors of $k$, denoted here as $\nu(k)$. Since $k \mid N$, $\nu(k) \leq \nu(N)$. $\nu(N) = O(\log N / \log \log N)$, and for average $N$, $\nu(N) = O(\log \log N)$. So by \cref{thm:efficient-algorithm-for-normal-core-statehsp}, we are done.
\end{proof}

\subsubsection{Bose symmetry learning over poly-near Hamiltonian groups}\label{sec:bose-symmetry-learning-over-poly-near-hamiltonian-groups}

A group $G$ is \emph{Hamiltonian} if every subgroup of $G$ is normal. Our quantum algorithm in \cref{thm:efficient-algorithm-for-normal-core-statehsp} shows that \prb{BoseSL} can be solved over any Hamiltonian group. However, the class of Hamiltonian groups is quite small: a group is Hamiltonian if and only if it is abelian or of the form $Q \times \ZZ_2^n \times A$, where $Q$ is the quaternion group of order $8$, $n \in \NN$, and $A$ is an abelian group of odd order.

An equivalent condition for $G$ being Hamiltonian is $[G: \Baer(G)] = 1$, where the \emph{Baer norm} of $G$ is defined as
\begin{equation*}
    \Baer(G) \coloneqq \bigcap_{K \leq G} N_G (K) = \{ g \in G: g K g^{-1} = K \text{ for all } K \leq G \},
\end{equation*}
i.e. the intersection of the normalisers of all subgroups of $G$. Generalising the definition of a Hamiltonian group, we say $G$ is $\emph{poly-near Hamiltonian}$ if $[G: \Baer(G)] = O(\polylog \abs{G})$.

\cite{Gav04HSP} used the normal core quantum algorithm from \cite{HRT03HSP} to give an efficient quantum algorithm for the hidden subgroup problem over any family of poly-near Hamiltonian groups, subject to the assumption that the quantum Fourier transform over certain subgroups of $G$ can be efficiently implemented by a universal gate set. Our normal core algorithm (from \cref{thm:efficient-algorithm-for-normal-core-statehsp}) for Bose symmetry learning (\prb{StateHSP}) has the same or better sample and time complexity as the algorithm from \cite{HRT03HSP} (for a given failure probability $\delta$), and so because the algorithm in \cite{Gav04HSP} uses the algorithm from \cite{HRT03HSP} as a blackbox, we can use the same argument to show that Bose symmetry learning can be solved efficiently over any family of poly-near Hamiltonian groups, by replacing blackbox calls to the algorithm from \cite{HRT03HSP} with blackbox calls to our normal core algorithm. We refer the reader to \cite{Gav04HSP} for details.

\begin{theorem}\label{thm:bose-symmetry-learning-over-poly-near-hamiltonian-groups}
    If the quantum Fourier transform over certain subgroups of $G$ (see \cite{Gav04HSP} for the descriptions of such subgroups) can be efficiently implemented, then \prb{BoseSL} can be solved efficiently over any family of poly-near Hamiltonian groups.
\end{theorem}

\section{Anyonic symmetry learning}\label{sec:relaxed-state-hsp}

In the \prb{BoseSL} problem, the Bose symmetry condition $\Tr(R(h) \rho) = 1$ is equivalent to $R(h) \rho = \rho$, i.e. $\rho$ is an eigenstate of $R(h)$ with eigenvalue $1$. However, given that in quantum mechanics the states $\ket{\psi}$ and $e^{i \theta} \ket{\psi}$ are physically indistinguishable, a reasonable definition of symmetry might be the weaker condition $\abs{\Tr(R(h) \rho)} = 1$. We refer to this kind of symmetry as \emph{anyonic symmetry}. It is equivalent to $\rho$ being an eigenstate of $R(h)$ (with no restriction on the associated eigenvalue). Thus, we introduce the \emph{anyonic symmetry learning} problem which is the anyonic symmetry analogue of \prb{BoseSL}:

\begin{problem}[Anyonic Symmetry Learning (\prb{AnyonicSL})]\label{prb:relaxed-state-hsp}\leavevmode
    \begin{description}[style=standard]
        \item[Input] Copies of a state $\rho \in \densmats{\mathcal{H}}$.
        \item[Promise] There is a finite group $G$, a unitary linear representation $R: G \rightarrow \mathcal{U}(\mathcal{H})$, a subgroup $S \leq G$, and $0 \leq \epsilon < 1$ such that: 
        \begin{itemize}
            \item $\rho$ possesses \emph{anyonic symmetry} under the action of $S$:
            \begin{equation*}
                \abs{\Tr(R(g) \rho)} = 1 \quad \forall g \in S .
            \end{equation*}
            \item For all $g \notin S$, $\rho$ is $\epsilon$-far from possessing anyonic symmetry under the action of $g$:
            \begin{equation*}
                \abs{\Tr(R(g) \rho)} \leq 1 - \epsilon \quad \forall g \in G \setminus S .
            \end{equation*}
        \end{itemize}
        \item[Task] Find (a set of generators of) the hidden symmetry subgroup $S$.
    \end{description}
    Write such an instance $\mathcal{IA}$ of \prb{AnyonicSL} as $\mathcal{IA} = \prb{AnyonicSL}(G, R, \rho)$. Write $\HSS(\mathcal{IA})$ for the associated hidden symmetry subgroup $S$ and $\MASD(\mathcal{IA})$ for the associated \emph{minimum asymmetry distance} $\epsilon$.
\end{problem}

\noindent Our goal is to reduce \prb{AnyonicSL} to \prb{BoseSL}, since there is already an efficient quantum algorithm for any instance of \prb{AbelianBoseSL} from \cite{HEC25abelianStateHSP}. This leads to the following idea:

\begin{definition}\label{def:representation-endomorphisation}
    If $R : G \rightarrow \mathcal{U}(\mathcal{H})$ is a unitary representation of an abelian group $G$,
    then so is the \textbf{endomorphisation of $R$}, $R \tp R^{-1}: G \rightarrow
    \mathcal{U}(\mathcal{H}) \tp \mathcal{U}(\mathcal{H})$, defined by
    \begin{equation*}
        (R \tp R^{-1}) (g) \coloneqq R(g) \tp R(g)^{-1} = R(g) \tp R(g^{-1}).
    \end{equation*}
\end{definition}

\noindent Suppose we have an instance $\mathcal{IA} = \prb{AnyonicSL}(G, R, \rho)$ of anyonic symmetry learning with state $\rho \in \densmats{\mathcal{H}}$, representation $R : G \rightarrow \mathcal{U}(\mathcal{H})$, hidden symmetry subgroup $S = \HSS(\mathcal{IA}) \leq G$ and minimum asymmetry distance $\epsilon = \MASD(\mathcal{IA})$.

Let $\rho' = \rho \tp \rho \in \densmats{\mathcal{H} \tp \mathcal{H}}$ and $R' = R \tp R^{-1} : G \rightarrow \mathcal{U}(\mathcal{H}) \tp \mathcal{U}(\mathcal{H})$. Then
\begin{equation*}
    \Tr(R'(x) \rho') = \Tr(R(x) \rho) \cdot \Tr(R(x^{-1}) \rho) = \Tr(R(x) \rho) \cdot \Tr(\rho^\dagger R(x)^\dagger) = \abs{\Tr(R(x) \rho)}^2,
\end{equation*}
since $\rho$ is Hermitian. Hence, this instance gives rise to an instance of BSSL with state $\rho'$, representation $R'$, subgroup $S$ and parameter $\epsilon' = 1 - (1 - \epsilon)^2 \leq 2 \epsilon$.

Moreover, we have the following:

\begin{lemma}\label{lmm:implementation-of-endomorphism-representation}
    If we can efficiently implement $\ctrl{R} = \sum_{g \in G} \ketbra{g}{g} \tp R(g)$, then we can efficiently implement $\ctrl{R \tp R^{-1}}$.
\end{lemma}

\begin{proof}
    Starting in the state $\ket{g} \ket{\psi} \ket{\phi}$, add an ancilla register whose states live in $\CC^{G}$: $\ket{g} \ket{\psi} \ket{0} \ket{\phi}$. Apply the controlled-sum unitary to the $G$-registers, giving $\ket{g} \ket{\psi} \ket{g} \ket{\psi}$. Then apply the negation gate to the third register, giving $\ket{g} \ket{\psi} \ket{-g} \ket{\phi}$. (The negation gate can be implemented by $\QFT_G^2$ for example.) Then apply $\ctrl{R}^{\tp 2}$, giving $\ket{g}\tp R(g) \ket{\psi}\tp \ket{-g} \tp R(g^{-1}) \ket{\phi}$.
    Finally, apply the negation gate again to the $\ket{- g}$ register, then discard this ancilla register, giving $\ket{g} \tp R(g) \ket{\psi} \tp R(g^{-1}) \ket{\phi}$.
\end{proof}

\noindent Thus, any instance of abelian \prb{AnyonicSL} can be reduced to an instance of abelian \prb{BoseSL}, with only a doubling in the number of copies of $\rho$ used, and only a doubling of the number of copies of $\rho$ needed at a time. In fact, we can reduce the number of copies needed, using a new quantum algorithmic primitive which we call \emph{Fourier difference sampling}:

\begin{algorithm}
\caption{Fourier Difference Sampling}\label{alg:fourier-difference-sampling}
\begin{algorithmic}[1]
    \State \textbf{Input}: $m + 1$ copies of a state $\rho$.
    \State Draw $m + 1$ Fourier samples $z'_0, z'_1, \dots, z'_m$ from the distribution $q_{\rho, R}$ using \cref{alg:state-fourier-sample} with representation $R$.
    \State \textbf{Output} $z'_1 - z'_0, \dots, z'_m - z'_0$.
\end{algorithmic}
\end{algorithm}

\begin{lemma}\label{lmm:fourier-difference-sampling-distribution}
    Each of the $m$ outputs of \cref{alg:fourier-difference-sampling} are distributed according to
    \begin{align*}
        q_{\rho \tp \rho, R \tp R^{-1}} (z) = \frac{1}{\abs{G}} \sum_{x \in G} \overline{\chi_{z} (x)} \abs{\Tr(R(x) \rho)}^{2},
    \end{align*}
    and are conditionally independent given $z'_0$.
\end{lemma}
\begin{proof}
    Each of the $m + 1$ samples $z'_0, z'_1, \dots, z'_m$ are independent and distributed according to $q_{\rho, R} = \hat{p}_{\rho, R}$. So $-z'_0$ is distributed according to
    \begin{align*}
        q_{\rho, R} (-z) & = \frac{1}{\abs{G}} \overline{\sum_{x \in G} \overline{\chi_{-z} (x)} \Tr(R(x) \rho)} \\
        & = \frac{1}{\abs{G}} \overline{\sum_{x \in G} {\chi_z (x)} \Tr(R(x) \rho)} = \frac{1}{\abs{G}} \sum_{x \in G} \overline{\chi_z (x)} \Tr(R^{-1}(x) \rho) = q_{\rho, R^{-1}} (z).
    \end{align*}
    So $z'_1 - z'_0$ is distributed according to the convolution $q_{\rho, R} * q_{\rho, R^{-1}}$ which is equal to the Fourier transform of $p_{\rho, R} \cdot p_{\rho, R^{-1}} = p_{\rho \tp \rho, R \tp R^{-1}}$ by \cref{fct:convolution-theorem}. 
\end{proof}

Using \cref{alg:fourier-difference-sampling}, we can solve any instance of abelian \prb{AnyonicSL} by solving the related abelian \prb{BoseSL} problem, while only needing one extra copy of $\rho$. Moreover, we still only require coherent access to one copy of $\rho$ needed at a time. We summarise this result in the following theorem:

\begin{theorem}\label{thm:algorithm-for-anyonic-symmetry-learning}
    Let $\mathcal{IA} = \prb{AnyonicSL}(G, R, \rho)$ be an instance of abelian \prb{AnyonicSL} with hidden symmetry subgroup $S = \HSS(\mathcal{IA})$ and minimum asymmetry distance $\epsilon = \MASD(\mathcal{IA})$. Let $\mathcal{M}$ be the set of all minimal subgroups of $G / S$, so $\abs{\mathcal{M}} \leq \abs{G}$.
    
    Then there is a quantum algorithm, using $O(\log \abs{\mathcal{M}} + \log(1 / \delta)) / \epsilon$ copies of $\rho$), which outputs a generating set of $S$ with probability at least $1 - \delta$. Moreover, if $\ctrl{R}$ and $\QFT_G$ can be implemented efficiently (in time $\polylog \abs{G}$), then the algorithm runs in time $\polylog(\abs{G}) / \epsilon$.
\end{theorem}
\begin{proof}
    The proof is almost identical to that of \cref{thm:efficient-algorithm-for-normal-core-statehsp}. The only difference is that we use \cref{alg:fourier-difference-sampling} to obtain samples from $q_{\rho \tp \rho, R \tp R^{-1}}$ instead of using \cref{alg:strong-state-fourier-sample}. Say the Fourier difference samples obtained are $z_1 = z'_1 - z'_0, \dots, z_m = z'_m - z'_0$. Instead of being unconditionally independent, $z_1, \dots, z_m$ are conditionally independent given $z'_0$. So we consider the probabilities conditioned on $z'_0$ (e.g. in \cref{eq:normal-core-state-hsp::failure-probability-upper-bound}). The same $m$ gives an upper bound of $\delta$ for the failure probability, and then we are done by the law of total probability.

    A more detailed proof which deals explicitly with the conditional independence is given in the proof of \cref{thm:algorithm-for-projective-state-hsp}.
\end{proof}

\subsection{Learning Bose symmetries when anyonic symmetries are present}

Suppose we have an instance $\prb{AnyonicSL}(G, R, \rho)$ of abelian anyonic symmetry learning, but we are only interested in learning the \emph{Bose} symmetries of $\rho$, but not the \emph{anyonic} symmetries.

Let $T \leq G$ be the subgroup of all anyonic symmetries, i.e. $T = \{g \in G: \abs{\Tr(R(g) \rho)} = 1\}$. Let $S \leq T$ be the subgroup of all Bose symmetries, i.e. $S = \{g \in G: \abs{\Tr(R(g) \rho)} = 1\}$.

\begin{lemma}\label{lmm:fourier-sampling-zero-mass-on-anyonic-subgroup-dual-for-bose-symmetry}
    The Fourier sampling distribution $q_{\rho, R}$ is supported on $S^\perp$, and has zero mass on $U^\perp$ for all $S < U \leq T$.
\end{lemma}
\begin{proof}
    The support on $S^\perp$ follows from \cref{lmm:fourier-sampling-distribution-mass-on-subgroup-dual}. The convoluted distribution $q_{\rho, R} * q_{\rho, R^{-1}}$ is supported on $T^\perp$ by \cref{lmm:fourier-difference-sampling-distribution}; hence, $q_{\rho, R}$ is also supported on $\lambda' + T^\perp$ for some $\lambda' \in \widehat{G}$. Since $T^\perp \leq S^\perp$, we must have $\lambda' \in S^\perp$.
    
    By \cref{lmm:poisson-summation-identity-on-cosets}, we have $1 = \sum_{z \in \lambda' + T^\perp} q_{\rho, R} (z) = \EE_{g \in T} \overline{\chi_{\lambda'} (g)} p_{\rho, R} (g)$, so it must be that $p_{\rho, R} (g) = \chi_{\lambda'} (g)$ for all $g \in T$. Thus, $S = T \cap \gen{\lambda'}^\perp$. For $S < U \leq T$,
    \begin{align*}
        \sum_{\lambda \in U^\perp} q_{\rho, R} (z) & = \EE_{g \in U} p_{\rho, R} (g) \quad & \text{by \cref{lmm:poisson-summation-identity}} \\
        & = \EE_{g \in U} \chi_{\lambda'} (g) \\
        & = \indic\{\lambda' \in U^\perp\} \quad & \text{by \cref{lmm:subgroup-orthogonality-relations}} \\
        & = \indic\{U \leq \gen{\lambda'}^\perp\} = \indic\{U \leq S\} = 0.
    \end{align*}
\end{proof}

\begin{lemma}\label{lmm:fourier-sampling-anti-concentration-for-bose-vs-anyonic-symmetry}
    For any $K > S$, $q_{\rho, R} (K^\perp) \leq 1 - \epsilon / 2$.
\end{lemma}
\begin{proof}
    By \cref{lmm:fourier-sampling-distribution-mass-on-subgroup-dual}, $q_{\rho, R} (K^\perp) = \EE_{g \in K} p_{\rho, R} (g)$. Suppose $K \cap T > S$. Then since $K \cap T \leq T$, we have $q_{\rho, R} ((K \cap T)^\perp) = 0$ by \cref{lmm:fourier-sampling-zero-mass-on-anyonic-subgroup-dual-for-bose-symmetry}. Since $K^\perp \subseteq (K \cap T)^\perp$, also $q_{\rho, R} (K^\perp) = 0$. Otherwise, if $K \cap T = S$, then since $(K \setminus S) \cap T = \emptyset$, we have $\abs{p_{\rho, R} (x)} \leq 1 - \epsilon$ for all $x \in K \setminus S$. Hence,
    \begin{equation*}
        q_{\rho, R} (K^\perp) = \frac{1}{\abs{K}} \left( \sum_{x \in K \cap S} p_{\rho, R} (x) + \sum_{x \in K \setminus S} p_{\rho, R} (x) \right) \leq \frac{\abs{S}}{\abs{K}} + \frac{\abs{K} - \abs{S}}{\abs{K}} (1 - \epsilon) \leq 1 - \epsilon / 2,
    \end{equation*}
    since $\abs{K} \geq 2 \abs{S}$ by Lagrange's theorem.
\end{proof}

\noindent \cref{lmm:fourier-sampling-zero-mass-on-anyonic-subgroup-dual-for-bose-symmetry} is analogous to \cref{lmm:weak-state-fourier-sampling-distribution-supported-on-trivial-irreps-of-hidden-subgroup-normal-core} and \cref{lmm:fourier-sampling-anti-concentration-for-bose-vs-anyonic-symmetry} is analogous to \cref{lmm:weak-state-fourier-sampling::anti-concentration-of-distribution}, and the proof of \cref{thm:efficient-algorithm-for-normal-core-statehsp} only depends on \cref{lmm:weak-state-fourier-sampling-distribution-supported-on-trivial-irreps-of-hidden-subgroup-normal-core,lmm:weak-state-fourier-sampling::anti-concentration-of-distribution}. Hence, we can relax the asymmetry condition in abelian Bose symmetry learning to $\abs{\Tr(R(x) \rho)} \leq 1 - \epsilon$ OR $\Tr(R(x) \rho) \in S^1 \setminus \{1\}$ for all $x \in G \setminus S$, and the algorithm in \cref{thm:efficient-algorithm-for-normal-core-statehsp} still produces the same output. Thus, we redefine abelian \prb{Bose} symmetry learning (abelian \prb{StateHSP}) as follows:

\begin{problem}[Abelian \prb{BoseSL} --- Redefined]\label{prb:redefined-abelian-bosesl}\leavevmode
    \begin{description}[style=standard]
        \item[Input] Copies of a state $\rho \in \densmats{\mathcal{H}}$.
        \item[Promise] There is a finite abelian group $G$, a unitary representation $R: G \rightarrow \mathcal{U}(\mathcal{H})$, a subgroup $S \leq G$, and $0 \leq \epsilon < 1$ such that: 
        \begin{itemize}
            \item $\rho$ possesses Bose symmetry under the action of $S$:
            \begin{equation*}
                \Tr(R(g) \rho) = 1 \quad \forall g \in S .
            \end{equation*}
            \item For all $g \notin S$, either $\rho$ is $\epsilon$-far from possessing Bose symmetry under the action of $g$, or $\rho$ possesses anyonic, non-Bose symmetry under the action of $g$:
            \begin{equation*}
                \abs{\Tr(R(g) \rho)} \leq 1 - \epsilon \quad \text{or} \quad \Tr(R(g) \rho) \in S^1 \setminus \{1\} \quad \forall g \in G \setminus S .
            \end{equation*}
        \end{itemize}
        \item[Task] Find (a set of generators of) the hidden symmetry subgroup $S$.
    \end{description}
\end{problem}

\noindent \cref{lmm:fourier-sampling-zero-mass-on-anyonic-subgroup-dual-for-bose-symmetry,lmm:fourier-sampling-anti-concentration-for-bose-vs-anyonic-symmetry}, together with the proof of \cref{thm:efficient-algorithm-for-normal-core-statehsp}, give the following:

\begin{theorem}
    Abelian \prb{BoseSL}, in its new formulation above, can be solved in $\polylog \abs{G}$ time with probability at least $1 - \delta$ using $O(\log \abs{\mathcal{M}} + \log(1 / \delta)) / \epsilon$ copies of $\rho$, where $\mathcal{M}$ is the set of all minimal subgroups of $G / S$.
\end{theorem}

\subsection{Anyonic symmetry learning over non-abelian groups}\label{sec:non-abelian-anyonic-symmetry-learning}

Unlike for representations of abelian groups, the endomorphisation of a representation $R$ of a non-abelian group, $R \tp R^{-1}$, is not necessarily a representation, and so we cannot use the endomorphisation trick (or Fourier difference sampling) to reduce \prb{AnyonicSL} over $G$ to \prb{BoseSL} over $G$ for non-abelian groups $G$. However, we show in this section that we can still reduce \prb{AnyonicSL} over $G$ to \prb{BoseSL} over a larger group which contains $G$.

Let $G$ be a finite group with exponent $E = \exp(G)$ and $R: G \to \mathcal{U}(\mathcal{H})$ be a linear representation of $G$. For all $g \in G$, we have
\begin{equation*}
    R(g)^E = R(g^E) = R(e_G) = I_{\mathcal{H}}.
\end{equation*}
Hence, the eigenvalues of $R(g)$ are all $E$-th roots of unity, i.e. powers of $\omega = e^{2 \pi i / E}$. Define the following representation $R': G \times \ZZ_E \to \mathcal{U}(\mathcal{H})$ of $G \times \ZZ_E$ by
\begin{equation*}
    R'(g, k) = \omega^{-k} R(g) \quad \forall g \in G, k \in \ZZ_E.
\end{equation*}

Now consider an instance $\mathcal{IA} = \prb{AnyonicSL}(G, R, \rho)$ of anyonic symmetry learning with hidden symmetry subgroup $S = \HSS(\mathcal{IA})$ and minimum asymmetry distance $\epsilon = \MASD(\mathcal{IA})$. For each $g \in S$, we have $\abs{\Tr(R(g) \rho)} = 1$, so there exists $k_g \in \ZZ_E$ such that $\Tr(R(g) \rho) = \omega^{k_g}$. We must have $k_{g + h} = k_g + k_h$ for all $g, h \in S$, since $R(g + h) \rho = R(g) R(h) \rho = \omega^{k_g + k_h} \rho$. Therefore the set $S' \coloneq \{ (g, k_g): g \in S \}$ is a subgroup of $G \times \ZZ_E$. It is clear that $R'(g, k_g) \rho = \rho$ for all $(g, k_g) \in S'$, and for all $(g, k) \notin S'$, $\abs{\Tr(R'(g, k) \rho)} \leq 1 - \epsilon$. Since $S = \pi_G (S')$ is the projection of $S'$ onto $G$, we can recover a set of generators for $S$ from a set of generators for $S'$, by taking the $G$-component of each generator.

Thus, we have proven the following:

\begin{theorem}\label{thm:reduction-of-non-abelian-anyonic-symmetry-learning-to-bose-symmetry-learning}
    Any instance $\mathcal{IA} = \prb{AnyonicSL}(G, R, \rho)$ of anyonic symmetry learning over a finite group $G$ can be reduced to an instance $\mathcal{IA}' = \prb{BoseSL}(G \times \ZZ_E, R', \rho)$ of Bose symmetry learning over the group $G \times \ZZ_E$.
\end{theorem}

Note that this reduction also allows us to learn Bose symmetries when other anyonic symmetries are present: if $T$ is the set of Bose symmetries of $\rho$ (i.e. $T = \{g \in G: \Tr(R(g) \rho) = 1\}$), then $T = \{g \in G: (g, 0) \in S'\}$, so we can recover a set of generators for $T$ from a set of generators for $S'$ by taking the $G$-component of all generators whose $\ZZ_E$-component is $0$.

\cref{thm:reduction-of-non-abelian-anyonic-symmetry-learning-to-bose-symmetry-learning} shows that if $G \times \ZZ_E$ admits an efficient quantum algorithm for \prb{BoseSL}, then $G$ admits an efficient quantum algorithm for \prb{AnyonicSL}. From \cref{thm:bose-symmetry-learning-over-poly-near-hamiltonian-groups}, we know that any poly-near Hamiltonian group admits an efficient quantum algorithm for \prb{BoseSL}. The following lemma gives a condition under which $G \times \ZZ_E$ is poly-near Hamiltonian, and hence \prb{SubsetBoseSL} is efficiently solvable. First, we introduce a definition similar in spirit to the definition of a poly-near Hamiltonian group.

\begin{definition}\label{def:poly-near-abelian-group}
    A finite group $G$ is \emph{poly-near abelian} if $[G: Z(G)] = O(\polylog \abs{G})$, where $Z(G) \coloneq \{ g \in G: g h = h g \text{ for all } h \in G \}$ is the \emph{center} of $G$.
\end{definition}

\begin{lemma}
    Suppose $G$ is a poly-near abelian group. Then $G \times \ZZ_E$, where $E = \exp(G)$, is poly-near Hamiltonian.
\end{lemma}
\begin{proof}
    Recall from \cref{sec:bose-symmetry-learning-over-poly-near-hamiltonian-groups} that $G$ is poly-near Hamiltonian if $[G: \Baer(G)] = O(\polylog \abs{G})$, where $\Baer(G) = \{ g \in G: g K g^{-1} = K \text{ for all } K \leq G \}$ is the Baer norm of $G$. Note that $Z(\Gamma) \leq \Baer(\Gamma)$ for any group $\Gamma$, and $Z(\Gamma_1 \times \Gamma_2) = Z(\Gamma_1) \times Z(\Gamma_2)$ for any groups $\Gamma_1, \Gamma_2$. Hence,
    \begin{equation*}
        Z(G) \times \ZZ_E = Z(G) \times Z(\ZZ_E) = Z(G \times \ZZ_E) \leq \Baer(G \times \ZZ_E).
    \end{equation*}
    This gives
    \begin{align*}
        [(G \times \ZZ_E) : \Baer(G \times \ZZ_E)] & = \frac{\abs{G} \cdot E}{\abs{\Baer(G \times \ZZ_E)}} \leq E \frac{\abs{G}}{\abs{Z(G \times \ZZ_E)}} \\
        & = [G: Z(G)] = O(\polylog \abs{G}).
    \end{align*}
\end{proof}

\begin{corollary}
    There is an efficient algorithm for \prb{AnyonicSL} over any poly-near abelian group $G$, provided that certain quantum Fourier transforms over $G$ can be performed efficiently.
\end{corollary}

\section{Projective anyonic symmetry learning}\label{sec:projective-state-hsp}

The Bose and anyonic symmetry learning problems encapsulate many problems that involve finding hidden symmetries of a quantum state. In both \prb{BoseSL} and \prb{AnyonicSL}, we have so far required the representation $R$ to be \emph{linear}. However, there are often physically realistic symmetries that are described by \emph{projective} representations instead.

One common example of a projective representation in quantum information is the Pauli representation of $\ZZ_d^{2n}$ given by $R_W (x) = \tilde{W}_x$, where $\tilde{W}_x$ is the \emph{phaseless Weyl operator} defined by $\tilde{W}_{(a, b)} = X^a Z^b$ for $a, b \in \ZZ_d^n$. Given a state $\rho$, the subgroup $S \leq \ZZ_d^{2n}$ of all $x \in \ZZ_d^{2n}$ such that $\abs{\Tr(R_W (x) \rho)} = 1$ is called the \emph{(phaseless) stabiliser group} of $\rho$. Learning the stabiliser group of \emph{pure states} has been studied in \cite{Mon17,HEC25abelianStateHSP,ADIS25quditBellSampling}; however, the mixed state case has not been studied.

Motivated by this, we introduce the \emph{projective anyonic symmetry learning} problem, which is the same as the anyonic symmetry learning problem, except now we allow for projective representations:

\begin{definition}[\prb{ProjAnyonicSL}]\leavevmode\label{prb:projective-state-hsp}
    \begin{description}
        \item[Input] Copies of a state $\rho \in \densmats{\mathcal{H}}$.
        \item[Promise] There is a finite group $G$, a \emph{projective} unitary representation
        $R: G \rightarrow \mathcal{U}(\mathcal{H})$, a subgroup $S \leq G$ and a constant $0 < \epsilon \leq 1$ such that:
        \begin{itemize}
            \item $\rho$ possesses \emph{anyonic symmetry} under the action of $S$:
            \begin{equation}
                \abs{\Tr(R(h) \rho)} = 1 \quad \forall h \in S.
            \end{equation}
            \item For all $x \notin S$, $\rho$ is $\epsilon$-far from being anyonic-symmetric under the action of $x$:
            \begin{equation}
                \abs{\Tr(R(g) \rho)} \leq 1 - \epsilon \quad
                \forall g \in G \setminus S.
            \end{equation}
        \end{itemize}
        \item[Task] Find (a set of generators of) the hidden symmetry subgroup $S$.
    \end{description}
    Write such an instance $\mathcal{IA}$ of \prb{ProjAnyonicSL} as $\mathcal{IA} = \prb{ProjAnyonicSL}(G, R, \rho)$. Write $\HSS(\mathcal{IA})$ for the associated hidden symmetry subgroup $S$ and $\MASD(\mathcal{IA})$ for the associated \emph{minimum asymmetry distance} $\epsilon$.
\end{definition}

\noindent For the rest of this section, fix an instance $\mathcal{IA} = \prb{ProjAnyonicSL}(G, R, \rho)$ of \prb{ProjAnyonicSL} with $G$ abelian, $R: G \to \mathcal{U}(\mathcal{H})$ a projective representation, and hidden symmetry subgroup $S = \HSS(\mathcal{IA})$ and minimum asymmetry distance $\epsilon = \MASD(\mathcal{IA})$.

\subsection{Linearisation of projective representations via linear codes}\label{sec:linearisation-of-projective-representations}

\cref{alg:state-fourier-sample} samples from $q_{\rho, R} (\lambda)$ which is a valid probability distribution. Unfortunately, when $R$ is a projective representation, $q_{\rho, R}$ is in general no longer a probability distribution. For example, consider the projective representation $R_W (a, b) = Z^a X^b$ of $\ZZ_2^2$, and the state $\ket{\psi} = \frac{1}{2} \ket{0} - \frac{\sqrt{3}}{2} \ket{1}$; we have $q_{\rho, R_W} (0) = \frac{1}{8} (1 - \sqrt{3}) < 0$. 

As we will see however, we are able to sample from a similar function which is a valid probability distribution, using a new technique of \emph{linearising} projective representations.

By \cref{fct:abelian-projective-representation-cocycle-form}, up to a gauge transformation $f: G \to S^1$, we may assume that $R(x) R(y) = \omega^{B(x, y)} R(x + y)$ for all $x, y \in G$, where $\omega$ is a primitive $E$-th root of unity, $E$ is the exponent of $G$, and $B: G \times G \to \ZZ_E$ is a bilinear form. In particular, $R(0)^2 = \omega^{B(0, 0)} R(0) = R(0)$, so $R(0) = I$. Here, we assume that the gauge transformation $f$ is known and the unitary $\ket{g} \mapsto f(g) \ket{g}$ is efficiently implementable, and so we can implement the equivalent projective representation $R'(x) = f(x) R(x)$ efficiently, given an efficient implementation of $R$.

We want to construct a linear representation $R_\text{lin}: G^s \to \mathcal{U} (\mathcal{H}^{\tp t})$ from $R: G \to \mathcal{U}(\mathcal{H})$, where $s$ and $t$ are positive integers. This will then allow us to perform Fourier sampling on $\rho^{\tp t}$ using $R_\text{lin}$.

With this in mind, we want $R_\text{lin}$ to satisfy
\begin{align}
    \abs{\Tr(R_\text{lin}(x_1, \dots, x_s) \rho^{\tp t})} & = 1 \quad \forall x_1, \dots, x_s \in S \tag{$\dagger$} \label{eq:multi-subgroup-invariance} \\
    \abs{\Tr(R_\text{lin}(y_1, \dots, y_s)) \rho^{\tp t})} & \ne 1 \quad \forall y_1, \dots, y_s \in G \text{ such that at least one } y_i \notin S. \tag{$*$} \label{eq:multi-subgroup-non-invariance}
\end{align}

Then, by \cref{lmm:fourier-difference-sampling-distribution,fct:fourier-sampling-distribution-supported-on-dual-of-symmetry-subgroup}, performing Fourier difference sampling on $\rho^{\tp t}$ with representation $R_\text{lin}$ will only yield elements of $(S^s)^\perp = (S^\perp)^s$. Given enough samples from $(S^\perp)^s$, we can recover $S^\perp$, and so also $S$.

Further, if we require that
\begin{equation}
    \Tr(R_\text{lin}(x_1, \dots, x_s) \rho^{\tp t}) = 1 \quad \forall x_1, \dots, x_s \in S, \tag{$\ddagger$} \label{eq:multi-subgroup-exact-invariance}
\end{equation}
then by \cref{fct:fourier-sampling-distribution-supported-on-dual-of-symmetry-subgroup} we can perform Fourier sampling with $R_\text{lin}$ to obtain elements of $(S^\perp)^s$ (instead of Fourier difference sampling).

In order to satisfy \cref{eq:multi-subgroup-invariance}, since $S$ is a subgroup, it suffices to restrict $R_\text{lin}$ to be of the form
\begin{align*}
    R_\text{lin}(\vd{x}) & = R({a_{1 1} x_1 + \cdots + a_{s 1} x_s}) \tp \cdots \tp R({a_{1 t} x_1 + \cdots + a_{s t} x_s}) \\
    & = R((\vd{x} M)_1) \tp \cdots \tp R((\vd{x} M)_t),
\end{align*}
where $a_{k \ell} \in \ZZ$ for all $k, \ell$ and $M \in \ZZ^{s \times t}$ is defined by its entries $M_{k \ell} = a_{k \ell}$. So for each $M \in \ZZ^{s \times t}$, we associate the corresponding map $R_M: G^s \to \mathcal{U} (\mathcal{H}^{\tp t})$ defined by
\begin{equation*}
    R_M (\vd{x}) = R((\vd{x} M)_1) \tp \cdots \tp R((\vd{x} M)_t).
\end{equation*}

The following lemma characterises all $M \in \ZZ^{s \times t}$ such that $R_M$ is a linear representation.

\begin{lemma}\label{lmm:characterisation-of-bell-sampling-representations}
    $R_M$ is a linear representation if
    \begin{equation}
        M M^T = 0 \pmod{E}, \tag{$M\star$} \label{eq:valid-bell-sampling-matrix}
    \end{equation}
    i.e. all rows of $M$ are mutually orthogonal modulo $E$. Moreover, this condition is necessary if there exist $x^*, y^* \in G$ such that $B(x^*, y^*) = 1$.
\end{lemma}

\begin{proof}
    For all $\vd{x}, \vd{y} \in G^s$,
    \begin{align*}
        R_M (\vd{x}) R_M (\vd{y}) & = R((\vd{x} M)_1) R((\vd{y} M)_1) \tp \cdots \tp R((\vd{x} M)_t) R((\vd{y} M)_t) \\
        & = \omega^{B((\vd{x} M)_1, (\vd{y} M)_1)} \cdots \omega^{B((\vd{x} M)_t, (\vd{y} M)_t)} R(((\vd{x} + \vd{y}) M)_1) \tp \cdots \tp R(((\vd{x} + \vd{y}) M)_t) \\
        & = \omega^{B((\vd{x} M)_1, (\vd{y} M)_1) + \cdots + B((\vd{x} M)_t, (\vd{y} M)_t)} R_M (\vd{x} + \vd{y})
    \end{align*}
    Hence, $R$ is a representation if and only if
    \begin{align*}
        0 & = B((\vd{x} M)_1, (\vd{y} M)_1) + \cdots + B((\vd{x} M)_t, (\vd{y} M)_t) \\
        & = \sum_{\ell = 1}^t \sum_{j, k = 1}^s a_{j \ell} a_{k \ell} B(x_j, y_k) \\
        & = \sum_{j, k = 1}^s (M M^T)_{j k} B(x_j, y_k) \pmod{E} \quad \forall x, y \in G^s,
    \end{align*}
    The equation above holds if \myrefeq{Condition}{eq:valid-bell-sampling-matrix} holds. The necessity condition can be seen to hold by taking $x_j = x^*$ and $y_k = y^*$, and all other $x_{j'}$ and $y_{k'}$ to be zero.
\end{proof}

\begin{corollary}\label{crl:product-form-of-linearisation-representation}
    If $R_M$ is a representation, then
    \begin{equation*}
        R_M (x_1, \dots, x_s) = \bigotimes_{\ell = 1}^t R(a_{1 \ell} x_1) \cdots R(a_{s \ell} x_s).
    \end{equation*}
\end{corollary}
\begin{proof}
    We have $R(a_{1 \ell} x_1 + \cdots + a_{s \ell} x_s) = \omega^{\sum_\ell \sum_{j < k} a_{j \ell} a_{k \ell} B(x_j, x_k)} R(a_{1 \ell} x_1) \cdots R(a_{s \ell} x_s)$. The power of $\omega$ is $\sum_{j < k} (M M^T)_{j k} B(x_j, x_k) = 0$ by \cref{lmm:characterisation-of-bell-sampling-representations}.
\end{proof}

\cref{lmm:characterisation-of-bell-sampling-representations} gives a sufficient criterion for $M$ to induce a valid representation; however, we need to impose further conditions on $M$ in order for $R_M$ to satisfy the condition in \cref{eq:multi-subgroup-non-invariance}. For example, the all-zeros matrix trivially satisfies \myrefeq{Condition}{eq:valid-bell-sampling-matrix}, but will induce the trivial representation, which does not satisfy \cref{eq:multi-subgroup-non-invariance}.

\begin{lemma}\label{lmm:characterisation-of-useful-bell-sampling-representations}
    The property
    \begin{equation}
        \forall x_1, \dots, x_s \in G, \left(\forall \ell, (\vd{x} M)_\ell \in S\right) \quad \Longrightarrow \quad \left(\forall k, x_k \in S\right) 
    \end{equation}
    is equivalent to $R_M$ satisfying \cref{eq:multi-subgroup-non-invariance}, which is satisfied if the following condition holds:
    \begin{equation*}
        \text{$M$ can be put into standard form by swapping rows and columns, and adding rows} \tag{$M*$} \label{eq:useful-bell-sampling-matrix-condition}
    \end{equation*}
    By standard form, we mean of the form $(I_s \mid A)$ with $A \in \ZZ^{s \times (t - s)})$. Note that adding rows implies we are also able to multiply any row by an integer.
\end{lemma}

\begin{proof}
    The equivalence is clear. Being able to put $M$ into standard form by the above operations means we can solve we can determine $x_1, \dots, x_s$ from $(\vd{x} M)_1, \dots, (\vd{x} M)_t$ by standard Gaussian elimination techniques. Due to the linearity of each operation used in Gaussian elimination and the closure property of $S$, this means that if $(\vd{x} M)_\ell \in S$ for all $\ell \in [t]$, then $x_k \in S$ for all $k \in [s]$.
\end{proof}

If $R_M$ satisfies \cref{eq:multi-subgroup-invariance} and \cref{eq:multi-subgroup-non-invariance}, then as mentioned above we may use the Fourier difference sampling approach to recover $S$. However, we could also enforce the following condition on $M$ which will mean that $R_M$ satisfies \cref{eq:multi-subgroup-exact-invariance}, which allows us to use Fourier sampling instead of Fourier difference sampling to recover $S$.

\begin{lemma}\label{lmm:condition-for-zero-sum-property}
    Let $M$ satisfy \myrefeq{Condition}{eq:valid-bell-sampling-matrix} and \myrefeq{Condition}{eq:useful-bell-sampling-matrix-condition}. $R_M$ satisfies the property in \cref{eq:multi-subgroup-exact-invariance} if $M$ satisfies:
    \begin{align*}
        \text{$E$ is odd} & \implies \sum_{\ell = 1}^t a_{k \ell} = 0 \pmod{E} \text{ for all } k, \tag{$M\ddagger$} \\
        \text{$E$ is even} & \implies \sum_{\ell = 1}^t a_{k \ell} = 0 \pmod{E} \text{ for all } k \quad \text{and} \quad \sum_{\ell = 1}^t a_{k \ell}^2 = 0 \pmod{2E} \text{ for all } k, \tag{$M\ddagger$} \label{eq:zero-sum-property}
    \end{align*}
    i.e. the sum of the entries of each row of $M$ is zero modulo $E$, and if $E$ is even then each row of $M$ is \emph{isotropic} (self-orthogonal) modulo $2E$.
\end{lemma}
\begin{proof}
    For $x \in G$, we have 
    \begin{align*}
        R(x)^E & = \omega^{\frac{E(E - 1)}{2} B(x, x)} R(E x) = (-1)^{\indic\{2 \mid E\} \cdot B(x, x)} R(E x) \\
        & = (-1)^{\indic\{2 \mid E\} \cdot B(x, x)} R(0) = (-1)^{\indic\{2 \mid E\} \cdot B(x, x)} I.
    \end{align*}
    
    First suppose $E$ is odd. Then the eigenvalues of $R(x)$ are $E$-th roots of unity. Let $x_1, \dots, x_s \in S$ and $R(x_k) \rho = \omega^{b_k} \rho$ for each $k$, with each $b_k \in \ZZ_E$. Then
    \begin{align*}
        R(a_{1 \ell} x_1) \cdots R(a_{s \ell} x_s) \rho = \omega^{-\sum_{k = 1}^s \binom{a_{k \ell}}{2} B(x_k, x_k)} \omega^{\sum_{k = 1}^s a_{k \ell} b_k} \rho,
    \end{align*}
    therefore by \cref{crl:product-form-of-linearisation-representation},
    \begin{align*}
        R_M (x_1, \dots, x_s) \rho^{\tp t} & = \omega^{-\sum_{k = 1}^s \sum_{\ell = 1}^t \binom{a_{k \ell}}{2} B(x_k, x_k)} \omega^{\sum_{k = 1}^s \sum_{\ell = 1}^t a_{k \ell} b_k} \rho^{\tp t}.
    \end{align*}
    So it suffices that $\sum_{\ell = 1}^t a_{k \ell} = 0 \pmod{E}$ for all $k$ and that $\sum_{\ell = 1}^t \binom{a_{k \ell}}{2} = 2^{-1} \sum_{\ell = 1}^t a_{k \ell}^2 - 2^{-1} \sum_{\ell = 1}^t a_{k \ell} = 0 \pmod{E}$ for all $k$, where $2^{-1}$ is the multiplicative inverse of $2$ in $\ZZ_E$ (which exists since $E$ is odd). We already have that $\sum_{\ell = 1}^t a_{k \ell}^2 = 0 \pmod{E}$ for all $k$ by \cref{lmm:characterisation-of-bell-sampling-representations}. Hence, the condition $\sum_{\ell = 1}^t a_{k \ell} = 0 \pmod{E}$ for all $k$ is sufficient.

    Now suppose $E$ is even. For the rest of the proof, we consider each $B(x, y)$ to be embedded in $\ZZ$ (the choice of embedding does not matter). Let $\tau = e^{2 \pi i / 2E}$. Then the eigenvalues of $R(x)$ are $E$-th roots of unity (i.e. even powers of $\tau$) if $B(x, x)$ is even, and odd powers of $\tau$ if $B(x, x)$ is odd. Let $x_1, \dots, x_s \in S$ and $R(x_k) \rho = \tau^{b_k} \rho$ for each $k$, with each $b_k \in \ZZ$. Importantly, for each $k$, $b_k = B(x_k, x_k) \pmod{2}$, so $b_k + B(x_k, x_k) = 2 c_k$ for some $c_k \in \ZZ$. Then
    \begin{align*}
        R(a_{1 \ell} x_1) \cdots R(a_{s \ell} x_s) \rho & = \omega^{-\sum_{k = 1}^s \binom{a_{k \ell}}{2} B(x_k, x_k)} \tau^{\sum_{k = 1}^s a_{k \ell} b_k} \rho \\
        & = \tau^{-\sum_{k = 1}^s a_{k \ell}^2 B(x_k, x_k)} \tau^{\sum_{k = 1}^s a_{k \ell} B(x_k, x_k)} \tau^{\sum_{k = 1}^s a_{k \ell} b_k} \rho \\
        & = \tau^{-\sum_{k = 1}^s a_{k \ell}^2 B(x_k, x_k)} \omega^{\sum_{k = 1}^s a_{k \ell} c_k} \rho
    \end{align*}
    therefore by \cref{crl:product-form-of-linearisation-representation},
    \begin{align*}
        R_M (x_1, \dots, x_s) \rho^{\tp t} & = \tau^{-\sum_{k = 1}^s \sum_{\ell = 1}^t a_{k \ell}^2 B(x_k, x_k)} \omega^{\sum_{k = 1}^s \sum_{\ell = 1}^t a_{k \ell} c_k} \rho^{\tp t}.
    \end{align*}
    It is clearly that the stated conditions suffice to make this expression equal to $\rho^{\tp t}$ for all $x_1, \dots, x_s \in S$.
\end{proof}

The following lemma shows that the above conditions on $M$ are preserved under row and column operations on $M$.

\begin{lemma}\label{lmm:row-and-column-operations-preserve-bell-sampling-matrix-conditions}
    Multiplying rows of $M$ by scalars in $\ZZ_E$, adding rows, swapping rows and swapping columns preserves \myrefeq{Condition}{eq:valid-bell-sampling-matrix} and \myrefeq{Condition}{eq:useful-bell-sampling-matrix-condition}. These operations also preserve \myrefeq{Condition}{eq:zero-sum-property}, assuming \myrefeq{Condition}{eq:valid-bell-sampling-matrix} and \myrefeq{Condition}{eq:useful-bell-sampling-matrix-condition} already hold.
\end{lemma}

\begin{proof}
    Mutual orthogonality of rows (\myrefeq{Condition}{eq:valid-bell-sampling-matrix}) are clearly preserved, as is whether or not $M$ can be put into standard form (\myrefeq{Condition}{eq:useful-bell-sampling-matrix-condition}) and the zero-sum property (part of \myrefeq{Condition}{eq:zero-sum-property}).

    If $E$ is even and we assume that $M$ satisfies \myrefeq{Condition}{eq:valid-bell-sampling-matrix} and \myrefeq{Condition}{eq:useful-bell-sampling-matrix-condition}, then for two rows $x$ and $y$ of $M$, we have $\gen{x + y, x + y} = \gen{x, x} + \gen{y, y} + 2 \gen{x, y}$. We have assumed that $\gen{x, x} = \gen{y, y} = 0$ modulo $2E$ and $\gen{x, y} = 0$ modulo $E$, so $\gen{x + y, x + y} = 0$ modulo $2E$. Hence, the isotropicity of rows is preserved under adding rows, so is also preserved under multiplying rows by scalars in $\ZZ_E$. Swapping rows and columns clearly preserves isotropicity of rows.
\end{proof}

\myrefeq{Condition}{eq:valid-bell-sampling-matrix}, \myrefeq{Condition}{eq:useful-bell-sampling-matrix-condition} and \myrefeq{Condition}{eq:zero-sum-property} can be viewed from a coding-theoretic perspective:

\begin{theorem}\label{thm:linearisation-matrix-relation-to-linear-codes}
    Let $M \in \ZZ_E^{s \times t}$ be a matrix satisfying \myrefeq{Condition}{eq:valid-bell-sampling-matrix} and \myrefeq{Condition}{eq:useful-bell-sampling-matrix-condition}. Then $M$ is a generator matrix of a free, self-orthogonal linear code $C \subseteq \ZZ_E^t$ over $\ZZ_E$ of block length $t$ and dimension $s$. If $M$ also satisfies \myrefeq{Condition}{eq:zero-sum-property}, then $C$ is a zero-sum code which is also isotropic modulo $2E$ if $E$ is even.
\end{theorem}
\begin{proof}
    The self-orthogonality follows from \myrefeq{Condition}{eq:valid-bell-sampling-matrix}, the linear independence of rows follows from \myrefeq{Condition}{eq:useful-bell-sampling-matrix-condition}, and the zero-sum property follows from \myrefeq{Condition}{eq:zero-sum-property}. The fact that $C$ is free also follows from \myrefeq{Condition}{eq:useful-bell-sampling-matrix-condition}. \cref{lmm:row-and-column-operations-preserve-bell-sampling-matrix-conditions} shows that the zero-sum, orthogonality and isotropicity properties extend from the rows of the generator matrix to all codewords.
\end{proof}

\begin{notation}
    When $R$ is clear from the context, write $p_{\rho, M} = p_{\rho, R_M}$ and $q_{\rho, M} = q_{\rho, R_M}$ for all $M$ such that $R_M$ is a representation.
\end{notation}

\begin{notation}
    Write $M_k$ for the $k$th row of $M$. If $M$ satisfies any of \myrefeq{Condition}{eq:valid-bell-sampling-matrix}, \myrefeq{Condition}{eq:useful-bell-sampling-matrix-condition}, \myrefeq{Condition}{eq:zero-sum-property}, then trivially so does $M_k$ for any $k$.
\end{notation}

\begin{lemma}
    The marginal distribution of $q_{\rho, M}$ on the $k$th coordinate is $q_{\rho, M_k}$ for all $k$.
\end{lemma}
\begin{proof}
    For $k = 1$,
    \begin{align*}
        \sum_{x_2, \dots, x_s \in G} q_{\rho, M} (x_1, \dots, x_s) & = \frac{1}{\abs{G}^s} \sum_{z_1, \dots, z_s \in G} \sum_{x_2, \dots, x_s \in G} \overline{\chi_{z_1}(x_1) \cdots \chi_{z_s}(x_s)} p_{\rho, M} (x_1, \dots, x_s) \\
        & = \frac{1}{\abs{G}} \sum_{z_1 \in G} \overline{\chi_{z_1}(x_1)} p_{\rho, M} (x_1, 0, \dots, 0) \\
        & = q_{\rho, M_1} (x_1).
    \end{align*}
    The proof is exactly the same for $k > 1$.
\end{proof}

\begin{lemma}
    Given an oracle to $\ctrl{R}$, we can implement $\ctrl{R_M}$ using $t$ queries to $\ctrl{R}$ and $O(1)$ ancillae and depth $O(s t)$, or $\Theta(t)$ ancillae and depth $O(s)$.
\end{lemma}
\begin{proof}
    The proof is similar to that of \cref{lmm:implementation-of-endomorphism-representation}: we use an ancillary registers to compute the linear combinations of the input group elements specified by $M$, and then apply $\ctrl{R}$ to the resulting group elements in parallel, before uncomputing the ancillary registers.

    We may either compute the linear combinations one at a time (compute, perform $\ctrl{R}$, uncompute, repeat with next linear combination), or compute all linear combinations in parallel (compute all, perform $\ctrl{R}$ in parallel, uncompute all).
\end{proof}

The following lemma shows that left-multiplying $M$ by an invertible matrix produces the same output distribution, up to a linear transformation of the output:

\begin{lemma}\label{lmm:distribution-invariance-under-left-multiplication}
    Let $A \in \ZZ_E^{s \times s}$ be an invertible matrix. Then $q_{\rho, A M} (\vd{z}) = q_{\rho, M} (\vd{z} A^{-1})$ for all $\vd{z} \in G^s$.
\end{lemma}
\begin{proof}
    Let $\vd{x} = (x_1, \dots, x_s)$. We have
    \begin{equation*}
        p_{\rho, A M} (\vd{x}) = \Tr(R_{A M}(\vd{x}) \rho^{\tp t}) = \Tr(R((\vd{x} A M)_1) \rho) \cdots \Tr(R((\vd{x} A M)_t) \rho) = p_{\rho, M} ((\vd{x} A)).
    \end{equation*}
    Hence,
    \begin{align*}
        q_{\rho, A M} (z_1, \dots, z_s) & = \frac{1}{\abs{G}^s} \sum_{x_1, \dots, x_s \in G} \overline{\chi_{z_1}(x_1) \cdots \chi_{z_s}(x_s)} p_{\rho, A M} (\vd{x}) \\
        & = \frac{1}{\abs{G}^s} \sum_{\vd{x} \in G^s} \overline{\chi_{\vd{z}} (\vd{x})} p_{\rho, M} (\vd{x} A) \\
        & = \frac{1}{\abs{G}^s} \sum_{\vd{x} \in G^s} \overline{\chi_{\vd{z}} (\vd{x} A^{-1})} p_{\rho, M} (\vd{x}),
    \end{align*}
    where the final equality is by the change of variable $\vd{x} \mapsto \vd{x} A^{-1}$.
    It is straightforward to verify that $\chi_{\vd{z}} (\vd{x} B) = \chi_{\vd{z} B} (\vd{x})$ for any square matrix $B$, so we are done.
\end{proof}

\begin{remark}\label{rmk:linearisation-matrix-in-standard-form}
    Note that swapping columns of $M$ clearly does not change the distribution $q_{\rho, M}$, and swapping rows, adding rows and multiplying rows by invertible scalars in $\ZZ_E$ corresponds to left-multiplying $M$ by an invertible matrix. So by \cref{lmm:distribution-invariance-under-left-multiplication} and \cref{lmm:row-and-column-operations-preserve-bell-sampling-matrix-conditions}, we may assume WLOG that $M$ is in standard form. Hence, by \cref{thm:linearisation-matrix-relation-to-linear-codes}, each linearisation of $R$ corresponds (up to a linear transform) to a free self-orthogonal linear code over $\ZZ_E$. Furthermore, each linearisation of $R$ also satisfying \myrefeq{Condition}{eq:zero-sum-property} corresponds to a code which is also zero-sum.
\end{remark}

\begin{lemma}\label{lmm:projective-fourier-sampling-subgroup-mass}
    For all subgroups $K_1, \dots, K_s \leq G$,
    \begin{equation*}
        q_{\rho, R_M} (K_1^\perp \times \cdots \times K_s^\perp) \leq \prod_{k = 1}^s \EE_{x_k \in K_k} \abs{\Tr(P(x_k) \rho)}.
    \end{equation*}
\end{lemma}
\begin{proof}
    We have
    \begin{align*}
        q_{\rho, R_M} (K_1^\perp \times \cdots \times K_s^\perp) & = q_{\rho, R_M} ((K_1 \times \cdots \times K_s)^\perp) \\
        & = \EE_{(x_1, \dots, x_s) \in K_1 \times \cdots \times K_s} \Tr(R_M(x_1, \dots, x_s) \rho^{\tp t}) \\
        & = \Abs{\EE_{(x_1, \dots, x_s) \in K_1 \times \cdots \times K_s} \prod_{\ell = 1}^t \Tr(R((\vd{x} M)_\ell) \rho)} \\
        & \leq \EE_{(x_1, \dots, x_s) \in K_1 \times \cdots \times K_s} \Abs{\prod_{\ell = 1}^t \Tr(R((\vd{x} M)_\ell) \rho)} \\
        & \leq \EE_{(x_1, \dots, x_s) \in K_1 \times \cdots \times K_s} \prod_{\ell = 1}^s \Abs{\Tr(R((\vd{x} M)_\ell) \rho)} \\
        & = \prod_{k = 1}^s \EE_{x_k \in K_k} \abs{\Tr(R(x_k) \rho)}
    \end{align*}
    where the first inequality is by the triangle inequality.
\end{proof}

\begin{corollary}
    Let $K > S$ be a subgroup of $G$. Then
    \begin{equation*}
        q_{\rho, R_M} (K^\perp \times \cdots \times K^\perp) \leq (1 - \epsilon / 2)^s.
    \end{equation*}
\end{corollary}
\begin{proof}
    Follows from \cref{lmm:projective-fourier-sampling-subgroup-mass} and \cref{lmm:weak-state-fourier-sampling::anti-concentration-of-distribution}.
\end{proof}

\begin{definition}
    We define the \emph{sampling rate} of the Fourier sampling distribution $q_{\rho, R_M}$ as $s / t$, the rate of the corresponding linear code.
\end{definition}

\begin{theorem}\label{thm:algorithm-for-projective-state-hsp}
    Let $\mathcal{IA} = \prb{ProjAnyonicSL}(G, R, \rho)$ be an instance of abelian \prb{ProjAnyonicSL} with hidden symmetry subgroup $S = \HSS(\mathcal{IA})$ and minimum asymmetry distance $\epsilon = \MASD(\mathcal{IA})$. Let $\mathcal{M}$ be the set of all minimal subgroups $K > S$. Let $M$ be a matrix satisfying \myrefeq{Condition}{eq:valid-bell-sampling-matrix} and \myrefeq{Condition}{eq:useful-bell-sampling-matrix-condition}.
    
    Then there is a quantum algorithm, using $O(\log \abs{\mathcal{M}} + \log(1 / \delta)) \cdot r(M) / \epsilon$ copies of $\rho$, which outputs a generating set of $S$ with probability at least $1 - \delta$. Moreover, if $\ctrl{R}$ and $\QFT_G$ can be implemented efficiently (in time $\polylog \abs{G}$), then the algorithm runs in time $\polylog(\abs{G}) / \epsilon$.

    Moreover, the algorithm runs in time $r(M) \cdot \polylog \abs{G} / \epsilon$ if $\ctrl{R}$ and $\QFT_G$ can be implemented in time $\polylog \abs{G}$.
\end{theorem}
\begin{proof}
    The instance $\mathcal{IA}$ gives rise to an instance $\prb{AnyonicSL}(G^s, R_M, \rho^{\tp t})$ of \prb{AnyonicSL} with hidden subgroup $S^s$ and minimum asymmetry distance $\epsilon' = 1 - (1 - \epsilon)^s$. So the problem in solved in the same way that we solve any abelian \prb{AnyonicSL} instance.

    The algorithm is simply: perform Fourier difference sampling on $m + 1$ copies of $\rho^{\tp s}$ using $R_M$, and output the intersection of the kernels of the samples.

    Let the Fourier difference samples of the algorithm be $z_1, \dots, z_{m s} \in \widehat{G}$ and the sample that is subtracted from the others (in the Fourier difference sampling subroutine) be $(z'_{-s + 1}, \dots, z'_0) \in \widehat{G}^s$. Each $z_i \in S^\perp$. The output of the algorithm is $A = \ker{z_1} \cap \cdots \cap \ker{z_{m s}}$. We have $S \leq A$, so the algorithm fails if and only if there is some $K \in \mathcal{M}$ such that $K \leq A$, i.e.
    \begin{equation*}
        (z_1, \dots, z_s), \dots, (z_{(m - 1) s + 1}, \dots, z_{m s}) \in (K^\perp)^s
    \end{equation*}
    for some $K \in \mathcal{M}$. Note that $(z_1, \dots, z_s), \dots, (z_{(m - 1) s + 1}, \dots, z_{m s})$ are conditionally independent given $z'_{-s + 1}, \dots, z'_0$. Hence, by a simple union bound over all $K \in \mathcal{M}$, the probability of failure, given $z'_{-s + 1}, \dots, z'_0$, is at most
    \begin{align*}
        & \sum_{K \in \mathcal{M}} \prod_{i = 1}^m \Pr((z_{(i - 1) s + 1}, \dots, z_{i s}) \in (K^\perp)^s \mid z'_{-s + 1}, \dots, z'_0) \\
        & \leq \abs{\mathcal{M}} \cdot ((1 - \epsilon / 2)^s)^m,
    \end{align*}
    and so by the law of total probability, the probability of failure is at most $\abs{\mathcal{M}} \cdot (1 - \epsilon / 2)^{m s} \leq \abs{\mathcal{M}} \cdot e^{-\epsilon m s / 2}$. In order for this to be upper bounded by $\delta$, it suffices that
    \begin{equation*}
        m s = \Ceil{\frac{\log \abs{\mathcal{M}} + \log(1 / \delta)}{-\log(1 - \epsilon / 2)}} = O((\log \abs{\mathcal{M}} + \log(1 / \delta)) / \epsilon).
    \end{equation*}
    The total number of copies of $\rho$ used is $t (m + 1) = r(M) s (m + 1)$.

    Note that if $M$ also satisfies \myrefeq{Condition}{eq:zero-sum-property}, then this gives rise to an instance $\prb{BoseSL}(G^s, R_M, \rho^{\tp t})$ of \prb{BoseSL} with hidden symmetry subgroup $S^s$ and minimum asymmetry distance $\epsilon' = 1 - (1 - \epsilon)^s$. In this case, we can use the algorithm from \cref{thm:efficient-algorithm-for-normal-core-statehsp} for \prb{BoseSL}; in particular, we can use Fourier sampling instead of Fourier difference sampling. The analysis is the same, except that we do not need to consider conditional independence of the samples.
\end{proof}


\subsection{Optimising sample rate and number of copies needed per sample}\label{subsec:optimising-rate-and-block-length}

As shown in the previous subsection, each linearisation of the projective representation $R$ corresponds to a free, self-orthogonal linear code over $\ZZ_E$. Clearly, the choice of code affects the linearisation representation and so the Fourier sampling distribution it induces. Therefore, certain desirable properties of the induced distribution may be achieved by choosing a certain code over another. Two properties which we wish to optimise are the number of state copies needed to sample from the distribution, and the sampling rate (number of samples obtained per copy of the state).

The number of copies needed to sample from the distribution is the block length $t$ of the code, and the sampling rate is the rate $s / t$ of the code. Minimising the block length of the code is particularly relevant to near-term applications, where it is difficult to maintain coherent access to many copies of the state. Maximising the rate would be useful in scenarios where the total number of state copies used is limited, and we want to obtain as many samples as possible from those copies.

Thus, we seek free, self-orthogonal linear codes $C$ over $\ZZ_E$ with high rate and small block length. We leave it as an open question whether other properties of the code, e.g. the minimum distance, have any relevance to the induced distribution.

We first consider codes with generator matrices $M$ which satisfy \myrefeq{Condition}{eq:valid-bell-sampling-matrix} and \myrefeq{Condition}{eq:useful-bell-sampling-matrix-condition}; recall these give rise to representations $R_M$ with which we may solve abelian \prb{ProjAnyonicSL} via Fourier difference sampling.

We then consider codes with generator matrices $M$ which also satisfy \myrefeq{Condition}{eq:zero-sum-property}; these give rise to representations $R_M$ with which we may solve abelian \prb{ProjAnyonicSL} via Fourier sampling.

When running our algorithm from \cref{thm:algorithm-for-projective-state-hsp} for abelian \prb{ProjAnyonicSL} using Fourier difference sampling (or indeed when running any algorithm that uses Fourier difference sampling), if we perform Fourier sampling $m + 1$ times, then we obtain $m s$ samples from the Fourier difference sampling distribution, which requires $t (m + 1)$ copies. So the rate of samples per copy is $\frac{m s}{t (m + 1)} = \frac{s}{t} \cdot \frac{m}{m + 1}$. For large $m$, this is approximately the rate $s / t$ of the code. So the rate of the code is the ``asymptotic'' rate of samples per copy of the state; for codes satisfying \myrefeq{Condition}{eq:zero-sum-property}, the rate of the code is exactly the rate of samples per copy of the state.

\subsubsection{Optimising without zero-row-sum condition}

Let $C \subseteq \ZZ_E^t$ be a free, self-orthogonal linear code of block length $t$ and dimension $s$, with some generator matrix $M$.

First, we seek to optimise the block length of $C$. The block length is clearly optimised when the dimension $s$ is $1$, so we are looking for matrices of the form $(\begin{matrix}
    1 & s_1 & \cdots & s_{t - 1}
\end{matrix})$ for minimal $t$ (since by \cref{rmk:linearisation-matrix-in-standard-form} we may assume $M$ is in standard form). The self-orthogonality condition implies that $1 + s_1^2 + \cdots + s_{t-1}^2 = 0 \pmod{E}$; hence, the minimal block length for the code is dependent on the decomposability of $E$ as a modular sum of squares. As the following lemmas establish, this in turn has a simple dependence on the prime factorisation of $E$, and such a decomposition can be found efficiently.

\begin{lemma}\label{lmm:sum-of-squares-decomposition-conditions}
    Let $k \in \NN$ with $k \geq 2$. Then:
    \begin{enumerate}
        \item $k \neq 0 \pmod{4}$ and $k$ contains no prime factor of the form $4 m + 3$ if and only if there exists $s_1 \in \ZZ$ such that $1 + s_1^2 = 0 \pmod{k}$.
        \item $k \neq 0 \pmod{4}$ if and only if there exist $s_1, s_2 \in \ZZ$ such that $1 + s_1^2 + s_2^2 = 0 \pmod{k}$.
        \item If $k \neq 0 \pmod{8}$ if and only if there exist $s_1, s_2, s_3, s_4 \in \ZZ$ such that $1 + s_1^2 + s_2^2 + s_3^2 = 0 \pmod{k}$.
        \item If $k = 0 \pmod{8}$, then there exist $s_1, s_2, s_3, s_4 \in \ZZ$ such that $1 + s_1^2 + s_2^2 + s_3^2 + s_4^2 = 0 \pmod{k}$.
    \end{enumerate}
\end{lemma}

\begin{lemma}\label{lmm:sum-of-squares-decomposition-is-efficient}
    For each of the 4 cases in \cref{lmm:sum-of-squares-decomposition-conditions}, the condition on $k$ can be checked in $O(\polylog k)$ time, and if the condition on $k$ holds, then a set of corresponding $s_i$ can be found in $O(\polylog k)$ time.
\end{lemma}

We defer the proofs of \cref{lmm:sum-of-squares-decomposition-conditions,lmm:sum-of-squares-decomposition-is-efficient} to \cref{sec:appendix}.

\begin{corollary}\label{crl:linearisation-no-zero-row-sum-condition-min-block-length}
    Fourier difference sampling with a linearisation of the  projective representation $R$ can be performed with:
    \begin{itemize}
        \item at most 2 copies of $\rho$ if $E \neq 0 \pmod{4}$ and $E$ contains no prime factor of the form $4 m + 3$.
        \item at most 3 copies of $\rho$ if $E \neq 0 \pmod{4}$.
        \item at most 4 copies of $\rho$ if $E \neq 0 \pmod{8}$.
        \item at most 5 copies of $\rho$ for all $E$.
    \end{itemize}
\end{corollary}

Now we seek to optimise the sampling rate. We have the following very simple upper bound:
\begin{lemma}\label{lmm:upper-bound-on-sampling-rate}
    The sampling rate is at most $1 / 2$.
\end{lemma}
\begin{proof}
    Since $C$ is self-orthogonal, we have $C \subseteq C^\perp$, and so $d^s = \abs{C} \leq \abs{C^\perp} = d^{t - s}$, i.e. $s \leq t - s$.
\end{proof}

For any $E$, given a decomposition of $E - 1$ into a sum of four squares $s_1^2 + s_2^2 + s_3^2 + s_4^2 = E - 1$ (which exists by Lagrange's four-square theorem and may be found efficiently by \cite{PT18lagrangeFourSquares}), \cite{ADIS25quditBellSampling} constructed a matrix $\vd{R}$ such that $\vd{R}^T \vd{R} = \vd{R} \vd{R}^T = (E - 1) I$; explicitly,
\begin{equation*}
    \vd{R} = \begin{pmatrix}
        s_1 & s_2 & s_3 & s_4 \\
        s_2 & -s_1 & s_4 & -s_3 \\
        s_3 & -s_4 & -s_1 & s_2 \\
        s_4 & s_3 & -s_2 & -s_1
    \end{pmatrix}
\end{equation*}
Then $M_8 = (I \mid \vd{R})$ satisfies the conditions in \cref{lmm:characterisation-of-bell-sampling-representations} and \cref{lmm:characterisation-of-useful-bell-sampling-representations}, so the Fourier difference sampling distribution $q_{\rho, M_8}$ has sampling rate $1 / 2$, which is optimal, while the block length of the code is only $8$.

Thus, to achieve optimal sample rate, we never need consider codes of block length more than $8$. However, we needed codes of block length only $5$ to achieve the optimal block length. So there is a trade-off between optimal block length and optimal sampling rate. The following lemma establishes the best possible sampling rates for each block length.

\begin{lemma}\label{lmm:rate-block-length-tradeoff}\leavevmode
    \begin{itemize}
        \item If $E \neq 0 \pmod{4}$ and $E$ contains no prime factor of the form $4 m + 3$, then the optimal sampling rate of $1 / 2$ can be achieved using a code with optimal block length of $2$. (Recall $2$ is the optimal block length regardless of sampling rate for this case.)
        \item If $E \neq 0 \pmod{4}$, then the optimal sampling rate of $1 / 2$ can be achieved using a code with block length of $4$, which is optimal. (Recall $3$ is the optimal block length regardless of sampling rate for this case.)
        \item If $E = 0 \pmod{4}$, the optimal sampling rate of $1 / 2$ is achieved by the code above with block length $8$, which is optimal. There is a code with block length $7$ and sampling rate $3 / 7$ which is optimal for this block length, a code with block length $6$ and sampling rate $1 / 3$ which is optimal for this block length. For block length $5$, the best sampling rate is $1 / 5$.
    \end{itemize}
    Moreover, all these codes can be found in time $O(\polylog E)$.
\end{lemma}

We defer the proof of \cref{lmm:rate-block-length-tradeoff} to \cref{sec:appendix}.

\subsubsection{Optimising with zero-row-sum condition}

Now suppose $M$ in addition satisfies \myrefeq{Condition}{eq:zero-sum-property}, so that $C$ is a zero-sum code. Clearly, adding an extra condition to $C$ means the optimal block length and rate cannot improve. Given any matrix $M$ satisfying the conditions in \cref{lmm:characterisation-of-bell-sampling-representations} and \cref{lmm:characterisation-of-useful-bell-sampling-representations}, we can always consider the matrix $M' = (M \mid -M)$ which satisfies the condition in \cref{lmm:condition-for-zero-sum-property} as well as \cref{lmm:characterisation-of-bell-sampling-representations} and \cref{lmm:characterisation-of-useful-bell-sampling-representations}. However, this doubles the block length and halves the rate. Can we do better than this?

First, we consider how small the block length $t$ can be. We want to find $s_1, \dots, s_{t - 1}$ such that $1 + s_1^2 + \cdots + s_{t-1}^2 = 0 \pmod{E}$ and $1 + s_1 + \cdots + s_{t-1} = 0 \pmod{E}$ when $E$ is odd, and $1 + s_1^2 + \cdots + s_{t-1}^2 = 0 \pmod{2E}$ and $1 + s_1 + \cdots + s_{t-1} = 0 \pmod{E}$ when $E$ is even. The following lemma shows that $t = 4$ is sufficient for all $E$.

\begin{lemma}\label{lmm:sum-of-squares-decomposition-with-zero-sum-condition}
    Let $k \in \NN$ and $k \geq 2$.
    \begin{enumerate}
        \item $k = 2$ if and only if there exists $s_1 \in \ZZ$ such that $1 + s_1^2 = 0 \pmod{k}$ and $1 + s_1 = 0 \pmod{k}$.
        \item $k \neq 0 \pmod{4}$ and every prime factor of $k$ is either $3$ or of the form $3m + 1$ if and only there exist $s_1, s_2 \in \ZZ$ such that $1 + s_1^2 + s_2^2 = 0 \pmod{k}$ and $1 + s_1 + s_2 = 0 \pmod{k}$.
        \item $k \neq 0 \pmod{8}$ if and only if there exist $s_1, s_2, s_3 \in \ZZ$ such that $1 + s_1^2 + s_2^2 + s_3^2 = 0 \pmod{k}$ and $1 + s_1 + s_2 + s_3 = 0 \pmod{k}$.
        \item If $k = 0 \pmod{8}$, then there exist $s_1, s_2, s_3, s_4 \in \ZZ$ such that $1 + s_1^2 + s_2^2 + s_3^2 + s_4^2 = 0 \pmod{k}$ and $1 + s_1 + s_2 + s_3 + s_4 = 0 \pmod{k}$.
    \end{enumerate}
    Moreover, each decomposition can be found in expected time $O(\polylog(k))$.
\end{lemma}

We defer the proof of \cref{lmm:sum-of-squares-decomposition-with-zero-sum-condition} to \cref{sec:appendix}.

\begin{corollary}\label{crl:linearisation-zero-row-sum-condition-min-block-length}
    Fourier sampling with a linearisation of the projective representation $R$ can be performed with:
    \begin{itemize}
        \item at most 3 copies of $\rho$ if every prime factor of $E$ is either $3$ or of the form $3m + 1$.
        \item at most 4 copies of $\rho$ if $E \neq 0 \pmod{4}$.
        \item at most 5 copies of $\rho$ for all $E$.
    \end{itemize}
\end{corollary}
\begin{proof}
    Follows from \cref{lmm:sum-of-squares-decomposition-with-zero-sum-condition}. Note that if $E$ is even, we must decompose $2E$, not $E$ (so we take $k = 2E$ in \cref{lmm:sum-of-squares-decomposition-with-zero-sum-condition}). Importantly, if $E$ is even, for any solution $s_1, ..., s_{t-1}$ satisfying $s_1 + \cdots + s_{t-1} = s_1^2 + \cdots + s_{t-1}^2 = -1 \pmod{2E}$, we can assume WLOG that $s_1, ..., s_{t-1}$ are in $[0, E - 1]$: if some $s_i \in [E, 2E - 1]$, then replace $s_i$ with $s_i - E$: this means we still have $s_1 + ... + s_{t-1} = -1 \pmod{E}$ and $s_1^2 + ... + s_{t-1}^2 = -1 \pmod{2E}$, since $(s_i - E)^2 = s_i^2 \pmod{2E}$.
\end{proof}

Now consider how large the rate can be. Suppose there is such a valid $M$ with optimal rate $1 / 2$. Then $C = C^\perp$, so $C$ is a self-dual code. Since $C$ is zero-sum, we have $(1, \dots, 1) \in C^\perp = C$. Hence, $(1, \dots, 1)$ is orthogonal to itself, so $2s = t = 0 \pmod{E}$, so $E \mid t$. Hence, the block length $t$ must be at least $E$. So we see that requiring optimal rate here means that the minimum sized code must have block length at least $E$, whereas before, ignoring the row-sum condition allowed the block-length to be constant (at most $8$).

To conclude this section, we consider codes $C$ which do achieve the optimal sampling rate of $1 / 2$ for small values of $E$.

\begin{example}
    Suppose the group exponent is $E = 2$. Let $M_\text{Ham}$ be the generator matrix of the extended binary $[8, 4, 4]$ Hamming code, i.e.
    \begin{equation*}
        M_\text{Ham} = \begin{pmatrix}
            1 & 0 & 0 & 0 & 0 & 1 & 1 & 1 \\
            0 & 1 & 0 & 0 & 1 & 0 & 1 & 1 \\
            0 & 0 & 1 & 0 & 1 & 1 & 0 & 1 \\
            0 & 0 & 0 & 1 & 1 & 1 & 1 & 0
        \end{pmatrix} \in \ZZ_2^{4 \times 8}.
    \end{equation*}
    $M_\text{Ham}$ is in standard form and its rows are mutually orthogonal modulo $2$, have entries summing to zero modulo $2$, and are isotropic modulo $4$. Thus, $M_\text{Ham}$ satisfies \myrefeq{Condition}{eq:valid-bell-sampling-matrix}, \myrefeq{Condition}{eq:useful-bell-sampling-matrix-condition} and \myrefeq{Condition}{eq:zero-sum-property}. Hence, the Fourier sampling distribution induced by $R_{M_\text{Ham}}$ has optimal sampling rate $1 / 2$. Furthermore, $t = 8$ is the smallest block length for which this is possible: WLOG, $M$ must be in standard form, and the isotropicity condition in \myrefeq{Condition}{eq:zero-sum-property} means that each row of $M$ must have weight at least $4$, which already implies $t \geq 7$.
\end{example}

\begin{example}
    Suppose the group exponent is $E = 3$. Let $M_\text{Golay}$ be the generator matrix of the ternary $[12, 6, 6]$ extended Golay code, i.e.
    \begin{equation*}
        M_\text{Golay} = \left( 
            \begin{array}{cccccc|cccccc}
                1 & 0 & 0 & 0 & 0 & 0 & 0 & 1 & 1 & 2 & 2 & 2 \\
                0 & 1 & 0 & 0 & 0 & 0 & 1 & 0 & 2 & 1 & 2 & 2 \\
                0 & 0 & 1 & 0 & 0 & 0 & 1 & 2 & 0 & 2 & 1 & 2 \\
                0 & 0 & 0 & 1 & 0 & 0 & 2 & 1 & 2 & 0 & 1 & 2 \\
                0 & 0 & 0 & 0 & 1 & 0 & 2 & 2 & 1 & 1 & 0 & 2 \\
                0 & 0 & 0 & 0 & 0 & 1 & 1 & 1 & 1 & 1 & 1 & 0
            \end{array}
        \right) \in \ZZ_3^{6 \times 12}.
    \end{equation*}
    $M_\text{Golay}$ is in standard form and its rows are mutually orthogonal modulo $3$ and have entries summing to zero modulo $3$. Thus, $M_\text{Golay}$ satisfies \myrefeq{Condition}{eq:valid-bell-sampling-matrix}, \myrefeq{Condition}{eq:useful-bell-sampling-matrix-condition} and \myrefeq{Condition}{eq:zero-sum-property}. Hence, the Fourier sampling distribution induced by $R_{M_\text{Golay}}$ has optimal sampling rate $1 / 2$.
\end{example}

\subsection{Applications}

\subsubsection{Stabiliser group learning}

As mentioned earlier, the representation $R_W: \ZZ_d^{2n} \to \mathcal{U}(\CC^{d^n})$ defined by $R_W (a, b) = \tilde{W}_{(a, b)} = X^a Z^b$ is a projective representation of $\ZZ_d^{2n}$ with cocycle $\omega((a, b), (a', b')) = \omega_d^{b \cdot a'}$.
$\tilde{W}_{(a, b)}$ is a phaseless version of the \emph{Weyl operator} $W_{(a, b)} \coloneq \tau^{a \cdot b} X^a Z^b$, where $\tau = e^{\pi i (d^2 + 1) / d}$. Since $\tilde{W}_{(a, b)}$ and $W_{(a, b)}$ differ only by a phase, the \emph{phaseless stabiliser group} $\Weyl(\rho) = \{x \in \ZZ_d^{2n}: \abs{\Tr(W_x \rho)} = 1\}$ is equal to $\{x \in \ZZ_d^{2n}: \abs{\Tr(\tilde{W}_x \rho)} = 1\}$. Hence, learning the stabiliser group of an arbitrary $n$-qudit state $\rho$ can be formulated as instance of \prb{ProjAnyonicSL}; in particular, the instance is $\prb{ProjAnyonicSL}(\ZZ_d^{2n}, R, \rho)$ with hidden symmetry subgroup $\Weyl(\rho)$. Quantum lgorithms for learning stabiliser groups of pure qubit and qudit states were already known \cite{Mon17,HEC25abelianStateHSP,ADIS25quditBellSampling}: the above formulation gives the first explicit quantum algorithm for stabiliser learning of arbitrary qudit states. It also improves upon the number of copies needed at a time from the algorithm of \cite{ADIS25quditBellSampling} (at most 5 for worst case $d$ compared to at most 8).

Note that from the phaseless stabiliser group $\Weyl(\rho)$, we can recover the stabiliser group $\Stab(\rho) = \{P \in \mathcal{P}_d^{\tp n}: P \rho = \rho\}$ by measuring $\rho$ in the eigenbasis of $W_x$ for each output generator $x$ from the algorithm.

\begin{theorem}\label{thm:stabiliser-group-learning-algorithm}
    Given copies of an $n$-qudit state $\rho \in \densmats{(\CC^d)^{\tp n}}$, there is a $\poly(n)$-time quantum algorithm which uses $O(n)$ copies of $\rho$ at a time and outputs with high probability a generating set for the stabiliser group $\Stab(\rho)$ of $\rho$.
\end{theorem}

In an upcoming work, we show that performing Fourier sampling and Fourier difference sampling with linearisations of the projective representation $x \mapsto \tilde{W}_x$ leads to a general notion of Bell sampling and Bell difference sampling for qudits, which we use in a variety of applications aside from stabiliser group learning.

\subsubsection{Bose symmetry learning over central extensions}

A central extension of a group $Q$ by a group $Z$ is a group $G$ such that $Z \leq Z(G)$ and $G / Z \cong Q$, where $Z(G)$ is the centre of $G$.

A common way of linearising a projective representation $R: Q \to \mathcal{U}(\mathcal{H})$ is to construct a linear representation $R'$ of the central extension $G$ of $Q$ by $Z = \ZZ_E$, where $E$ is the exponent of $Q$. If we were able to solve \prb{BoseSL} over $G$ with representation $R'$, then we could solve \prb{ProjAnyonicSL} over $Q$ with representation $R$. However, an efficient quantum algorithm for \prb{BoseSL} over general central extensions is not known, since they are generally non-abelian and not poly-near Hamiltonian. This is why we instead used the approach in \cref{sec:linearisation-of-projective-representations} of constructing linearisations of $R$ which are representations of \emph{abelian} groups, which allows us to use the efficient quantum algorithm for \prb{BoseSL} over abelian groups.

Given that we now have an efficient quantum algorithm for \prb{ProjAnyonicSL} over abelian groups by \cref{thm:algorithm-for-projective-state-hsp}, it is natural to ask whether the reduction from \prb{ProjAnyonicSL} over abelian groups to \prb{BoseSL} over central extensions can be reversed; i.e. can we use the algorithm for \prb{ProjAnyonicSL} over abelian groups to solve \prb{BoseSL} over central extensions of abelian groups?

\begin{conjecture}
    There is an efficient quantum algorithm for \prb{BoseSL} over central extensions of abelian groups.
\end{conjecture}

There is good reason to suspect this is true: it can be shown that if we perform Fourier sampling over the abelian subgroup $Z$ of $G$, then the restriction to $Q$ of the linear representation of $G$ is a projective representation of $Q$ which acts on the post-measurement states obtained from Fourier sampling over $Z$.
We have made good progress on resolving this conjecture, and hope to fully answer it soon.\isaac{TODO: do we need this last sentence?}

\section{Symmetries of other quantum objects}\label{sec:symmetries-of-other-quantum-objects}

So far, we have only considered symmetry learning problems in which the object of interest is a quantum state. In this section, we introduce similar problems pertaining to learning symmetries of other quantum objects; namely, unitaries, Hamiltonians, finite subsets of a Hilbert space, and subspaces of a Hilbert space.



\subsection{Learning symmetries of unitaries}

We introduce the problem of learning symmetries of unitaries:

\begin{problem}[Unitary Symmetry Learning (\prb{UnitarySL})]\label{prb:quantum-unitary-symmetry-learning}\leavevmode
    \begin{description}
        \item[Input] Blackbox access to a unitary $U \in \mathcal{U}(\mathcal{H})$.
        \item[Promise] There is a finite group $G$, a projective unitary representation $R: G \to \mathcal{U}(\mathcal{H})$, and a subgroup $K \leq G$ such that $U$ is symmetric under the action of $K$, i.e. $R(g) U R(g)^\dagger = U$ for all $g \in K$, and $U$ is not symmetric under the action of any $g \notin K$, i.e. $\norm{R(g) U R(g)^\dagger - U} \geq \epsilon$ for all $g \notin K$, for some specified norm $\norm{\cdot}$ and $\epsilon > 0$.
        \item[Task] Learn the hidden subgroup $K$.
    \end{description}
\end{problem}

One might think that our linearisation construction from \cref{sec:projective-state-hsp} would be needed for an algorithm that solves \prb{UnitarySL} with respect to a projective representation, as was needed for solving \prb{AnyonicSL} with respect to a projective representation. However, the next two facts show that in fact we can instead consider the linear representation $\overline{R} \tp R$, from which we can perform Fourier sampling due to the nature of the Choi state of $U$.

$\overline{R} \tp R$ can be viewed as a ``linearisation'' of $R$ in a sense as it is a linear representation constructed from $R$, but it does not belong to the class of linearisations we consider in \cref{sec:projective-state-hsp}.

\begin{lemma}\label{lmm:complex-conjugate-linearisation}
    If $R: G \to \mathcal{U}(\mathcal{H})$ is a projective unitary representation, then $\overline{R} \tp R: G \to \mathcal{U}(\mathcal{H} \tp \mathcal{H})$ is a linear unitary representation.
\end{lemma}
\begin{proof}
    For all $g, h \in G$, we have
    \begin{align*}
        (\overline{R} \tp R)(g) (\overline{R} \tp R)(h) & = \overline{R(g) R(h)} \tp R(g) R(h) \\
        & = \overline{\omega(g, h) R(g h)} \tp \omega(g, h) R(g h) \\
        & = \overline{R(g h)} \tp R(g h)
    \end{align*}
    where $\omega(g, h)$ is the cocycle of $R$.
\end{proof}

\begin{lemma}[{\cite[Figure 2]{LW22hamiltonianSymmetries}}]\label{lmm:fourier-sampling-with-complex-conjugate-linearisation}
    We can perform Fourier sampling with representation $R' = \overline{R} \tp R$ on the Choi state $\Phi^\mathcal{E}$ (where $\mathcal{E}(\rho) = U \rho U^\dagger$) using $1$ query to $U$.
\end{lemma}
\begin{proof}
    The standard Fourier sampling circuit is: start with state $\ket{0} \tp \ket{\Phi^\mathcal{E}}$, apply $\QFT_G$ to the first register, apply controlled-$\overline{R}(g)$ to the first and second registers, apply controlled-$R(g)$ to the first and third registers, then apply $\QFT_G^\dagger$ to the first register, then measure the first register.
    
    Since $\ket{\Phi^\mathcal{E}} = (I \tp U) \ket{\Phi}$, we can instead start with the maximally entangled state $\ket{\Phi}$ and apply the controlled-$\overline{R}(g)$ to the first and second registers, then apply $U$ to the third register, then proceed as before. But now by the transpose trick $(A \tp I) \ket{\Phi} = (I \tp A^T) \ket{\Phi}$, we can instead apply controlled-$R(g)^\dagger$ to the first and third registers, and proceed as before. Hence, the second register is not operated on, so can be traced out.
    
    This leaves us with the following circuit which samples from the same distribution as the original: start with $\ketbra{0}{0} \tp I / N$, apply $\QFT_G$ to the first register, apply controlled-$R(g)^\dagger$ to the first and second registers, apply $U$ to the second register, apply controlled-$R(g)$ to the first and second registers, then apply $\QFT_G^\dagger$ to the first register, then measure the first register.
\end{proof}

The following fact shows that the measure of Bose-symmetry of the Choi state $\Phi^U$ of $U$ with respect to $R'(g)$ is related to a measure of how symmetric $U$ is with respect to $R(g)$.

\begin{lemma}\label{lmm:choi-state-bose-symmetry-measure}
    Let $V$ and $U$ be $N \times N$ unitaries, let $\mathcal{E}$ denote the quantum channel $\mathcal{E}(\rho) = U \rho U^\dagger$, let $\Phi^\mathcal{E}$ denote the Choi state of $\mathcal{E}$. Then
    \begin{equation*}
        \Tr((\overline{V} \tp V) \Phi^\mathcal{E}) = \frac{1}{N} \Tr(U V U^\dagger V^\dagger).
    \end{equation*}
\end{lemma}
\begin{proof}
    Let $\ket{\Phi} = \frac{1}{\sqrt{N}} \sum_{x \in \ZZ_N} \ket{x} \ket{x}$ denote the maximally entangled state, so that $\ket{\Phi^\mathcal{E}} = (I \tp U) \ket{\Phi}$. Then
    \begin{align*}
        \Tr((\overline{V} \tp V) \Phi^\mathcal{E}) & = \bra{\Phi^\mathcal{E}} (\overline{V} \tp V) \ket{\Phi^\mathcal{E}} \\
        & = \bra{\Phi} (I \tp U^\dagger) (\overline{V} \tp V) (I \tp U) \ket{\Phi} \\
        & = \bra{\Phi} (I \tp U^\dagger V U) (\overline{V} \tp I) \ket{\Phi} \\
        & = \bra{\Phi} (I \tp U^\dagger V U V^\dagger) \ket{\Phi} \\
        & = \frac{1}{N} \Tr(U V U^\dagger V^\dagger),
    \end{align*}
    where in the penultimate equality we have used the transpose trick $(A \tp I) \ket{\Phi} = (I \tp A^T) \ket{\Phi}$.
\end{proof}

\begin{theorem}\label{thm:unitary-symmetry-learning-reduction-to-bose-symmetry-learning}
    Let $U \in \mathcal{U}(\mathcal{H})$ be an $N \times N$ unitary. Then when working with the normalised Frobenius norm, \prb{UnitarySL} (\cref{prb:quantum-unitary-symmetry-learning}) is reducible to \prb{BoseSL} with state $\ket{\Phi^U}$, group $G$, representation $R' = \overline{R} \tp R$, asymmetry parameter $\epsilon^2 / 2$; in particular, abelian \prb{UnitarySL} can be solved using $O(\log \abs{G} / \epsilon^2)$ queries to $U$.
\end{theorem}
\begin{proof}
    We are promised that $\frac{1}{N} \norm{R(g) U R(g)^\dagger - U}_F^2 \geq \epsilon^2$ for all $g \notin K$, where $\norm{\cdot}_F$ denotes the unnormalised Frobenius norm. Let $R'(g) = \overline{R(g)} \tp R(g)$, which is a linear representation by \cref{lmm:complex-conjugate-linearisation}. We have
    \begin{align*}
        \norm{R(g) U R(g)^\dagger - U}_F^2 & = \norm{R(g) U R(g)^\dagger}_F^2 + \norm{U}_F^2 - 2 \Re(\Tr(R(g) U R(g)^\dagger U^\dagger)) \\
        & = 2(N - \Re(\Tr(R(g) U R(g)^\dagger U^\dagger)))
    \end{align*}
    Hence, $\frac{1}{N} \Re(\Tr(R(g) U R(g)^\dagger U^\dagger)) \leq 1 - \epsilon^2 / 2$ for all $g \notin K$. By \cref{lmm:choi-state-bose-symmetry-measure}, $\abs{\Tr(R'(g) \Phi^U)} \leq 1 - \epsilon^2 / 2$ for all $g \notin K$. Also, it is clear that $\Tr(R'(g) \Phi^U) = 1$ for all $g \in K$. We can implement Fourier sampling for $\overline{R} \tp R$ on $\Phi^U$ using $1$ query to $U$ by \cref{lmm:fourier-sampling-with-complex-conjugate-linearisation}. Hence, \prb{UnitarySL} reduces to \prb{BoseSL}.
\end{proof}

\subsection{Learning symmetries of Hamiltonians}

Similarly to unitary symmetry learning (\cref{prb:quantum-unitary-symmetry-learning}), we introduce the problem of learning symmetries of Hamiltonians. We take inspiration from \cite{LW22hamiltonianSymmetries}, who considered the related problem of testing symmetries of Hamiltonians.

\begin{problem}[Hamiltonian Symmetry Learning (\prb{HamiltonianSL})]\label{prb:quantum-hamiltonian-symmetry-learning}\leavevmode
    \begin{description}
        \item[Input] Blackbox access to the evolution operator $U(t) = e^{-i H t}$ of an $N \times N$ Hamiltonian $H$.
        \item[Promise] There is a finite group $G$, a projective unitary representation $R: G \to \mathcal{U}(\mathcal{H})$, and a subgroup $S \leq G$ such that $H$ is symmetric under the action of $S$, i.e. $R(g) H R(g)^\dagger = H$ for all $g \in S$, and $H$ is not symmetric under the action of any $g \notin S$, i.e. $\norm{R(g) H R(g)^\dagger - H} \geq \epsilon$ for all $g \notin S$, for some specified norm $\norm{\cdot}$ and $\epsilon > 0$.
        \item[Task] Learn the hidden subgroup $S$.
    \end{description}
\end{problem}

Given that our access to the Hamiltonian $H$ is via its unitary evolution operator, it makes sense to attempt to reduce this problem to unitary symmetry learning. The following lemma is the key to this reduction.

\begin{lemma}\label{lmm:symmetry-measure-of-unitary-lower-bound-by-symmetry-measure-of-hamiltonian}
    Let $U$ be unitary and $B$ be Hermitian. Let $\mu = \lambda_{\max}(B) - \lambda_{\min}(B)$ be the spectral width of $B$. If $\mu \in (0, 2\pi)$, then $\norm{U e^{-i B} U^\dagger - e^{-i B}}_F \geq 2 \frac{\sin(\mu/2)}{\mu} \norm{U B U^\dagger - B}_F$.
\end{lemma}
\begin{proof}
    Write $B$ in its eigen-decomposition as $B = \sum_i \lambda_i \ketbra{v_i}{v_i}$. For ease of notation, index matrices in the eigenbasis $\{\ket{v_i}\}$ of $B$. Then $[U, B]_{j k} = \braket{v_j | [U, B] | v_k} = U_{j k} (\lambda_k - \lambda_j)$ and $[U, e^{-i B}]_{j k} = U_{j k} (e^{-i \lambda_k} - e^{-i \lambda_j})$. We have $\abs{e^{-i \lambda_k} - e^{-i \lambda_j}} = 2 \abs{\sin((\lambda_k - \lambda_j)/2)}$. Define
    \begin{equation*}
        a_{j k} = \abs{[U, e^{-i B}]_{j k} / [U, B]_{j k}} = \abs{f(\lambda_k - \lambda_j)}.
    \end{equation*}
    where $f(x) = 2 \sin(x/2) / x$. Since $f$ is decreasing on $(0, 2\pi)$ and even, we have $a_{j k} \geq f(\mu)$ for all $j, k$. The inequality follows from the definition of the Frobenius norm, and unitary invariance of the Frobenius norm.
\end{proof}

\begin{theorem}\label{thm:hamiltonian-symmetry-learning-reduction-to-unitary-symmetry-learning}
    Let $H$ be an $N \times N$ Hamiltonian. Then, when working with the normalised Frobenius norm, \prb{HamiltonianSL} (\cref{prb:quantum-hamiltonian-symmetry-learning}) is reducible to \prb{UnitarySL} (\cref{prb:quantum-unitary-symmetry-learning}) with unitary $U = U(t^*) = e^{-i H t^*}$, asymmetry parameter $\epsilon a b / (2 \pi \mu)$, where $t^*$ is any value in the interval $\frac{1}{\mu} \cdot [a, 2\pi - b]$, $\mu$ is the spectral spread $\lambda_{\max} (H) - \lambda_{\min} (H)$, and $a, b \in (0, 2\pi)$. Hence, it is also reducible to \prb{BoseSL}.
    
    In particular, assuming we can make queries to $U(t) = e^{-i H t}$ for at least one $t \in \frac{1}{\mu} \cdot [\Omega(1 / \polylog \abs{G}), 2\pi - \Omega(1 / \polylog \abs{G})]$ (note the input $t$ to $U(t)$ is allowed to be different across queries), abelian \prb{HamiltonianSL} can be solved in $O(\polylog \abs{G} \cdot \mu^2 / \epsilon^2)$ queries.
\end{theorem}
\begin{proof}
    For $g \in S$ (i.e. $R(g) H R(g)^\dagger = H$), we also have $R(g) U(t) R(g)^\dagger = U(t)$ for all $t$. For the rest of the proof, let $g \notin S$, so $\norm{R(g) H R(g)^\dagger - H}_F \geq \epsilon$.

    Let $t^* \in \frac{1}{\mu} \cdot [a, 2\pi - b]$ and let $\mu^* = \mu t^* = \mu(H t^*)$. Clearly $\mu^* \in [a, 2 \pi - b]$, so by \cref{lmm:symmetry-measure-of-unitary-lower-bound-by-symmetry-measure-of-hamiltonian},
    \begin{equation*}
        \norm{R(g) e^{-i H t^*} R(g)^\dagger - e^{-i H t^*}}_F \geq 2 \frac{\sin(\mu^*/2)}{\mu^*} \norm{R(g) H t^* R(g)^\dagger - H t^*}_F \geq \frac{\sin(\pi - b/2)}{\pi - b/2} \epsilon t^*.
    \end{equation*}
    By the simple inequality $\sin(x) / x \geq 1 - x / \pi$ for $x \in [0, \pi]$, and the bound on $t^*$, we therefore have $\norm{R(g) U(t^*) R(g)^\dagger - U(t^*)}_F \geq \epsilon a b / (2 \pi \mu)$.

    Finally, we note that when after applying the reduction to \prb{BoseSL}, we need not have same Choi state $\Phi^{U(t^*)}$ for each Fourier sample obtained, since performing Fourier sampling on $\Phi^{U(t^*)}$ for any $t^* \in \frac{1}{\mu} \cdot [a, 2\pi - b]$ will yield a distribution which is supported on $\hat{G}[\core (S)]$ and of which at most $1 - \epsilon a b / (2 \pi \mu)$ of the probability mass is supported on $\hat{G}[K]$ for any $K > S$ by \cref{lmm:weak-state-fourier-sampling-distribution-supported-on-trivial-irreps-of-hidden-subgroup-normal-core} and \cref{lmm:weak-state-fourier-sampling::anti-concentration-of-distribution}. In the proof of \cref{thm:efficient-algorithm-for-normal-core-statehsp}, we do not require that each Fourier sample is identically distributed, only that they are independent and satisfy the two above support properties.
\end{proof}

\subsection{Learning symmetries of sets of states}

We introduce the problem of learning the Bose symmetry of a finite set of states:

\begin{problem}[\prb{SubsetBoseSL}]\label{prb:set-of-states-bose-symmetry-learning}\leavevmode
    \begin{description}
        \item[Input] Access to copies of a finite set of states $A = \{\ket{\psi_1}, \dots, \ket{\psi_m}\} \subset \mathcal{H}$.
        \item[Promise] There is a finite group $G$, a linear unitary representation $R: G \to \mathcal{U}(\mathcal{H})$, and a subgroup $S \leq G$ such that $A$ is symmetric under the action of $S$, i.e. $R(g) A = A$ for all $g \in S$, and $A$ is not symmetric under the action of any $g \notin S$, i.e. $\min_{j \in [m]} \max_{k \in [m]} \abs{\braket{\psi_j | R(g) \ket{\psi_k}}} \leq 1 - \epsilon$.
        \item[Task] Learn the hidden symmetry subgroup $S$.
    \end{description}
\end{problem}

Note that the condition $R(g) A = A$ is equivalent to $\forall i \in [m], \exists j \in [m]$ such that $R(g) \ket{\psi_i} = \ket{\psi_j}$. Since the $j$ must be unique, this is equivalent to the existence of a permutation $\sigma \in S_m$ (the symmetry group on $m$ letters) such that $R(g) \ket{\psi_i} = \ket{\psi_{\sigma(i)}}$ for all $i \in [m]$.

For $\sigma \in S_m$, write $P(\sigma)$ for the permutation operator on $\mathcal{H}^{\tp m}$ which permutes the $m$ tensor factors according to $\sigma$, i.e. $P(\sigma) \ket{x_1} \tp \cdots \tp \ket{x_m} = \ket{x_{\sigma^{-1}(1)}} \tp \cdots \tp \ket{x_{\sigma^{-1}(m)}}$. $P$ is a representation of $S_m$. Note that for any unitary $U \in \mathcal{U}(\mathcal{H})$, we have $P(\sigma) U^{\tp m} = U^{\tp m} P(\sigma)$. 

\begin{theorem}\label{thm:subset-bose-symmetry-learning-reduction-to-bose-symmetry-learning}
    \prb{SubsetBoseSL} over $G$ is reducible to \prb{BoseSL} over $G \times S_m$.
\end{theorem}
\begin{proof}
    Let $G' = G \times S_m$. Let $\ket{\psi'} = \ket{\psi_1} \tp \cdots \tp \ket{\psi_m}$. Define the map $R': G' \to \mathcal{U}((\mathcal{H})^{\tp m})$ by
    \begin{equation*}
        R'(g, \sigma) \coloneq R(g)^{\tp m} P(\sigma) = P(\sigma) R(g)^{\tp m}.
    \end{equation*}
    Since $R(g)^{\tp m}$ and $P(\sigma)$ commute, $R'$ is a representation of $G'$. Let $S' = \{(g, \sigma) \in G' : R'(g, \sigma) \ket{\psi'} = \ket{\psi'}\}$ be the Bose symmetry subgroup of $\ket{\psi'}$ with respect to $R'$. We can obtain the hidden subgroup $S \leq G$ of $A$ is simply the projection of $S'$ onto $G$; in particular, we can obtain a set of generators for $S$ from a set of generators for $S'$ by taking the $G$-component of each generator.

    If $g \notin S$, then there is some $j \in [m]$ such that $R(g) \ket{\psi_j} \neq \ket{\psi_k}$ for all $k \in [m]$. In particular, $\abs{\braket{\psi_k | R(g) | \psi_j}} \leq 1 - \epsilon$ for all $k \in [m]$. This implies that
    \begin{equation*}
        \abs{\braket{\psi' | R'(g, \sigma) | \psi'}} = \abs{\braket{\psi_1 | R(g) | \psi_{\sigma^{-1} (1)}} \cdots \braket{\psi_m | R(g) | \psi_{\sigma^{-1}(m)}}} \leq 1 - \epsilon
    \end{equation*}
    for all $\sigma \in S_m$.
\end{proof}

\cref{thm:subset-bose-symmetry-learning-reduction-to-bose-symmetry-learning} shows that if the non-abelian group $G \times S_m$ admits an efficient quantum algorithm for \prb{BoseSL}, then $G$ admits an efficient quantum algorithm for \prb{SubsetBoseSL} for $A$ of size $m$. From \cref{thm:bose-symmetry-learning-over-poly-near-hamiltonian-groups}, we know that any poly-near Hamiltonian group admits an efficient quantum algorithm for \prb{BoseSL}. The following lemma gives a condition under which $G \times S_m$ is poly-near Hamiltonian, and hence \prb{SubsetBoseSL} is efficiently solvable.

Recall from \cref{def:poly-near-abelian-group} that a group $G$ is poly-near abelian if $[G: Z(G)] = O(\polylog \abs{G})$, where $Z(G)$ is the centre of $G$.

\begin{lemma}
    Suppose $G$ is a poly-near abelian group, and $m = O(\log \log \abs{G} / \log \log \log \abs{G})$. Then $G \times S_m$ is a poly-near Hamiltonian group.
\end{lemma}
\begin{proof}
    Recall from \cref{sec:bose-symmetry-learning-over-poly-near-hamiltonian-groups} that $G$ is poly-near Hamiltonian if $[G: \Baer(G)] = O(\polylog \abs{G})$, where $\Baer(G) = \{ g \in G: g K g^{-1} = K \text{ for all } K \leq G \}$ is the Baer norm of $G$. Note that $Z(\Gamma) \leq \Baer(\Gamma)$ for any group $\Gamma$, and $Z(\Gamma_1 \times \Gamma_2) = Z(\Gamma_1) \times Z(\Gamma_2)$ for any groups $\Gamma_1, \Gamma_2$. Hence, $Z(G) \times Z(S_m) = Z(G \times S_m) \leq \Baer(G \times S_m)$. This gives
    \begin{align*}
        [(G \times S_m) : \Baer(G \times S_m)] & = \frac{\abs{G} \cdot m!}{\abs{\Baer(G \times S_m)}} \leq m! \frac{\abs{G}}{\abs{Z(G \times S_m)}} \\
        & = m! \cdot [G: Z(G)] = m! \cdot O(\polylog \abs{G}) = m! \cdot O(\polylog \abs{G \times S_m})
    \end{align*}
    To see why the upper bound on $m$ is sufficient, note that we need to have $m! = O(\polylog \abs{G})$. So it suffices that $m \log m = O(\log \log \abs{G})$, which is satisfied by the given bound on $m$.
\end{proof}

\begin{corollary}
    \prb{SubsetBoseSL} is efficiently solvable for poly-near abelian groups $G$, provided that the number of states in the set is $m = O(\log \log \abs{G} / \log \log \log \abs{G})$, and quantum Fourier transforms over certain subgroups of $G \times S_m$ can be performed efficiently.
\end{corollary}

Note that we can define a subset \emph{anyonic} symmetry learning (\prb{SubsetAnyonicSL}) problem analogously to \prb{SubsetBoseSL} using the definition of anyonic symmetry from \cref{prb:relaxed-state-hsp}. As in the reductions from anyonic symmetry learning to Bose symmetry learning, if $G$ is abelian we can reduce the problem to \prb{SubsetBoseSL} using the endomorphisation of $R$ (see \cref{sec:relaxed-state-hsp}) if $R$ is a linear representation, or a self-orthogonal, zero-sum linear code (see \cref{sec:projective-state-hsp}) if $R$ is a projective representation. We can also combine the above ideas with those of \cref{sec:non-abelian-anyonic-symmetry-learning} to reduce \prb{SubsetAnyonicSL} over any group $G$ to \prb{BoseSL} over $G \times \ZZ_E \times S_m$.

\subsection{Learning symmetries of subspaces}

In \prb{AnyonicSL}, we want to learn the subgroup of all $g \in G$ such that $R(g) \ket{\psi} = e^{i \theta} \ket{\psi}$ for some $\theta \in [0, 2\pi)$, i.e. $R(g) \ket{\psi} \in V$, where $V$ is the one-dimensional subspace spanned by $\ket{\psi}$. However, what if we wanted to relax this symmetry condition by allowing $V$ to be a higher-dimensional subspace? This motivates the following \emph{subspace symmetry learning} problem.

\begin{problem}[\prb{SubspaceSL}]\label{prb:subspace-symmetry-learning}\leavevmode
    \begin{description}
        \item [Input] A subspace $V \subseteq \mathcal{H}$.
        \item [Promise] There is a finite group $G$, a projective unitary representation $R: G \to \mathcal{U}(\mathcal{H})$, and a subgroup $S \leq G$ such that $V$ is symmetric under the action of $S$, i.e. $R(g) V = V$ for all $g \in S$, and $V$ is not symmetric under the action of any $g \notin S$, i.e. $\abs{\braket{\psi | R(g) | \phi}} \leq 1 - \epsilon$ for all $g \notin S$, $\ket{\psi}, \ket{\phi} \in V$, for some $\epsilon > 0$.
        \item [Task] Learn the hidden symmetry subgroup $S$.
    \end{description}
\end{problem}

If the form of access to $V$ is maximally mixed states on $V$, then \prb{SubspaceSL} can be reduced to the following \emph{conjugate symmetry learning} problem:

\begin{problem}[\prb{ConjugateSL}]\label{prb:conjugate-symmetry-learning}\leavevmode
    \begin{description}
        \item [Input] Copies of a mixed state $\rho \in \densmats{\mathcal{H}}$.
        \item [Promise] There is a finite group $G$, a projective unitary representation $R: G \to \mathcal{U}(\mathcal{H})$, and a subgroup $S \leq G$ such that $\rho$ is symmetric under the action of $S$, i.e. $R(g) \rho R(g)^\dagger = \rho$ for all $g \in S$, and $\rho$ is not symmetric under the action of any $g \notin S$, i.e. $\norm{R(g) \rho R(g)^\dagger - \rho} \geq \epsilon$ for all $g \notin S$, for some specified norm $\norm{\cdot}$ and $\epsilon > 0$.
        \item [Task] Learn the hidden symmetry subgroup $S$.
    \end{description}
\end{problem}

To see the connection between these two problems, let $\rho_V = \frac{1}{k} \Pi_V$ be the maximally mixed state on $V$, where $\Pi_V$ is the projector onto $V$ and $k = \dim(V)$. Then since $\rho_V$ is a uniform mixture of all states in $V$, we have $R(g) \rho_V R(g)^\dagger = \rho_V$ for all $g \in S$. Now let $g \notin S$ and let $U = R(g)$. We want to show that $\norm{U \rho_V U^\dagger - \rho_V}_1 \geq \epsilon$, i.e. $\norm{A}_1 \geq k \epsilon$, where $A = \Pi_V - U \Pi_V U^\dagger$.

We have $\norm{\Pi_V U \ket{\phi}} = \max_{\ket{\psi} \in V} \abs{\braket{\psi | U | \phi}} \leq 1 - \epsilon$ for all $\ket{\phi} \in V$. Let $\ket{\phi} \in V$. Then $\braket{\phi | A | \phi} = 1 - \braket{\phi | U \Pi_V U^\dagger | \phi} = 1 - \norm{\Pi_V U \ket{\phi}}^2 \geq 1 - (1 - \epsilon)^2 = 2 \epsilon - \epsilon^2$. Let $\ket{\phi_1}, \dots, \ket{\phi_k}$ be an orthonormal basis for $V$, so that $\Pi_V = \sum_{j = 1}^k \ketbra{\phi_j}{\phi_j}$. Then
\begin{align*}
    \norm{A}_1 \norm{\Pi_V}_\infty & \geq \norm{\Pi_V A}_1 \geq \Tr(\Pi_V A) \\
    & = \sum_{j = 1}^k \braket{\phi_j | A | \phi_j} \geq k (2 \epsilon - \epsilon^2) \geq k \epsilon,
\end{align*}
and since $\norm{\Pi_V}_\infty = 1$, we have $\norm{A}_1 \geq k \epsilon$ as desired.

\section*{Acknowledgements} Subramanian acknowledges support from the Royal Society through a University Research Fellowship. 

\paragraph{AI Disclosure Statement} We used Gemini and ChatGPT to assist with the proofs of some technical lemmas: namely, \cref{lmm:probability-of-triv-set-under-weak-state-fourier-sampling,lmm:sum-of-squares-decomposition-conditions,lmm:sum-of-squares-decomposition-is-efficient,lmm:sum-of-squares-decomposition-with-zero-sum-condition}, as well as \cref{thm:hamiltonian-symmetry-learning-reduction-to-unitary-symmetry-learning}. The authors verified the correctness and originality of all content including references.

\printbibliography

\appendix

\section{Proofs of number-theoretic lemmas}\label{sec:appendix}

In this section, we prove the number-theoretic lemmas in \cref{subsec:optimising-rate-and-block-length}, which we restate here for convenience.

\subsection{Preliminaries}

\begin{fact}[Hensel's Lemma]\label{fct:hensels-lemma}
    Let $p$ be prime, let $f(x_1, \dots, x_m) \in \ZZ[x_1, \dots, x_m]$ be a polynomial in indeterminate $x_1, \dots, x_m$ with integer coefficients. Write $\partial_j f(x_1, \dots, x_m)$ for the formal derivative of $f$ with respect to $x_j$. Let $r = (r_1, \dots, r_m) \in \ZZ^m$ be such that $f(r) = 0 \pmod{p}$ and $\partial_j f(r) \neq 0 \pmod{p}$ for some $j$. Then for every $e \in \NN$, there exists $s = (s_1, \dots, s_m) \in \ZZ^m$ (which is unique modulo $p^e$) such that $f(s) = 0 \pmod{p^e}$ and $s_j = r_j \pmod{p}$ for all $j$.
\end{fact}

\begin{fact}[Legendre's Three-Square Theorem]\label{fct:legendre-three-square-theorem}
    For every natural number $n \in \NN$, $n$ can be expressed as a sum of three integer squares, $n = a^2 + b^2 + c^2$ for $a, b, c \in \ZZ$, if and only if $n$ is not of the form $4^k (8 m + 7)$ for $k, m \in \NN \cup \{0\}$.
\end{fact}

\cite{PS19threeSquareTheorem} gives an efficient ($O(\polylog n)$-time) algorithm for decomposing a valid $n \in \NN$ into a sum of three squares.

\begin{definition}[Quadratic Residue]
    A \emph{quadratic residue (QR)} modulo $k$ is an integer $a$ such that there exists an integer $x$ with $x^2 = a \pmod{k}$.
\end{definition}

\begin{definition}[Legendre Symbol]
    For a prime $p$ and integer $a \in \ZZ$, the \emph{Legendre symbol} is defined as
    \begin{equation*}
        \legendre{a}{p} \coloneq \begin{cases}
            1 & \text{if } a \text{ is a QR modulo } p \text{ and } a \neq 0 \pmod{p}, \\
            -1 & \text{if } a \text{ is not a QR modulo } p, \\
            0 & \text{if } a = 0 \pmod{p}.
        \end{cases}
    \end{equation*}
\end{definition}

\noindent The Legendre symbol satisfy several nice properties:

\begin{fact}\label{fct:legendre-symbol-multiplicativity}
    The Legendre symbol is completely multiplicative in its top argument: for any prime $p$ and integers $a, b \in \ZZ$, $\legendre{a b}{p} = \legendre{a}{p} \legendre{b}{p}$.
\end{fact}

\begin{fact}\label{fct:legendre-symbol-periodicity}
    The Legendre symbol is periodic in its top argument: $\legendre{a}{p} = \legendre{b}{p}$ if $a = b \pmod{p}$.
\end{fact}

\begin{fact}\label{fct:minus-1-quadratic-residue}
    $-1$ is a QR modulo a prime $p$, i.e. $\legendre{-1}{p} = 1$, if and only if $p \equiv 1 \pmod{4}$.
\end{fact}

\begin{fact}[Quadratic Reciprocity Law]\label{fct:quadratic-reciprocity-law}
    Let $p \neq q$ be primes. Then
    \begin{equation*}
        \legendre{p}{q} \legendre{q}{p} = (-1)^{\frac{p - 1}{2} \frac{q - 1}{2}}.
    \end{equation*}
\end{fact}

\begin{fact}[{Quadratic Character Sum -- \cite[Theorem 5.48]{LN97finiteFields}}]\label{fct:jacobsthal-sum}
    Let $p$ be an odd prime, let $q(x) = a_2 x^2 + a_1 x + a_0 \in \FF_p [x]$ be a quadratic polynomial over $\FF_p$. Let $d = a_1^2 - 4 a_2 a_0$ be the discriminant of $q$. Then if $d \neq 0$,
    \begin{equation*}
        \sum_{x \in \FF_p} \legendre{q(x)}{p} = -\legendre{a_2}{p}.
    \end{equation*}
\end{fact}

\noindent Both checking whether $a$ is a QR mod $p$, and finding a square root of a quadratic residue, can be performed efficiently (in time $\polylog p$):

\begin{fact}[Euler's Criterion]\label{fct:euler-criterion}
    Let $p$ be an odd prime. Then $a \in \ZZ$ is a QR mod $p$ if and only if $a^{(p-1)/2} \equiv 1 \pmod{p}$.

    Furthermore, this condition can be checked in time $O(\polylog p)$ using exponentiation by squaring.
\end{fact}

\begin{fact}[Tonelli-Shanks Algorithm \cite{Sha73numberTheoreticAlgorithms}]\label{fct:tonelli-shanks-algorithm}
    Let $p$ be an odd prime and let $a$ be a quadratic residue modulo $p$. There exists a (classical) algorithm which finds an integer $x$ such that $x^2 \equiv a \pmod{p}$ in time $O(\log^4 p)$.
\end{fact}

\subsection{Proofs of lemmas}

\begin{lemma}[Restatment of \cref{lmm:sum-of-squares-decomposition-conditions}]
    Let $k \in \NN$ with $k \geq 2$. Then:
    \begin{enumerate}
        \item $k \neq 0 \pmod{4}$ and $k$ contains no prime factor of the form $4 m + 3$ if and only if there exists $s_1 \in \ZZ$ such that $1 + s_1^2 = 0 \pmod{k}$.
        \item $k \neq 0 \pmod{4}$ if and only if there exist $s_1, s_2 \in \ZZ$ such that $1 + s_1^2 + s_2^2 = 0 \pmod{k}$.
        \item If $k \neq 0 \pmod{8}$ if and only if there exist $s_1, s_2, s_3, s_4 \in \ZZ$ such that $1 + s_1^2 + s_2^2 + s_3^2 = 0 \pmod{k}$.
        \item If $k = 0 \pmod{8}$, then there exist $s_1, s_2, s_3, s_4 \in \ZZ$ such that $1 + s_1^2 + s_2^2 + s_3^2 + s_4^2 = 0 \pmod{k}$.
    \end{enumerate}
\end{lemma}
\begin{proof}
    Let $k = 2^a \prod_{i = 1}^\ell p_i^{a_i}$ be the prime factorisation of $k$, where $p_i$ are distinct odd primes. In each case, by the Chinese remainder theorem, a solution exists if and only if a solution exists modulo $2^a$ and modulo $p_i^{a_i}$ for all $i$.
    \begin{enumerate}
        \item We have $s_1^2 = -1 \pmod{k}$, i.e. $-1$ is a quadratic residue modulo $k$. By the Chinese remainder theorem, this is equivalent to $-1$ being a quadratic residue (QR) modulo $2^a$ and modulo $p_i^{a_i}$ for all $i$. For $a \leq 2$, $-1$ is not a QR mod $2^a$, as otherwise there exists $x \in \ZZ$ such that $x^2 + 1 = 0 \pmod{2^a}$, but then $x^2 + 1 = 0 \pmod{4}$ i.e. $x^2 = 3 \pmod{4}$, which is impossible. Clearly, $-1$ is a QR mod $2^0$ and $2^1$.
        
        Let $f(x) = x^2 + 1 \in \ZZ[x]$, so $f'(x) = 2x$. For $p$ prime, suppose $s^2 = -1 \pmod{p}$, so $f(s) = 0 \pmod{p}$. We must have $s \neq 0 \pmod{p}$, hence $f'(s) \neq 0 \pmod{p}$. So by Hensel's lemma (\cref{fct:hensels-lemma}), there is exists $s_e$ such that $f(s_e) = 0 \pmod{p^e}$ for all $e \in \NN$. Also, a solution to $f(x) = 0 \pmod{p^e}$ for all $e$ clearly implies a solution to $f(x) = 0 \pmod{p}$. Hence, $-1$ is a QR mod $p^e$ for all $e \in \NN$ if and only if $-1$ is a QR mod $p$, and by \cref{fct:minus-1-quadratic-residue}, $-1$ is a QR mod $p$ if and only if $p = 1 \pmod{4}$. This completes the proof.
        \item We have $s_1^2 + s_2^2 = -1 \pmod{k}$. First, if $k = p$ is prime, then since there are $(p + 1) / 2$ quadratic residues modulo $p$, there must be a non-empty intersection of the sets $\{s_1^2: s_1 \in \ZZ_p\}$ and $\{-1 - s_2^2: s_2 \in \ZZ_p\}$, so a solution exists for $k$ prime. Let $f(x, y) = x_1^2 + x_2^2 + 1 \in \ZZ[x_1, x_2]$, so $\partial_1 f(x_1, x_2) = 2 x_1$ and $\partial_2 f(x_1, x_2) = 2 x_2$. For a solution $(s_1, s_2)$ to $f(x, y) = 0 \pmod{p}$, we must have $s_1 \neq 0 \pmod{p}$ or $s_2 \neq 0 \pmod{p}$, hence $\partial_1 f(s_1, s_2) \neq 0 \pmod{p}$ or $\partial_2 f(s_1, s_2) \neq 0 \pmod{p}$. So by Hensel's lemma, there exists a solution $(s_e, t_e)$ to $f(x, y) = 0 \pmod{p^e}$ for all $e \in \NN$.
        
        Clearly a solution exists modulo $2^a$ for $a = 0, 1$. For $a \geq 2$, we must have $s_1^2 + s_2^2 + 1 = 0 \pmod{4}$. But every integer square is congruent to $0$ or $1$ modulo $4$, so no solutions exist.
        \item We have $s_1^2 + s_2^2 + s_3^2 = -1 \pmod{k}$. From above, we know that a solution exists modulo $p^e$ for all odd primes $p$ and all $e \in \NN$ (by taking $s_3 = 0$). A solution exists modulo $2^0$ (any assignment is a solution), modulo $2^1$ (a solution is $(1, 0, 0)$), and modulo $2^2$ (a solution is $(1, 1, 1)$). For a solution modulo $2^a$ for $a \geq 3$, we must have $s_1^2 + s_2^2 + s_3^2 + 1 = 0 \pmod{8}$. But every integer square is congruent to $0, 1$ or $4$ modulo $8$, so no solutions exist.
        \item The existence of $s_1, s_2, s_3, s_4$ follows from Lagrange's four-square theorem (\cref{fct:lagrange-four-square-theorem}) for $k - 1$.
    \end{enumerate}
\end{proof}

\begin{lemma}[Restatment of \cref{lmm:sum-of-squares-decomposition-is-efficient}]
    For each of the 4 cases in \cref{lmm:sum-of-squares-decomposition-conditions}, the condition on $k$ can be checked in $O(\polylog k)$ time, and if the condition on $k$ holds, then a set of corresponding $s_i$ can be found in $O(\polylog k)$ time.
\end{lemma}
\begin{proof}
    Checking $k = 0 \pmod{4}$ or $k = 0 \pmod{8}$ can be done efficiently (in constant time when $k$ is represented in binary). In the first case, we want to find the factorisation of $k = 2^a \prod_{i = 1}^\ell p_i^{a_i}$ into primes; this can be done in $O(\polylog k)$ time using Shor's algorithm \cite{Sho97quantumFactoring}.
    \begin{enumerate}
        \item Checking whether $k$ contains a prime factor of the form $4m + 3$ can be done in $O(\ell) \leq O(\log k)$ time once we have found the prime factorisation of $k$. We may solve $x_i^2 = -1 \pmod{p_i}$ for each $i$ by the Tonelli-Shanks algorithm (\cref{fct:tonelli-shanks-algorithm}) in time $O(\log^4 p_i) \leq O(\polylog k)$. We can lift each solution $x_i$ to a solution $y_i$ to $y_i^2 = -1 \pmod{p_i^{a_i}}$ using Hensel's lemma (\cref{fct:hensels-lemma}) in time $O(a_i^3 \log^2 p_i) \leq O(\log^5 k)$. Finally, we can combine the solutions $y_i$ using the Chinese remainder theorem in time $O(\ell \log^2 k) \leq O(\log^3 k)$.
        \item The number of $b$ modulo $p$ such that $-1 - b^2$ is a QR mod $p$ is
        \begin{equation*}
            \frac{1}{2} \sum_{b \in \ZZ_p} \left(1 + \legendre{-1 - b^2}{p}\right) = \frac{1}{2} \left(p + \sum_{b \in \ZZ_p} \legendre{-1 - b^2}{p}\right) = \frac{1}{2} \left(p - (-1)^{(p - 1) / 2}\right) \geq (p - 1)/2,
        \end{equation*}
        where in the last equality we have used \cref{fct:jacobsthal-sum}. Hence, by picking uniformly random $b \in \ZZ_p$, and checking whether $-1 - b^2$ is a QR mod $p$ using \cref{fct:euler-criterion} (Euler's criterion), which occurs with probability at least $1 / 4$. Then using Tonelli-Shanks, we find a solution $x$ to $x^2 = -1 - b^2 \pmod{p}$. As in the previous case, we do this for all primes $p$, then lift the solutions to $p^a$ using Hensel's lemma, and combine the solutions using the Chinese remainder theorem.
        \item We can find a solution for each odd prime factor $p_i^{a_i}$ as in the previous case. For the $2^a$ factor, we can find a solution in constant time (e.g. $(1, 0, 0)$ for $a = 1$, $(1, 1, 1)$ for $a = 2$). We can then combine the solutions using the Chinese remainder theorem in time $O(\log^3 k)$.
        \item We use the (classical) algorithm in \cite{PT18lagrangeFourSquares} to find a decomposition of $k - 1$ into a sum of four squares in time $O(\polylog k)$.
    \end{enumerate}
\end{proof}

\begin{lemma}[Restatement of \cref{lmm:rate-block-length-tradeoff}]\leavevmode
    \begin{itemize}
        \item If $E \neq 0 \pmod{4}$ and $E$ contains no prime factor of the form $4 m + 3$, then the optimal sampling rate of $1 / 2$ can be achieved using a code with optimal block length of $2$. (Recall $2$ is the optimal block length regardless of sampling rate for this case.)
        \item If $E \neq 0 \pmod{4}$, then the optimal sampling rate of $1 / 2$ can be achieved using a code with block length of $4$, which is optimal. (Recall $3$ is the optimal block length regardless of sampling rate for this case.)
        \item If $E = 0 \pmod{8}$, the optimal sampling rate of $1 / 2$ is achieved by the code above with block length $8$, which is optimal. There is a code with block length $7$ and sampling rate $3 / 7$ which is optimal for this block length, a code with block length $6$ and sampling rate $1 / 3$ which is optimal for this block length. For block length $5$, the best sampling rate is $1 / 5$.
    \end{itemize}
    Moreover, all these codes can be found in time $O(\polylog E)$.
\end{lemma}
\begin{proof}\leavevmode
    \begin{itemize}
        \item If $E \neq 0 \pmod{4}$ and $E$ contains no prime factor of the form $4m + 3$, then the optimal sampling rate and block length are achieved by the same code generated by $(\begin{smallmatrix}
            1 & s
        \end{smallmatrix})$, where $s^2 = -1 \pmod{E}$.

        \item If $E \neq 0 \pmod{4}$, then by \cref{lmm:sum-of-squares-decomposition-conditions} and \cref{lmm:sum-of-squares-decomposition-is-efficient}, there exist $s_1, s_2$ such that $1 + s_1^2 + s_2^2 = 0 \pmod{E}$ which are efficiently (in time $\polylog(E)$) findable. Then taking $\vd{R} = \begin{pmatrix}
            s_1 & s_2 \\
            s_2 & -s_1
        \end{pmatrix}$, the matrix $M = (I \mid \vd{R})$ satisfies the conditions in \cref{lmm:characterisation-of-bell-sampling-representations} and \cref{lmm:characterisation-of-useful-bell-sampling-representations}, so generates a free self-orthogonal code with block length $4$ and rate $1 / 2$. Given sampling rate $1 / 2$, block length $4$ is optimal: by \cref{lmm:sum-of-squares-decomposition-conditions}, the block length must be at least $3$, and a better sampling rate than $1 / 3$ for block length $3$ is not possible, since it would require the generator matrix to have at least two rows, which would mean the sample rate was $\geq 2 / 3$, which violates \cref{lmm:upper-bound-on-sampling-rate}.

        \item If we remove the last row from the above matrix $M_8$ and remove from the resulting matrix the all-zeros column (here, the fourth column), we obtain a $3 \times 7$ matrix $M_7$. This generates a code with block length $7$ and rate $3 / 7$. This is the optimal rate for block length $7$: a dimension higher than $3$ would mean that the rate is $\geq 4 / 7 > 1 / 2$, which violates \cref{lmm:upper-bound-on-sampling-rate}.
        
        Again, we can remove the last row from $M_7$ and remove from the resulting matrix the all-zeros column (the third column), giving a $2 \times 6$ matrix $M_6$, which has a block length of $6$ and a sample rate of $2 / 6 = 1 / 3$. To see that this the optimal rate for block length $6$, suppose the rate was higher. Then the dimension must be at least $3$, and in fact equal to $3$, otherwise the rate would be $> 1 / 2$. But then since the generator matrix $M$ is WLOG in standard form, we must have $M = (I \mid \vd{R})$, where $\vd{R}$ is a $3 \times 3$ integer matrix such that the rows $\vd{r}_1, \vd{r}_2, \vd{r}_3$ of $\vd{R}$ satisfy $\vd{r}_i \cdot \vd{r}_j = -\delta_{ij} \pmod{E}$.

        This is clearly not possible if $E = 0 \pmod{8}$, since it would imply that there exist $s_1, s_2, s_3 \in \ZZ$ such that $1 + s_1^2 + s_2^2 + s_3^2 = 0 \pmod{E}$, which violates \cref{lmm:sum-of-squares-decomposition-conditions}. If $E = 4 \pmod{8}$, then we must have that all entries of $\vd{r}_1, \vd{r}_2, \vd{r}_3$ are odd, since each integer square is $0$ or $1$ modulo $4$. But then $\vd{r}_i \cdot \vd{r}_j$ is $3$ modulo $4$, which is a contradiction.

        The same argument shows that for $E = 0 \pmod{4}$, a block length $5$ code cannot have rate higher than $1 / 5$.
    \end{itemize}
\end{proof}

\begin{lemma}[Restatment of \cref{lmm:sum-of-squares-decomposition-with-zero-sum-condition}]
    Let $k \in \NN$ and $k \geq 2$.
    \begin{enumerate}
        \item $k = 2$ if and only if there exists $s_1 \in \ZZ$ such that $1 + s_1^2 = 0 \pmod{k}$ and $1 + s_1 = 0 \pmod{k}$.
        \item $k \neq 0 \pmod{4}$ and every prime factor of $k$ is either $3$ or of the form $3m + 1$ if and only there exist $s_1, s_2 \in \ZZ$ such that $1 + s_1^2 + s_2^2 = 0 \pmod{k}$ and $1 + s_1 + s_2 = 0 \pmod{k}$.
        \item $k \neq 0 \pmod{8}$ if and only if there exist $s_1, s_2, s_3 \in \ZZ$ such that $1 + s_1^2 + s_2^2 + s_3^2 = 0 \pmod{k}$ and $1 + s_1 + s_2 + s_3 = 0 \pmod{k}$.
        \item If $k = 0 \pmod{8}$, then there exist $s_1, s_2, s_3, s_4 \in \ZZ$ such that $1 + s_1^2 + s_2^2 + s_3^2 + s_4^2 = 0 \pmod{k}$ and $1 + s_1 + s_2 + s_3 + s_4 = 0 \pmod{k}$.
    \end{enumerate}
    Moreover, each decomposition can be found in expected time $O(\polylog(k))$.
\end{lemma}
\begin{proof}
    Let $k = 2^a \prod_{i = 1}^\ell p_i^{a_i}$ be the prime factorisation of $k$, where $p_i$ are distinct odd primes. As in the proof of \cref{lmm:sum-of-squares-decomposition-conditions}, we use the Chinese remainder theorem and Hensel's lemma to reduce to solving the equations modulo $2^a$ and modulo each $p_i$.
    \begin{enumerate}
        \item We have $a = -1 \pmod{k}$ so $2 = 0 \pmod{k}$.
        \item Any solution $(s_1, s_2)$ must satisfy $s_1 = -(1 + s_2) \pmod{k}$, so it suffices to find $s_2$ such that $f(s_2) \coloneqq 1 + s_2^2 + (1 + s_2)^2 = 2(s_2^2 + s_2 + 1) = 0 \pmod{k}$, i.e. $(2s_2 + 1)^2 = -3 \pmod{k}$. Taking $s_2 = 1$ gives a solution for $k = 2$ and $k = 3$. Since $s_2^2 + s_2 + 1$ is odd, there are no solutions for $k = 4$, so no solutions for any $2^a$ with $a \geq 2$. It is easy to check that $f(s_2) = 0$ has no solutions modulo $9$, so no solutions modulo $3^a$ for $a \geq 2$. Now for a prime $p > 3$, a solution exists iff $-3$ is a quadratic residue modulo $p$, i.e. $\legendre{-3}{p} = 1$. We can find such a $s_2$ using the Tonelli-Shanks algorithm. By the quadratic reciprocity law and multiplicativity of the Legendre symbol,
        \begin{equation*}
            \legendre{-3}{p} = \legendre{-1}{p} \legendre{3}{p} = (-1)^{(p - 1) / 2} \legendre{p}{3} (-1)^{(p - 1) / 2} = \legendre{p}{3}.
        \end{equation*}
        Since $\legendre{1}{3} = 1$ and $\legendre{2}{3} = -1$, we have $\legendre{p}{3} = 1$ if and only if $p = 1 \pmod{3}$. Since $f'(s_2) = 4s_2 + 2$, the solution $s_2$ satisfies $f'(s_2) = 0 = f(s_2) \pmod{p}$ if and only if $2s_2^2 + 2s_2 = 4s_2 \pmod{p}$, i.e. $s_2^2 = s_2 \pmod{p}$, i.e. $s_2 = 0$ or $s_2 = 1 \pmod{p}$. But $f(0) = 2$ and $f(1) = 6$ so $f(s_2) \neq 0 \pmod{p}$. Hence, $f'(s_2) \neq 0 \pmod{p}$, so by Hensel's lemma, a solution exists modulo $p^a$ for all $a \in \NN$. As usual, we can combine the solutions modulo $p_i^{a_i}$ for all $i$ using the Chinese remainder theorem.
        \item Any solution $(s_1, s_2, s_3)$ must satisfy $s_1 = -(1 + s_2 + s_3) \pmod{k}$, so it suffices to find $s_2, s_3$ such that $f(s_2, s_3) \coloneqq 1 + s_2^2 + s_3^2 + (1 + s_2 + s_3)^2 = 0 \pmod{k}$. Expanding $f$ gives
        \begin{align*}
            f(x, y) & = 2 x^2 + 2 y^2 + 2 x y + 2 x + 2 y + 2 \eqcolon 2 g(x, y) \\
            & 
        \end{align*}
        Taking $x = y = 1$ gives a solution for $k = 2$, $k = 3$ and $k = 2^2$. There are no solutions for $k = 2^3$ (and so for any $2^a$ with $a \geq 3$): any solution $(s_2, s_3)$ must satisfy $g(s_2, s_3) = 0 \pmod{4}$, but it can easily be checked that no solutions for this exist. For prime $p > 3$, let $t = 3^{-1}$ denote the multiplicative inverse of $3$ modulo $p$. Take $u = x + t$ and $v = y + t$, then we want to solve the quadratic equation $g(u - t, v - t) = u^2 + u v + v^2 + 2 t \pmod{p}$.
        
        For fixed $v \in \ZZ_p$, let $\Delta_v = v^2 - 4 (v^2 + 2 t) = -3 v^2 - 8 t$ be the discriminant of the quadratic equation. If $\Delta_v = r^2$ is a quadratic residue modulo $p$ (which can be checked in $O(\polylog p)$ time by Euler's criterion), then a solution $u$ exists: take $u = (-v + r) \cdot 2^{-1}$. Moreover, $u$ can be found in time $O(\polylog p)$ by using Tonelli-Shanks to find $r$. The number of $v \in \ZZ_p$ such that $\Delta_v$ is a quadratic residue modulo $p$ is
        \begin{equation*}
            \frac{1}{2} \sum_{v \in \ZZ_p} \left(1 + \legendre{\Delta_v}{p}\right) = \frac{1}{2} \left(p + \sum_{v \in \ZZ_p} \legendre{-3 v^2 - 8 t}{p}\right) = \frac{1}{2} \left(p - \legendre{-3}{p}\right) \geq (p - 1)/2,
        \end{equation*}
        where the third equality is by \cref{fct:jacobsthal-sum}.
        
        Hence, by choosing $v \in \ZZ_p$ uniformly at random, a solution exists with probability at least $1 / 4$, so we repeat this process a constant number of times to find a solution with high probability. As in the previous case, we can lift the solution to $p^a$ using Hensel's lemma, and combine the solutions using the Chinese remainder theorem.
        \item Any solution $(a, b, c, d)$ must satisfy $a = -(1 + b + c + d) \pmod{k}$, so it suffices to find $b, c, d$ such that $f(b, c, d) \coloneqq 1 + b^2 + c^2 + d^2 + (1 + b + c + d)^2 = 0 \pmod{k}$. Expanding $f$ gives
        \begin{align*}
            f(b, c, d) & = 2 b^2 + 2 c^2 + 2 d^2 + 2 b c + 2 b d + 2 c d + 2 b + 2 c + 2 d + 2 \\
            & = \frac{1}{4} \big( (2b + 2c + 1)^2 + (2b + 2d + 1)^2 + (2c + 2d + 1)^2 + 5 \big) \\
            & = \frac{1}{4} (X^2 + Y^2 + Z^2 + 5),
        \end{align*}
        where $X = 2b + 2c + 1$, $Y = 2b + 2d + 1$ and $Z = 2c + 2d + 1$. So it suffices to find $X, Y, Z \in \ZZ$ such that $X^2 + Y^2 + Z^2 = 8k - 5$. $8k - 5$ is not of the form $4^l (8 m + 7)$, so by Legendre's three-square theorem, such $X, Y, Z \in \ZZ$ do indeed exist. They can be found in expected time $O(\polylog k)$ by the algorithm in \cite{PS19threeSquareTheorem}. Moreover, we have $X^2 + Y^2 + Z^2 = 3 \pmod{8}$, which means $X, Y, Z$ are all odd (since each integer square is $0, 1$ or $4 \pmod{8}$). Also, without loss of generality (by potentially negating some of $X, Y, Z$), we may assume that $X = Y = Z = 1 \pmod{4}$. So write $X = 4 x + 1$, $Y = 4 y + 1$ and $Z = 4 z + 1$ for $x, y, z \in \ZZ$.
    \end{enumerate}
\end{proof}


\end{document}